\documentclass[aps,twocolumn,superscriptaddress]{revtex4-1}%
\usepackage{amsfonts}
\usepackage{amsmath}
\usepackage{amssymb}
\usepackage{mathtools}
\usepackage[colorlinks=true,linkcolor=blue,citecolor=red,plainpages=false,pdfpagelabels]%
{hyperref}
\usepackage{graphicx}
\usepackage[section]{placeins}%
\providecommand{\U}[1]{\protect\rule{.1in}{.1in}}
\newtheorem{theorem}{Theorem}

\newtheorem{proposition}[theorem]{Proposition}
\newtheorem{remark}[theorem]{Remark}

\newenvironment{proof}[1][Proof]{\noindent\textbf{#1.} }{\ \rule{0.5em}{0.5em}}
\allowdisplaybreaks
\begin{document}
\preprint{ }
\title[Quantifying the performance of bidirectional quantum teleportation]{Quantifying the performance of bidirectional quantum teleportation}
\author{Aliza U.~Siddiqui}
\affiliation{Division of Computer Science and Engineering, Louisiana State University,
Baton Rouge, Louisiana 70803, USA}
\affiliation{Hearne Institute for Theoretical Physics, Department of Physics and Astronomy,
and Center for Computation and Technology, Louisiana State University, Baton
Rouge, Louisiana 70803, USA}
\author{Mark M.~Wilde}
\affiliation{Hearne Institute for Theoretical Physics, Department of Physics and Astronomy,
and Center for Computation and Technology, Louisiana State University, Baton
Rouge, Louisiana 70803, USA}
\affiliation{Stanford Institute for Theoretical Physics, Stanford University, Stanford,
California 94305, USA}

\begin{abstract}
Bidirectional teleportation is a fundamental protocol for exchanging quantum
information between two parties by means of a shared resource state and local
operations and classical communication (LOCC). In this paper, we develop two
seemingly different ways of quantifying the simulation error of unideal
bidirectional teleportation by means of the normalized diamond distance and
the channel infidelity, and we prove that they are equivalent. By relaxing the
set of operations allowed from LOCC\ to those that completely preserve the
positivity of the partial transpose, we obtain semi-definite programming lower
bounds on the simulation error of unideal bidirectional teleportation. We
evaluate these bounds for several key examples: when there is no resource
state at all and for isotropic and Werner states, in each case finding an
analytical solution. The first aforementioned example establishes a benchmark
for classical versus quantum bidirectional teleportation. Another example
consists of a resource state resulting from the action of a generalized
amplitude damping channel on two Bell states, for which we find an analytical
expression for the simulation error that is in agreement with numerical
estimates (up to numerical precision). We then evaluate the performance of
some schemes for bidirectional teleportation due to [Kiktenko \textit{et al}.,
Phys.~Rev.~A \textbf{93}, 062305 (2016)] and find that they are suboptimal and
do not go beyond the aforementioned classical limit for bidirectional
teleportation. We offer a scheme alternative to theirs that is provably
optimal. Finally, we generalize the whole development to the setting of
bidirectional controlled teleportation, in which there is an additional
assisting party who helps with the exchange of quantum information, and we
establish semi-definite programming lower bounds on the simulation error for
this task. More generally, we provide semi-definite programming lower bounds
on the performance of bipartite and multipartite channel simulation using a
shared resource state and LOCC.

\end{abstract}
\date{\today}
\maketitle
\tableofcontents

\section{Introduction}

Quantum teleportation is one of the most remarkable protocols in quantum
information \cite{BBC+93}, and it is one of the earliest examples of the
fascinating possibilities that local operations and classical communication
processing \cite{BDSW96}\ on entangled states offers. Due to the fragile
nature of qubits, the protocol was established as an alternative to direct
transmission of a quantum state between two parties, as depicted in
Figure~\ref{fig: Ideal QT}. Teleportation is now used routinely as a basic
primitive in quantum information science, with applications in quantum
communication, quantum error correction, quantum networking, etc. See
Figure~\ref{fig: Quantum Teleportation Circuit} for a depiction of the
teleportation protocol.

Teleportation has been extended in various ways, and one interesting way of
doing so in a basic quantum network is via the method of bidirectional
teleportation. In the ideal version of this protocol, Alice and Bob share two
ebits of entanglement and teleport qubits to each other in opposite
directions. The ideal protocol realizes a perfect swap channel, as shown in
Figure~\ref{fig: Ideal BQT Channel}. This possibility was observed early on
\cite{V94}, and it was subsequently considered in
\cite{PhysRevA.63.042303,PhysRevA.65.042316}. More recently, there has been a
flurry of research on the topic, with various proposals for bidirectional
teleportation \cite{Fu2014,Hassanpour2016}. There has been even more interest
recently in a variation called bidirectional controlled teleportation, using
five-qubit \cite{ZZQS13,Shukla2013,Li2013a,Li2013,Chen2014}, six-qubit
\cite{Yan2013,Sun2013,Duan2014,Li2016a,ZXL19}, seven-qubit
\cite{Duan2014a,Hong2016,Sang2016}, eight-qubit
\cite{Zhang2015,SadeghiZadeh2017}, and nine-qubit\ \cite{Li2016}\ entangled
resource states (see also \cite{Thapliyal2015}). Bidirectional controlled
teleportation is a three-party protocol in which three parties, typically
called Alice, Bob, and Charlie, share an entangled resource state, and they
use local operations and classical communication to exchange qubits between
Alice and Bob. See also \cite{Gou2017} for other variations of bidirectional teleportation.

The applications of bidirectional teleportation align with those of standard,
unidirectional teleportation, but they apply in a basic quantum network
setting in which two parties would like to exchange quantum information.
Although the ideal version of bidirectional teleportation is manifestly a
trivial extension of the original protocol in which it is simply conducted
twice (but in opposite directions), the situation becomes less trivial and
more relevant to experimental practice when the resource state shared by the
two parties deviates from the ideal resource of two maximally entangled
states. Indeed, much of the prior work cited above focuses on precisely this
kind of case, when the resource state is different from two maximally
entangled states, either by being a different pure state, a mixed state, or a
state with insufficient entanglement to accomplish the task. These kinds of
investigations are important for understanding ways to simulate the ideal
protocol approximately in an experimental setting. More generally, a way of
framing bidirectional teleportation is that Alice and Bob share a resource
state $\rho_{AB}$ and they are allowed local operations and classical
communication (LOCC) for free, with the goal of simulating an ideal swap
channel, as depicted in Figure~\ref{fig: Unideal BQT LOCC}.

\begin{figure}[ptb]
\centering
\includegraphics[scale=0.5]{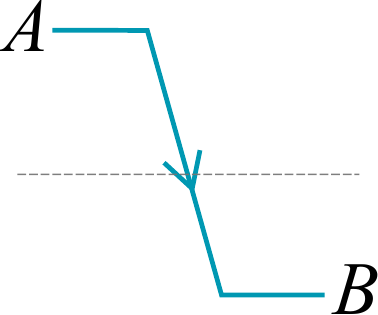}\caption{Ideal
unidirectional quantum channel from one party, Alice, to another party Bob.}%
\label{fig: Ideal QT}%
\end{figure}

\begin{figure}[ptb]
\begin{center}
\includegraphics[scale=0.30]{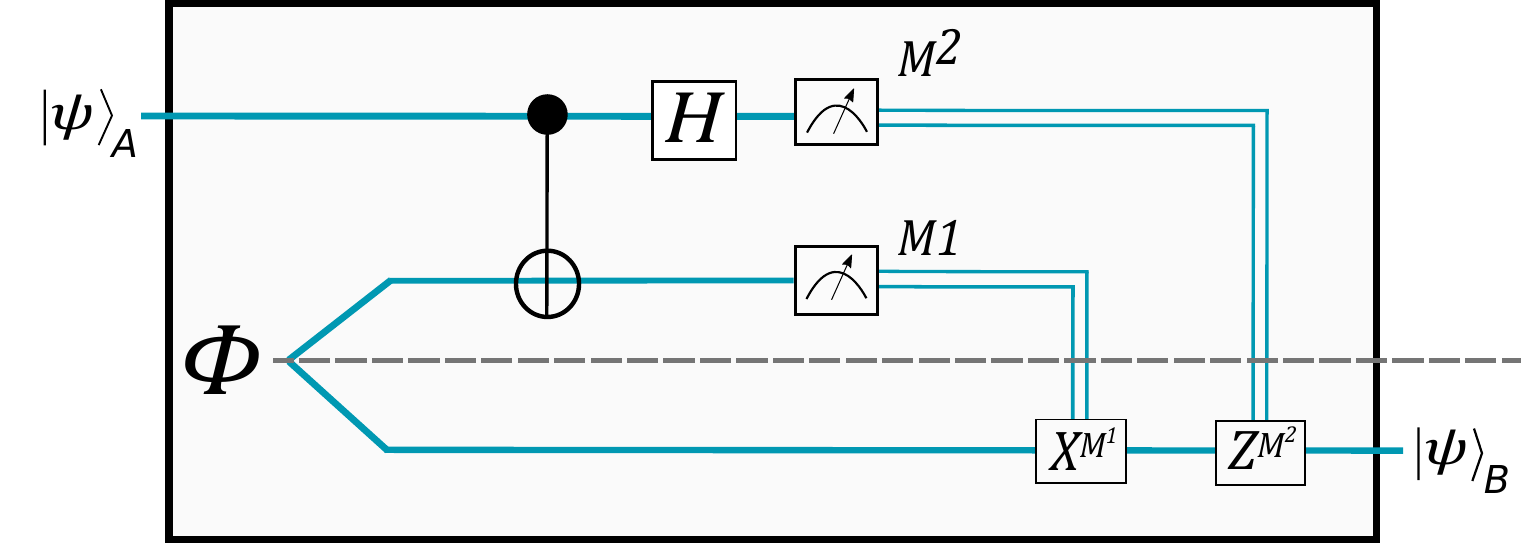}
\end{center}
\caption{Due to the fragile nature of quantum bits, the unidirectional quantum
teleportation protocol was devised as a method for simulating an ideal
unidirectional quantum channel, i.e., to transmit quantum information from one
party Alice, to another party Bob.}%
\label{fig: Quantum Teleportation Circuit}%
\end{figure}

In spite of the many works listed above on the topic of bidirectional
teleportation, what appears to be missing is a systematic method for
quantifying its performance in the case that it does not operate perfectly.
There is a need for this because any experimental implementation of
bidirectional teleportation will necessarily be imperfect. Indeed, entangled
states generated in experimental protocols such as spontaneous parametric
down-conversion are only approximations to ideal maximally entangled states
\cite{Cout18}. Our aim here is to fill this void.

The contributions of our paper are as follows:

\begin{enumerate}
\item After reviewing ideal bidirectional teleportation in
Section~\ref{sec:ideal-BQT}\ and recognizing that it implements a unitary swap
channel between Alice and Bob (see Figure~\ref{fig: SWAP Gate}), we define two
seemingly different ways to quantify the performance of unideal bidirectional
teleportation by means of the normalized diamond distance and the channel
infidelity (Sections~\ref{sec:q-error-NDD} and \ref{sec:q-error-CIF}). We
provide definitions of channel infidelity and diamond distance here, but they
can also be found in Sections~3.5.2 and 3.5.3 of \cite{KW20book},
respectively. Even though these performance measures are generally different,
we prove that they lead to the same values for the case of simulating the swap
channel (Section~\ref{sec:equality-sim-errors-LOCC-BQT}). More generally, we
also discuss how these measures can be employed for quantifying the
performance of bipartite channel simulation by means of a shared quantum state
and local operations and classical communication (LOCC)
(Section~\ref{sec:LOCC-sim-gen-bi}).

\begin{figure}[ptb]
\par
\begin{center}
\includegraphics[scale=0.40]{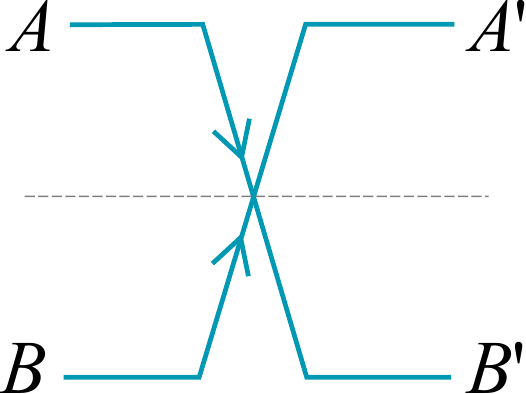}
\end{center}
\caption{Ideal swap channel between two parties, Alice and Bob, realized by
ideal bidirectional teleportation. This is a two-party generalization of the
ideal unidirectional channel depicted in Figure~\ref{fig: Ideal QT}.}%
\label{fig: Ideal BQT Channel}%
\end{figure}

\begin{figure}[ptb]
\par
\begin{center}
\includegraphics[scale=0.40]{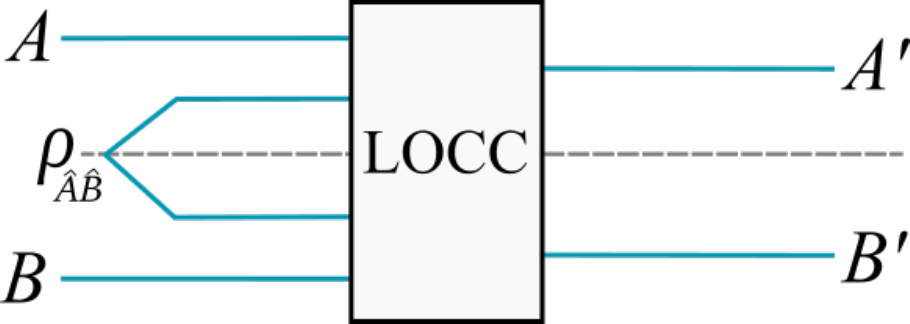}
\end{center}
\caption{The figure depicts a general framework for understanding the
simulation of bipartite quantum channels, realized by combining an LOCC
protocol and a quantum resource state $\rho_{\hat{A}\hat{B}}$. In experimental
implementations, the resource state $\rho_{\hat{A}\hat{B}}$ is imperfect. An
example of a target bipartite channel to simulate is the ideal bidirectional
teleportation, which is equivalent to a swap channel, as depicted in
Figure~\ref{fig: Ideal BQT Channel}.}%
\label{fig: Unideal BQT LOCC}%
\end{figure}

\item Optimizing these performance measures is in general challenging because
such optimizations are conducted over the set of LOCC channels, and it is
known that optimizing over LOCC channels is difficult. We then relax the
optimization problem such that it is conducted over the larger set of channels
that completely preserve the positivity of the partial transpose (C-PPT-P
channels), as shown in Figure~\ref{fig:Relax LOCC}
(Section~\ref{sec:SDP-lower-bound-general}). For both error measures
(normalized diamond distance and channel infidelity), we show how the relaxed
optimization problems can be evaluated by means of semi-definite programs
(SDPs). See Sections~\ref{sec:SDP-lower-bound-general}\ and
\ref{sec:SDP-ch-infid}. SDPs are optimization problems in which the cost or
objective function is linear, along with constraints and optimization
variables that are semi-definite. This optimization technique is widely used
in quantum information theory due to the semi-definite constraints that apply
to the basic constituents of quantum mechanics (including states and
channels). More information on semi-definite programs and C-PPT-P channels can
be found in Sections~2.4 and 3.2.12 of \cite{KW20book}, respectively.

\begin{figure}[ptb]
\begin{center}
\includegraphics[scale=0.40]{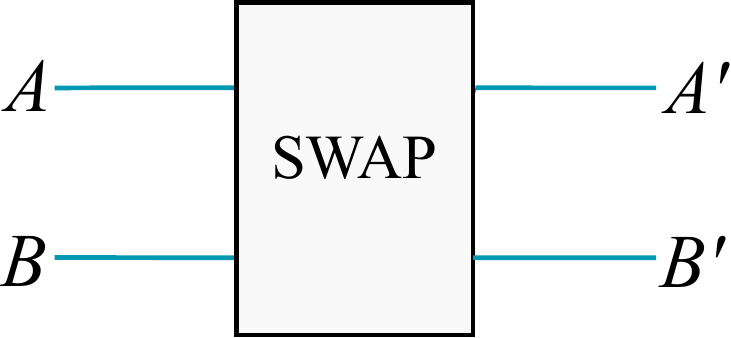}
\end{center}
\caption{Ideal bidirectional quantum teleportation realizes a perfect SWAP
channel between two parties. This swap channel is the same as depicted in
Figure~\ref{fig: Ideal BQT Channel}, but throughout the paper, we often think
of the swap operation as a bipartite channel and depict it as shown above.}%
\label{fig: SWAP Gate}%
\end{figure}

\begin{figure}[ptb]
\begin{center}
\includegraphics[scale=0.30]{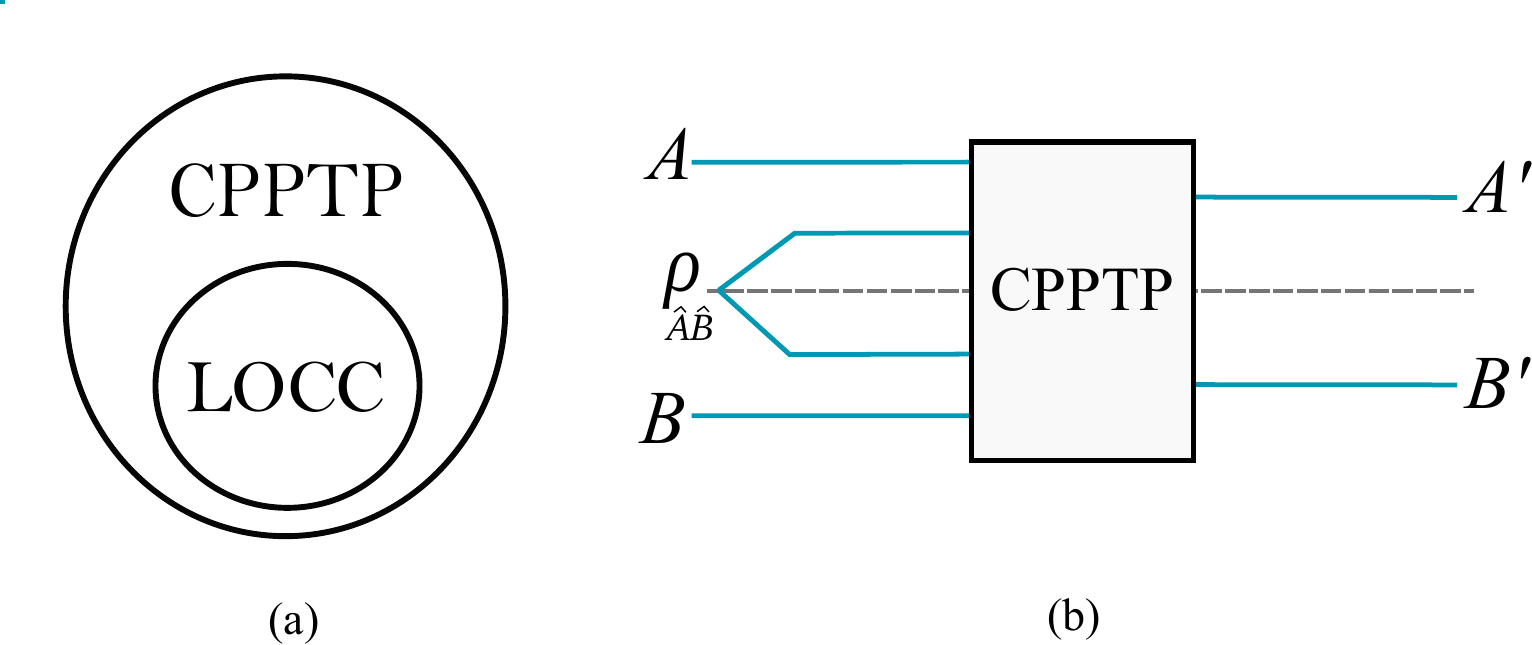}
\end{center}
\caption{Optimizing over all LOCC channels is known to be computationally
intensive. (a) We utilize the fact that LOCC channels are a subset of channels
that completely preserve the positivity of the partial transpose (C-PPT-P).
(b) Instead of optimizing over all LOCC channels and the ensuing protocols
depicted in Figure~\ref{fig: Unideal BQT LOCC}, we relax the optimization to
the larger set of C-PPT-P channels. Conducting the optimization problem over
this larger set can be solved in time polynomial in the dimension of the
resource state and the swap channel to be simulated.}%
\label{fig:Relax LOCC}%
\end{figure}


\item For the specific case of bidirectional teleportation, we show how
symmetries of the unitary swap channel, some of which are depicted in
Figure~\ref{fig: Symmetry of SWAP}, lead to a much simpler semi-definite
program for quantifying performance (Section~\ref{sec:SDP-lower-bound-BTP}).
The resulting semi-definite program has significantly reduced complexity that
is polynomial in the dimension of the resource state being used for
bidirectional teleportation. We also prove that the error measures based on
normalized diamond distance and channel infidelity coincide in this case
(Section~\ref{sec:SDP-ch-infid}).

\begin{figure}[ptb]
\begin{center}
\includegraphics[scale=0.30]{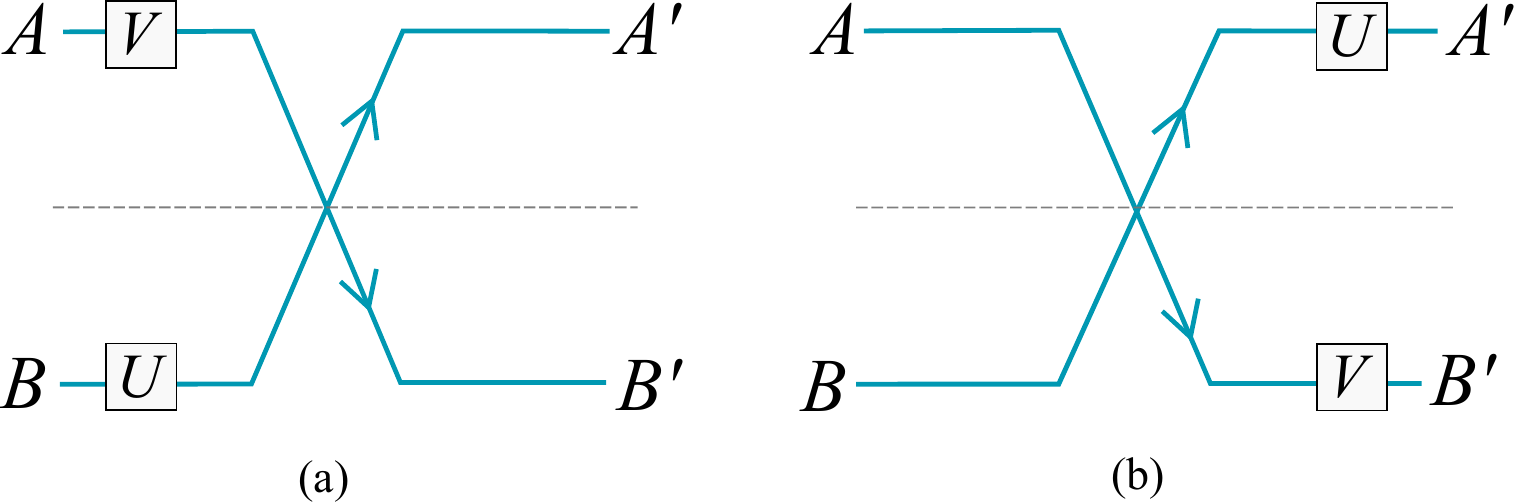}
\end{center}
\caption{Some symmetries of the SWAP channel are depicted in this figure. By
exploiting this unique property, (a) and (b) above are equivalent and the
optimization problem for quantifying the performance of unideal bidirectional
teleportation can be greatly simplified. (a) Alice and Bob perform unitary
operations $V$ and $U$, respectively, and then a SWAP operation to exchange
information. (b) Alice and Bob first perform a SWAP operation and then unitary
operations $U$ and $V$, respectively, on their individual qubits.}%
\label{fig: Symmetry of SWAP}%
\end{figure}



\item We then consider some specific examples of resource states for
bidirectional teleportation (Section~\ref{sec:examples}). These include the
case when no resource state is available, as well as isotropic and Werner
states (Sections~\ref{sec:no-res-state}, \ref{sec:isotropic-example}, and
\ref{sec:werner-example}, respectively). We also consider a resource state
resulting from the action of a generalized amplitude damping channel on two
Bell states (Section~\ref{sec:gadc}). In the first three cases, we reduce the
semi-definite program for quantifying performance to a linear program, which
we then solve analytically. The case when no resource state is available
provides a benchmark for classical versus quantum bidirectional teleportation,
and it is thus important for assessing any experimental implementation of
bidirectional teleportation. Specifically, we prove that, when no resource
state is available, the simulation error when simulating bidirectional
teleportation of $d$-dimensional systems cannot be smaller than $1-\frac
{1}{d^{2}}$. For some particular values of the parameters for isotropic and
Werner states, we also prove that the simulation error when using LOCC
channels is equal to the simulation error when using channels that completely
preserve the positivity of the partial transpose. For the generalized
amplitude damping channel example, we find an analytical expression for the
simulation error, which is correct up to numerical precision (thus we expect
there to be an analytical proof).

\item Next, in Section~\ref{sec:KPF16}, we use our previous results to assess
the performance of previous proposals for bidirectional teleportation from
\cite{KPF16}, which are for the case when the resource state available is a
single ebit, instead of the required two ebits that are necessary for a
perfect implementation of bidirectional teleportation. We find that the
proposals from \cite{KPF16} are suboptimal and do not go beyond the classical
limit for bidirectional teleportation. We also provide a simple protocol that
is provably optimal.

\item We finally generalize the whole development to quantify the performance
of multipartite channel simulation, when using LOCC channels and a shared
resource state for channel simulation (Section~\ref{sec:MP-gen}). We then
analyze the specific case of bidirectional controlled teleportation and
provide a semi-definite program that can be used to assess the performance of
this latter protocol.
\end{enumerate}

We note here that our general approach is similar in spirit to the approach
taken in \cite{IP05,II15}, but the applications we consider here are
different. We begin in the next section with some preliminary material and set
some notation. The rest of our paper proceeds in the order presented above,
and we finally conclude in Section~\ref{sec:conclusion}\ with a summary and a
list of open questions for future work.

\section{Preliminaries}

\label{sec:prelim}In this section, we review some preliminary concepts in
quantum information and set some notation used throughout the rest of our
paper. More background on quantum information is available in
\cite{H17,H13book,Watrous2018,W17,KW20book}.

A quantum state is a positive semi-definite operator with trace equal to one.
A bipartite quantum state $\rho_{AB}$ acts on a tensor-product Hilbert space
$\mathcal{H}_{A}\otimes\mathcal{H}_{B}$, and we denote the dimension of system
$A$ by $d_{A}$ and that of $B$ by $d_{B}$. A bipartite state $\rho_{AB}$ is
entangled if it cannot be written in the following form:
\begin{equation}
\sum_{x} p(x) \sigma^{x}_{A} \otimes\tau^{x}_{B},
\end{equation}
where $\{p(x)\}_{x}$ is a probability distribution and $\{\sigma^{x}_{A}%
\}_{x}$ and $\{\tau^{x}_{B}\}_{x}$ are sets of states. For many discussions
and applications of entanglement, see the review \cite{Horodecki2009a}.

A quantum channel is a completely positive and trace-preserving map. We denote
the unnormalized maximally entangled operator by%
\begin{align}
\Gamma_{RC}  &  \coloneqq|\Gamma\rangle\!\langle\Gamma|_{RC},\\
|\Gamma\rangle_{RC}  &  \coloneqq\sum_{i=0}^{d-1}|i\rangle_{R}|i\rangle_{C},
\end{align}
where $R\simeq C$ with dimension $d$ and $\{|i\rangle_{R}\}_{i=0}^{d-1}$ and
$\{|i\rangle_{C}\}_{i=0}^{d-1}$ are orthonormal bases. The notation $R\simeq
C$\ means that the systems $R$ and $C$ are isomorphic. The maximally entangled
state is denoted by%
\begin{equation}
\Phi_{RC}\coloneqq\frac{1}{d}\Gamma_{RC},
\end{equation}
and the maximally mixed state by%
\begin{equation}
\pi_{C}\coloneqq\frac{1}{d}I_{C}.
\end{equation}
The Choi operator of a quantum channel $\mathcal{N}_{C\rightarrow D}$ (and
more generally a linear map) is defined as%
\begin{equation}
\Gamma_{RD}^{\mathcal{N}}\coloneqq\mathcal{N}_{C\rightarrow D}(\Gamma_{RC}).
\end{equation}
In the notation above, there is an implicit identity channel acting on the
reference system $R$, so that%
\begin{equation}
\mathcal{N}_{C\rightarrow D}(\Gamma_{RC})\equiv(\operatorname{id}_{R}%
\otimes\mathcal{N}_{C\rightarrow D})(\Gamma_{RC}),
\end{equation}
and we employ this convention throughout our paper. A linear map
$\mathcal{M}_{C\rightarrow D}$ is completely positive if and only if its Choi
operator $\Gamma_{RD}^{\mathcal{M}}$ is positive semi-definite, and
$\mathcal{M}_{C\rightarrow D}$ is trace preserving if and only if its Choi
operator satisfies $\operatorname{Tr}_{D}[\Gamma_{RD}^{\mathcal{M}}]=I_{R}$.

An LOCC channel $\mathcal{L}_{A^{\prime}B^{\prime}\rightarrow AB}$ is a
bipartite channel that can be written in the following form:%
\begin{equation}
\mathcal{L}_{A^{\prime}B^{\prime}\rightarrow AB}=\sum_{y}\mathcal{E}%
_{A^{\prime}\rightarrow A}^{y}\otimes\mathcal{F}_{B^{\prime}\rightarrow B}%
^{y}, \label{eq:LOCC-channel-decomp}%
\end{equation}
where $\{\mathcal{E}_{A^{\prime}\rightarrow A}^{y}\}_{y}$ and $\{\mathcal{F}%
_{B^{\prime}\rightarrow B}^{y}\}_{y}$ are sets of completely positive,
trace-non-increasing maps, such that the sum map $\sum_{y}\mathcal{E}%
_{A^{\prime}\rightarrow A}^{y}\otimes\mathcal{F}_{B^{\prime}\rightarrow B}%
^{y}$ is a quantum channel (completely positive and trace preserving)
\cite{CLM+14}. However, not every channel of the form in
\eqref{eq:LOCC-channel-decomp} is an LOCC channel (there are separable
channels of the form in \eqref{eq:LOCC-channel-decomp} that are not
implementable by LOCC \cite{PhysRevA.59.1070}).

We make extensive use of the following bilateral unitary twirl channel in our
paper:%
\begin{equation}
\widetilde{\mathcal{T}}_{CD}(X_{CD})\coloneqq \int dU\ (\mathcal{U}_{C}%
\otimes\overline{\mathcal{U}}_{D})(X_{CD}), \label{eq:bilat-twirl}%
\end{equation}
where $\mathcal{U}(\cdot)=U(\cdot)U^{\dag}$, $\overline{\mathcal{U}}%
(\cdot)=\overline{U}(\cdot)U^{T}$, the overline indicates the complex
conjugate, and $dU$ denotes the Haar measure (uniform distribution on unitary
operators). This channel is an LOCC channel, in the sense that Alice can pick
a unitary at random according to the Haar measure, apply it to her system,
report to Bob which one she selected, who can then apply the complex conjugate
unitary to his system. In order to make this process feasible in practice,
note that the channel in \eqref{eq:bilat-twirl} can be simulated by a unitary
two-design \cite{DM14}, in which only a finite amount of classical data is
required to communicate to Bob when implementing \eqref{eq:bilat-twirl} via
LOCC. The following identity from \cite{W89,Horodecki99,Watrous2018}
simplifies the calculation of the action of $\widetilde{\mathcal{T}}_{CD}$ on
an arbitrary input operator $X_{CD}$:%
\begin{multline}
\widetilde{\mathcal{T}}_{CD}(X_{CD})=\Phi_{CD}\operatorname{Tr}_{CD}[\Phi
_{CD}X_{CD}]\label{eq:simple-formula-bilat-twirl}\\
+\frac{I_{CD}-\Phi_{CD}}{d^{2}-1}\operatorname{Tr}_{CD}[(I_{CD}-\Phi
_{CD})X_{CD}].
\end{multline}

We denote the transpose map acting on the quantum system $C$ by%
\begin{equation}
T_{C}(\cdot)\coloneqq\sum_{i,j=0}^{d-1}|i\rangle\!\langle j|_{C}%
(\cdot)|i\rangle\!\langle j|_{C}.
\end{equation}
A state $\rho_{CD}$ is a positive partial transpose (PPT) state if $T_{D}%
(\rho_{CD})$ is positive semi-definite. The partial transpose is its own
adjoint, in the sense that%
\begin{equation}
\operatorname{Tr}[Y_{CD}T_{C}(X_{CD})]=\operatorname{Tr}[T_{C}(Y_{CD})X_{CD}]
\end{equation}
for all linear operators $X_{CD}$ and $Y_{CD}$.

The following post-selected teleportation identity \cite{B05}\ plays a role in
our analysis:%
\begin{equation}
\mathcal{N}_{C\rightarrow D}(\rho_{SC})=\langle\Gamma|_{CR}\rho_{SC}%
\otimes\Gamma_{RD}^{\mathcal{N}}|\Gamma\rangle_{CR}.
\label{eq:post-selected-TP-id}%
\end{equation}
We also make frequent use of the identities%
\begin{align}
\operatorname{Tr}_{C}[X_{CD}]  &  =\langle\Gamma|_{RC}(I_{R}\otimes
X_{CD})|\Gamma\rangle_{RC},\\
X_{CD}|\Gamma\rangle_{CR}  &  =T_{R}(X_{RD})|\Gamma\rangle_{CR}.
\end{align}
Given channels $\mathcal{N}_{C\rightarrow D}$ and $\mathcal{M}_{D\rightarrow
E}$, the Choi operator $\Gamma_{RE}^{\mathcal{M}\circ\mathcal{N}}$ of the
serial composition $\mathcal{M}_{D\rightarrow E}\circ\mathcal{N}_{C\rightarrow
D}$ is given by%
\begin{align}
\Gamma_{RE}^{\mathcal{M}\circ\mathcal{N}}  &  =\langle\Gamma|_{DS}\Gamma
_{RD}^{\mathcal{N}}\otimes\Gamma_{SE}^{\mathcal{M}}|\Gamma\rangle
_{DS}\label{eq:serial-compose-Choi}\\
&  =\operatorname{Tr}_{D}[\Gamma_{RD}^{\mathcal{N}}T_{D}(\Gamma_{DE}%
^{\mathcal{M}})],
\end{align}
where $D\simeq S$, the operator $\Gamma_{RD}^{\mathcal{N}}$ is the Choi
operator of $\mathcal{N}_{C\rightarrow D}$, and $\Gamma_{SE}^{\mathcal{M}}$ is
the Choi operator of $\mathcal{M}_{D\rightarrow E}$.

\section{Ideal bidirectional teleportation}

\label{sec:ideal-BQT}Let us examine the case of ideal bidirectional
teleportation on two qudits in detail \cite{V94}. Doing so is helpful for us
in establishing a basic metric for the performance of unideal bidirectional
teleportation. Put simply, ideal bidirectional teleportation consists of an
ideal unidirectional teleportation \cite{BBC+93} from Alice to Bob and an
ideal unidirectional teleportation from Bob to Alice, where Alice and Bob are
two spatially separated parties. As such, the protocol uses entanglement and
classical communication to simulate the following unitary swap channel:%
\begin{equation}
\mathcal{S}_{AB}^{d}(\rho_{AB})\coloneqq F_{AB}(\rho_{AB})F_{AB}^{\dag},
\label{eq:ideal-swap-channel}%
\end{equation}
where the unitary swap or flip operator $F$\ is defined as%
\begin{equation}
F_{AB}\coloneqq\sum_{i,j=0}^{d-1}|i\rangle\!\langle j|_{A}\otimes
|j\rangle\!\langle i|_{B}. \label{eq:swap-op-def}%
\end{equation}
In the above, $A\simeq B$, and $\{|i\rangle_{A}\}_{i=0}^{d-1}$ and
$\{|i\rangle_{B}\}_{i=0}^{d-1}$\ are orthonormal bases. Denoting the identity
operator from Alice to Bob by $I_{A\rightarrow B}$%
\begin{equation}
I_{A\rightarrow B}\coloneqq \sum_{i=0}^{d-1}|i\rangle_{B}\langle i|_{A},
\end{equation}
and the identity operator from Bob to Alice by $I_{B\rightarrow A}$:%
\begin{equation}
I_{B\rightarrow A}\coloneqq \sum_{i=0}^{d-1}|i\rangle_{A}\langle i|_{B},
\end{equation}
we see that%
\begin{align}
F_{AB}  &  =\sum_{i,j=0}^{d-1}|i\rangle_{A}\langle j|_{A}\otimes|j\rangle
_{B}\langle i|_{B}\\
&  =\sum_{i,j=0}^{d-1} |j\rangle_{B}\langle j|_{A} \otimes|i\rangle_{A}\langle
i|_{B}\\
&  =\sum_{j=0}^{d-1}|j\rangle_{B}\langle j|_{A} \otimes\sum_{i=0}%
^{d-1}|i\rangle_{A}\langle i|_{B}\\
&  =I_{A\rightarrow B}\otimes I_{B\rightarrow A}.
\end{align}
The identity operators realize the following identity channels%
\begin{align}
\operatorname{id}_{A\rightarrow B}(\cdot)  &  \coloneqq I_{A\rightarrow
B}(\cdot)(I_{A\rightarrow B})^{\dag},\\
\operatorname{id}_{B\rightarrow A}(\cdot)  &  \coloneqq I_{B\rightarrow
A}(\cdot)(I_{B\rightarrow A})^{\dag},
\end{align}
and we see that the ideal swap\ channel is equivalent to%
\begin{equation}
\mathcal{S}_{AB}^{d}=\operatorname{id}_{B\rightarrow A}\otimes
\operatorname{id}_{A\rightarrow B}.
\end{equation}
Even though our choice of notation might suggest that the swap channel is a
tensor product of local identity channels, we should note that this is not the
case:\ the swap channel is a global channel that cannot be realized by local
actions alone. Our notation $\operatorname{id}_{A\rightarrow B}$ indicates
that Alice's input system $A$ is placed at Bob's output $B$ and the notation
$\operatorname{id}_{B\rightarrow A}$ indicates that Bob's input system $B$ is
placed at Alice's output $A$.

In more detail, recall that the standard, ideal unidirectional teleportation
protocol \cite{BBC+93}\ begins with Alice and Bob sharing the following
maximally entangled resource state:%
\begin{equation}
\Phi_{\hat{A}B}^{d}\coloneqq\frac{1}{d}\sum_{i=0}^{d-1}|i\rangle\!\langle
j|_{\hat{A}}\otimes|i\rangle\!\langle j|_{B}, \label{eq:max-ent-state}%
\end{equation}
where $\hat{A}\simeq B$ and $\{|i\rangle_{\hat{A}}\}_{i=0}^{d-1}$ and
$\{|i\rangle_{B}\}_{i=0}^{d-1}$\ are orthonormal bases. The amount of
entanglement in this state is $\log_{2}d$ \cite{BBPS96}, and so the state
above is said to be equivalent to $\log_{2}d$ ebits. Alice then prepares the
system $A$ in the state $\rho_{A}$, where $A\simeq\hat{A}$, so that the
overall state is%
\begin{equation}
\rho_{A}\otimes\Phi_{\hat{A}B}^{d}.
\end{equation}
Alice performs a Bell measurement on systems $A\hat{A}$, which is specified in
terms of the following measurement operators:%
\begin{equation}
\{\Phi_{A\hat{A}}^{z,x}\}_{z,x\in\left\{  0,\ldots,d-1\right\}  },
\end{equation}
where%
\begin{align}
\Phi_{A\hat{A}}^{z,x}  &  \coloneqq(W_{A}^{z,x}\otimes I_{\hat{A}})\Phi
_{A\hat{A}}^{d}(W_{A}^{z,x}\otimes I_{\hat{A}})^{\dag},\\
W^{z,x}  &  \coloneqq Z(z)X(x),\\
Z(z)  &  \coloneqq\sum_{k=0}^{d-1}e^{\frac{2\pi ikz}{d}}|k\rangle\!\langle
k|,\\
X(x)  &  \coloneqq\sum_{k=0}^{d-1}|k\oplus_{d}x\rangle\!\langle k|,
\end{align}
and $\oplus_{d}$ denotes addition modulo $d$. Defining the Bell measurement in
terms of the following quantum instrument:%
\begin{equation}
\mathcal{B}_{A\hat{A}\rightarrow A\hat{A}C_{A}}(\omega_{A\hat{A}%
})\coloneqq\sum_{z,x=0}^{d-1}\Phi_{A\hat{A}}^{z,x}\omega_{A\hat{A}}\Phi
_{A\hat{A}}^{z,x}\otimes|z,x\rangle\!\langle z,x|_{C_{A}},
\end{equation}
the following identity holds%
\begin{align}
&  \mathcal{B}_{A\hat{A}\rightarrow A\hat{A}C_{A}}(\rho_{A}\otimes\Phi
_{\hat{A}B}^{d})\nonumber\\
&  =\sum_{z,x=0}^{d-1}\Phi_{A\hat{A}}^{z,x}(\rho_{A}\otimes\Phi_{\hat{A}B}%
^{d})\Phi_{A\hat{A}}^{z,x}\otimes|z,x\rangle\!\langle z,x|_{C_{A}}\\
&  =\frac{1}{d^{2}}\sum_{z,x=0}^{d-1}\Phi_{A\hat{A}}^{z,x}\otimes(W_{B}%
^{z,x})^{\dag}\rho_{B}W_{B}^{z,x}\otimes|z,x\rangle\!\langle z,x|_{C_{A}}.
\end{align}
Alice traces out the systems $A\hat{A}$, leaving the state%
\begin{multline}
(\operatorname{Tr}_{A\hat{A}}\circ\mathcal{B}_{A\hat{A}\rightarrow A\hat
{A}C_{A}})(\rho_{A}\otimes\Phi_{\hat{A}B}^{d})\\
=\frac{1}{d^{2}}\sum_{z,x=0}^{d-1}(W_{B}^{z,x})^{\dag}\rho_{B}W_{B}%
^{z,x}\otimes|z,x\rangle\!\langle z,x|_{C_{A}}%
\end{multline}
Alice then communicates the classical register $C_{A}$ to Bob over a $d^{2}%
$-dimensional classical channel, defined by%
\begin{equation}
\overline{\Delta}_{C_{A}\rightarrow C_{B}}(\cdot)\coloneqq \sum_{z,x=0}%
^{d-1}|z,x\rangle_{C_{B}}\langle z,x|_{C_{A}}(\cdot)|z,x\rangle_{C_{A}}\langle
z,x|_{C_{B}}.
\end{equation}
The amount of classical information that can be communicated by this channel
is $2\log_{2}d$ bits. Bob finally performs the following correction channel on
systems $BC_{B}$:%
\begin{multline}
\mathcal{C}_{BC_{B}\rightarrow B}(\omega_{BZX})\coloneqq\\
\sum_{z,x}W_{B}^{z,x}\langle z,x|_{C_{B}}\omega_{BC_{B}}|z,x\rangle_{C_{B}%
}(W_{B}^{z,x})^{\dag},
\end{multline}
so that%
\begin{multline}
\mathcal{C}_{BC_{B}\rightarrow B}\!\left(  \frac{1}{d^{2}}\sum_{z,x=0}%
^{d-1}(W_{B}^{z,x})^{\dag}\rho_{B}W_{B}^{z,x}\otimes|z,x\rangle\!\langle
z,x|_{C_{B}}\right) \\
=\rho_{B}.
\end{multline}
Since the state $\rho_{A}$ on Alice's system is perfectly reconstructed on
Bob's system $B$ at the end of this process, we conclude that the whole
process simulates the identity channel $\operatorname{id}_{A\rightarrow B}$.
In more detail, let us denote the channel realized by the whole protocol as%
\begin{multline}
\mathcal{T}_{A\rightarrow B}\coloneqq\\
\mathcal{C}_{BC_{B}\rightarrow B}\circ\overline{\Delta}_{C_{A}\rightarrow
C_{B}}\circ\operatorname{Tr}_{A\hat{A}}\circ\mathcal{B}_{A\hat{A}\rightarrow
A\hat{A}C_{A}}\circ\mathcal{P}_{A\rightarrow A\hat{A}B}^{\Phi}.
\end{multline}
The discussion above then argues that%
\begin{equation}
\mathcal{T}_{A\rightarrow B}=\operatorname{id}_{A\rightarrow B}.
\end{equation}
See, e.g., \cite[Section~6.5.3]{W17} for a more detailed argument.

The unidirectional teleportation protocol can be run in the opposite
direction, from Bob to Alice, and this process realizes the channel
$\mathcal{T}_{B\rightarrow A}$. Following the same argument above, but
swapping the roles of Alice and Bob, the equality $\mathcal{T}_{B\rightarrow
A}=\operatorname{id}_{B\rightarrow A}$ holds. Thus, by consuming $2\log_{2}d$
ebits, $2\log_{2}d$ bits of classical communication from Alice to Bob, and
$2\log_{2}d$ bits of classical communication from Bob to Alice, they can
realize the ideal swap channel in \eqref{eq:ideal-swap-channel}:%
\begin{equation}
\mathcal{S}_{AB}^{d}=\mathcal{T}_{A\rightarrow B}\otimes\mathcal{T}%
_{B\rightarrow A}.
\end{equation}

\section{Quantifying the performance of unideal bidirectional teleportation}

Having established that ideal bidirectional teleportation realizes an ideal
swap channel, let us now discuss unideal bidirectional teleportation.
Succinctly, the goal of unideal bidirectional teleportation is to use a
resource state $\rho_{\hat{A}\hat{B}}$ and local operations and classical
communication (LOCC) to simulate a $d$-dimensional swap channel of the form in
\eqref{eq:ideal-swap-channel}, where%
\begin{equation}
d=d_{A}=d_{B}. \label{eq:dim-swap-chan}%
\end{equation}
The goal is to make the error between the simulation and the ideal swap
channel as small as possible.

In more detail, we assume that Alice and Bob share a quantum state $\rho
_{\hat{A}\hat{B}}$, instead of two maximally entangled states of the form in
\eqref{eq:max-ent-state}. Such a state could be generated by an unideal
experimental process, such as spontaneous parametric down-conversion
\cite{Cout18}. We make no assumption about the dimensions of systems $\hat{A}$
and $\hat{B}$, other than that they are finite dimensional. In particular, it
need not be the case that $d_{\hat{A}}$ is equal to $d_{\hat{B}}$. There are
two other systems $A$ and $B$, that serve as inputs for Alice and Bob,
respectively, to the unideal bidirectional teleportation. They then act with
an LOCC channel $\mathcal{L}_{AB\hat{A}\hat{B}\rightarrow AB}$, on their input
systems $A$ and $B$ and their shares $\hat{A}$ and $\hat{B}$ of the resource
state $\rho_{\hat{A}\hat{B}}$, to produce the output systems $A$ and $B$. As
mentioned in Section~\ref{sec:prelim}, the LOCC channel $\mathcal{L}%
_{AB\hat{A}\hat{B}\rightarrow A^{\prime}B^{\prime}}$ can be written as%
\begin{equation}
\mathcal{L}_{AB\hat{A}\hat{B}\rightarrow AB}=\sum_{y}\mathcal{E}_{A\hat
{A}\rightarrow A}^{y}\otimes\mathcal{F}_{B\hat{B}\rightarrow B}^{y},
\end{equation}
where $\{\mathcal{E}_{A\hat{A}\rightarrow A}^{y}\}_{y}$ and $\{\mathcal{F}%
_{B\hat{B}\rightarrow B}^{y}\}_{y}$ are sets of completely positive maps such
that $\mathcal{L}_{AB\hat{A}\hat{B}\rightarrow AB}$ is a quantum channel.
Thus, the overall channel realized by the simulation is as follows:%
\begin{equation}
\widetilde{\mathcal{S}}_{AB}(\omega_{AB})\coloneqq \mathcal{L}_{AB\hat{A}%
\hat{B}\rightarrow AB}(\omega_{AB}\otimes\rho_{\hat{A}\hat{B}}),
\label{eq:swap-sim}%
\end{equation}
which is depicted in Figure~\ref{fig: Unideal BQT LOCC}. In the simulation, we
allow for classical communication between Alice and Bob for free, so that
$\mathcal{L}_{AB\hat{A}\hat{B}\rightarrow AB}$ can be considered a free
channel, as is common in the resource theory of
entanglement~\cite{BDSW96,CG18}.

\subsection{Quantifying error with normalized diamond distance}

\begin{figure}[ptb]
\begin{center}
\includegraphics[scale=1]{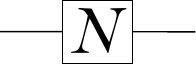}
\end{center}
\caption{Quantum channel $\mathcal{N}$.}%
\label{Ideal Channel N}%
\end{figure}

\begin{figure}[ptb]
\begin{center}
\includegraphics[scale=1]{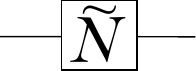}
\end{center}
\caption{$\widetilde{\mathcal{N}}$ is an imperfect simulation of the ideal
quantum channel $\mathcal{N}$.}%
\label{Noisy Channel N Tilde}%
\end{figure}

\label{sec:q-error-NDD}Let us now discuss how to quantify the simulation
error. The standard metric for doing so is the normalized diamond distance
\cite{Kit97}, having been used in both quantum computation \cite{Kit97} and
quantum information \cite{Watrous2018,W17}. In short,
this metric quantifies the maximum absolute deviation between the
probabilities of observing the same outcome when each quantum channel is
applied to the same input state and the same measurement is made.

Let us elaborate upon the explanation of diamond distance given above, in a
similar way to the motivation presented in \cite{Wang2019a,W20ISIT}. Suppose
that the ideal channel to be implemented is $\mathcal{N}_{C\rightarrow D}$,
shown in Figure~\ref{Ideal Channel N}. Suppose further that $\widetilde
{\mathcal{N}}_{C\rightarrow D}$ is the simulation of $\mathcal{N}%
_{C\rightarrow D}$, shown in Figure~\ref{Noisy Channel N Tilde}. To interface
with these channels and obtain classical data for the purpose of
distinguishing them, the most general way for doing so is to prepare a state
$\rho_{RC}$ of a reference system $R$ and the channel input system $C$, feed
system~$C$ into the unknown channel (either $\mathcal{N}_{C\rightarrow D}$ or
$\widetilde{\mathcal{N}}_{C\rightarrow D}$), and then perform a quantum
measurement $\{\Lambda_{RD}^{x}\}_{x\in\mathcal{X}}$ on the channel output
system $D$ and the reference system $R$. To be a legitimate quantum
measurement, the set $\{\Lambda_{RD}^{x}\}_{x\in\mathcal{X}}$ of operators
should satisfy $\sum_{x\in\mathcal{X}}\Lambda_{RD}^{x}=I_{RD}$ and
$\Lambda_{RD}^{x}\geq0$ for all $x\in\mathcal{X}$. The result of this
procedure (preparation, channel evolution, and measurement) is a classical
outcome $x\in\mathcal{X}$ that occurs with probability $\operatorname{Tr}%
[\Lambda_{RD}^{x}\mathcal{N}_{C\rightarrow D}(\rho_{RC})]$ if the channel
$\mathcal{N}_{C\rightarrow D}$ is applied, while the outcome $x\in\mathcal{X}$
occurs with probability $\operatorname{Tr}[\Lambda_{RD}^{x}\widetilde
{\mathcal{N}}_{C\rightarrow D}(\rho_{RC})]$ if the channel $\widetilde
{\mathcal{N}}_{C\rightarrow D}$ is applied. The error or difference between
these probabilities is naturally quantified by the absolute deviation%
\begin{equation}
\left\vert \operatorname{Tr}[\Lambda_{RD}^{x}\mathcal{N}_{C\rightarrow D}%
(\rho_{RC})]-\operatorname{Tr}[\Lambda_{RD}^{x}\widetilde{\mathcal{N}%
}_{C\rightarrow D}(\rho_{RC})]\right\vert . \label{eq:ch-error-proc-depend}%
\end{equation}
We can then quantify the maximum possible error between the channels
$\mathcal{N}_{C\rightarrow D}$ and $\widetilde{\mathcal{N}}_{C\rightarrow D}$
by optimizing \eqref{eq:ch-error-proc-depend} with respect to all preparations
and measurements:%
\begin{multline}
\sup_{\substack{\rho_{RC},\\\{\Lambda_{RD}^{x}\}_{x}}}\left\vert
\operatorname{Tr}[\Lambda_{RD}^{x}\mathcal{N}(\rho_{RC})]-\operatorname{Tr}%
[\Lambda_{RD}^{x}\widetilde{\mathcal{N}}(\rho_{RC})]\right\vert
\label{eq:ch-error}\\
=\sup_{\substack{\rho_{RC},\\0\leq\Lambda_{RD}\leq I_{RD}}}\left\vert
\operatorname{Tr}[\Lambda_{RD}\mathcal{N}(\rho_{RC})]-\operatorname{Tr}%
[\Lambda_{RD}\widetilde{\mathcal{N}}(\rho_{RC})]\right\vert ,
\end{multline}
where it is implicit that the channels $\mathcal{N}_{C\rightarrow D}$ and
$\widetilde{\mathcal{N}}_{C\rightarrow D}$ above have input system $C$ and
output system $D$. Mathematically, this optimization has the effect of
removing the dependence on the preparation and measurement such that the error
is a function solely of the two channels $\mathcal{N}_{C\rightarrow D}$ and
$\widetilde{\mathcal{N}}_{C\rightarrow D}$. It is a fundamental and well known
result in quantum information theory \cite{Kit97,AKN98}\ that the error in
\eqref{eq:ch-error} is equal to the normalized diamond distance:%
\begin{equation}
\text{Eq.}~\eqref{eq:ch-error}=\frac{1}{2}\left\Vert \mathcal{N}%
-\widetilde{\mathcal{N}}\right\Vert _{\diamond},
\end{equation}
where the diamond distance $\left\Vert \mathcal{N}-\widetilde{\mathcal{N}%
}\right\Vert _{\diamond}$ is defined as%
\begin{equation}
\left\Vert \mathcal{N}-\widetilde{\mathcal{N}}\right\Vert _{\diamond
}\coloneqq \sup_{\rho_{RC}}\left\Vert \mathcal{N}_{C\rightarrow D}(\rho
_{RC})-\widetilde{\mathcal{N}}_{C\rightarrow D}(\rho_{RC})\right\Vert _{1},
\label{eq:def-diamond-distance}%
\end{equation}
and the trace norm of an operator $X$ is given by $\left\Vert X\right\Vert
_{1}=\operatorname{Tr}[\left\vert X\right\vert ]$, where $\left\vert
X\right\vert \coloneqq \sqrt{X^{\dag}X}$. It is well known that the
calculation of the diamond distance simplifies as follows:%
\begin{equation}
\left\Vert \mathcal{N}-\widetilde{\mathcal{N}}\right\Vert _{\diamond}%
=\sup_{\psi_{RC}}\left\Vert \mathcal{N}_{C\rightarrow D}(\psi_{RC}%
)-\widetilde{\mathcal{N}}_{C\rightarrow D}(\psi_{RC})\right\Vert _{1},
\label{eq:d-dist-pure-states}%
\end{equation}
where the optimization is with respect to every pure bipartite state
$\psi_{RC}$ with system $R$ isomorphic to the channel input system $C$.

Intuitively, the diamond distance can be thought of as a metric characterizing
the distinguishability of two quantum channels. As indicated in the previous
paragraph, this metric quantifies the maximum absolute deviation between the
probabilities of observing the same outcome when each quantum channel is
applied to the same input state and the same measurement is made. It is used
as a way to quantify the distance between two quantum channels and is a
standard metric used in quantum computation and quantum information.

The normalized diamond distance can be computed by means of the following
semi-definite program \cite{Wat09}:%
\begin{multline}
\frac{1}{2}\left\Vert \mathcal{N}-\widetilde{\mathcal{N}}\right\Vert
_{\diamond}=\\
\inf_{Z_{RD}\geq0}\left\{  \left\Vert \operatorname{Tr}_{D}[Z_{RD}]\right\Vert
_{\infty}:Z_{RD}\geq\Gamma_{RD}^{\mathcal{N}}-\Gamma_{RD}^{\widetilde
{\mathcal{N}}}\right\}  ,
\end{multline}
where $\Gamma_{RD}^{\mathcal{N}}$ and $\Gamma_{RD}^{\widetilde{\mathcal{N}}}$
are the Choi operators of $\mathcal{N}$ and $\widetilde{\mathcal{N}}$,
respectively. This will be helpful for us later on in
Sections~\ref{sec:SDP-lower-bound-general} and \ref{sec:SDP-lower-bound-BTP}.

Returning to our case of interest, the simulation error when employing a
specific LOCC channel $\mathcal{L}_{AB\hat{A}\hat{B}\rightarrow AB}$ is
quantified as follows:%
\begin{equation}
e_{\operatorname{LOCC}}(\mathcal{S}_{AB}^{d},\rho_{\hat{A}\hat{B}}%
,\mathcal{L}_{AB\hat{A}\hat{B}\rightarrow AB})\coloneqq \frac{1}{2}\left\Vert
\mathcal{S}^{d}-\widetilde{\mathcal{S}}\right\Vert _{\diamond},
\end{equation}
where $d$ is the dimension of the swap channel $\mathcal{S}^{d}$\ (see
\eqref{eq:dim-swap-chan}), $\mathcal{S}^{d}$ is defined in
\eqref{eq:ideal-swap-channel}, and $\widetilde{\mathcal{S}}$ in
\eqref{eq:swap-sim}. Since we are interested in the minimum possible
simulation error, taken over all possible LOCC channels, we define the
simulation error of bidirectional teleportation, when employing the resource
state $\rho_{\hat{A}\hat{B}}$, as follows:%
\begin{multline}
e_{\operatorname{LOCC}}(\mathcal{S}_{AB}^{d},\rho_{\hat{A}\hat{B}%
})\coloneqq\label{eq:sim-err-swap-DD}\\
\inf_{\mathcal{L}\in\operatorname{LOCC}}e_{\operatorname{LOCC}}(\mathcal{S}%
_{AB}^{d},\rho_{\hat{A}\hat{B}},\mathcal{L}_{AB\hat{A}\hat{B}\rightarrow AB}).
\end{multline}

This simulation error is difficult to compute as $d$, $d_{\hat{A}}$, and
$d_{\hat{B}}$ become larger. This computational difficulty is related to how
it is difficult to optimize functions over the set of separable states
\cite{Gur04,Ghar10,Harrow:2013}. In Section~\ref{sec:SDP-lower-bound-BTP}, we
determine a lower bound on the simulation error $e_{\operatorname{LOCC}%
}(\mathcal{S}_{AB}^{d},\rho_{\hat{A}\hat{B}})$ that can be computed by means
of semi-definite programming \cite{Vandenberghe1996}\ and is thus efficiently
computable. For some states $\rho_{\hat{A}\hat{B}}$ of interest, the lower
bound is achievable, so that we can determine the error of unideal
bidirectional teleportation precisely in these cases.

\subsection{Quantifying error with channel infidelity}

\label{sec:q-error-CIF}

Another way to quantify error between channels is by using the fidelity.
Recall that the fidelity of quantum states $\omega$ and $\tau$ is defined as
follows \cite{Uhl76}:%
\begin{equation}
F(\omega,\tau)\coloneqq \left\Vert \sqrt{\omega}\sqrt{\tau}\right\Vert
_{1}^{2}.
\end{equation}
This quantity is equal to one if and only if the states $\omega$ and $\tau$
are the same, and it is equal to zero if and only if the states are
orthogonal. If $\omega$ is a pure state, i.e., equal to $|\psi\rangle
\!\langle\psi|$ for some unit vector $|\psi\rangle$, then the fidelity reduces
to the following expression:%
\begin{equation}
F(|\psi\rangle\!\langle\psi|,\tau)=\langle\psi|\tau|\psi\rangle.
\end{equation}
In this case, it has the operational meaning that $F(|\psi\rangle\!\langle
\psi|,\tau)$ is the probability with which the state $\tau$ passes a test for
being the state $|\psi\rangle\!\langle\psi|$. The test in this case is given
by the binary measurement $\{|\psi\rangle\!\langle\psi|,I-|\psi\rangle
\!\langle\psi|\}$, and the first outcome corresponds to the decision
\textquotedblleft pass.\textquotedblright\ So the probability of passing is
equal to $F(|\psi\rangle\!\langle\psi|,\tau)$.

We can then lift this to a measure of similarity for quantum channels
$\mathcal{N}_{C\rightarrow D}$ and $\widetilde{\mathcal{N}}_{C\rightarrow D}$
as follows:%
\begin{equation}
F(\mathcal{N},\widetilde{\mathcal{N}})\coloneqq \inf_{\rho_{RC}}%
F(\mathcal{N}_{C\rightarrow D}(\rho_{RC}),\widetilde{\mathcal{N}%
}_{C\rightarrow D}(\rho_{RC})),
\end{equation}
which can be viewed as the fidelity counterpart of the diamond distance in
\eqref{eq:def-diamond-distance}. Just like \eqref{eq:d-dist-pure-states}, the
following simplification holds%
\begin{equation}
F(\mathcal{N},\widetilde{\mathcal{N}})=\inf_{\psi_{RC}}F(\mathcal{N}%
_{C\rightarrow D}(\psi_{RC}),\widetilde{\mathcal{N}}_{C\rightarrow D}%
(\psi_{RC})),
\end{equation}
where the optimization is with respect to all pure bipartite states $\psi
_{RC}$ with system $R$ isomorphic to the channel input system $C$. We note
here that the channel infidelity is defined as
\begin{equation}
1 - F(\mathcal{N},\widetilde{\mathcal{N}}).
\end{equation}
Since the channel fidelity is a measure of similarity, the channel infidelity
is a measure of distinguishability and thus can be understood as an error
measure in our context. Thus, in what shortly follows, we employ it as a
simulation error, with the goal of minimizing it.

The root fidelity of channels can computed by means of the following
semi-definite program \cite{Yuan2017,KW20}:%
\begin{multline}
\sqrt{F}(\mathcal{N},\widetilde{\mathcal{N}})=\\
\frac{1}{2}\inf_{\rho_{R},W_{RD},Z_{RD}}\operatorname{Tr}[W_{RD}\Gamma
_{RD}^{\mathcal{N}}]+\operatorname{Tr}[Z_{RD}\Gamma_{RD}^{\widetilde
{\mathcal{N}}}],
\end{multline}
subject to%
\begin{equation}
\rho_{R}\geq0,\quad\operatorname{Tr}[\rho_{R}]=1,\quad%
\begin{bmatrix}
W_{RD} & \rho_{R}\otimes I_{D}\\
\rho_{R}\otimes I_{D} & Z_{RD}%
\end{bmatrix}
\geq0.
\end{equation}
In the above, the optimization is over all linear operators $W_{RD}$ and
$Z_{RD}$, and $\Gamma_{RD}^{\mathcal{N}}$ and $\Gamma_{RD}^{\widetilde
{\mathcal{N}}}$ are the Choi operators of $\mathcal{N}$ and $\widetilde
{\mathcal{N}}$, respectively. The dual of this semi-definite program is given
by%
\begin{equation}
\sup_{\lambda\geq0,Q_{RD}}\lambda\label{eq:SDP-ch-fid-1}%
\end{equation}
subject to%
\begin{align}
\lambda I_{R}  &  \leq\operatorname{Re}[\operatorname{Tr}_{D}[Q_{RD}]],\\%
\begin{bmatrix}
\Gamma_{RD}^{\widetilde{\mathcal{N}}} & Q_{RD}^{\dag}\\
Q_{RD} & \Gamma_{RD}^{\mathcal{N}}%
\end{bmatrix}
&  \geq0, \label{eq:SDP-ch-fid-3}%
\end{align}

Using the infidelity of channels, we can define an alternate notion of
simulation error as follows:%
\begin{equation}
e_{\operatorname{LOCC}}^{F}(\mathcal{S}_{AB}^{d},\rho_{\hat{A}\hat{B}%
},\mathcal{L}_{AB\hat{A}\hat{B}\rightarrow AB})\coloneqq 1-F(\mathcal{S}%
^{d},\widetilde{\mathcal{S}}),
\end{equation}
where $d$ is the dimension of the swap channel $\mathcal{S}^{d}$\ (see
\eqref{eq:dim-swap-chan}), $\mathcal{S}^{d}$ is defined in
\eqref{eq:ideal-swap-channel}, and $\widetilde{\mathcal{S}}$ in
\eqref{eq:swap-sim}. Minimizing this error with respect to all LOCC\ channels,
we arrive at the following:%
\begin{multline}
e_{\operatorname{LOCC}}^{F}(\mathcal{S}_{AB}^{d},\rho_{\hat{A}\hat{B}%
})\coloneqq\label{eq:sim-err-swap-ch-infid}\\
\inf_{\mathcal{L}\in\operatorname{LOCC}}e_{\operatorname{LOCC}}^{F}%
(\mathcal{S}_{AB}^{d},\rho_{\hat{A}\hat{B}},\mathcal{L}_{AB\hat{A}\hat
{B}\rightarrow AB}).
\end{multline}
For the same reasons given previously, this quantity is difficult to compute,
and so we seek alternative ways to estimate it.

\subsection{Equality of simulation errors when simulating the swap channel}

\label{sec:equality-sim-errors-LOCC-BQT}

Even though we have defined two different notions of LOCC\ simulation error of
bidirectional teleportation based on the normalized diamond distance and
channel infidelity, it turns out that they are equal. This result follows as a
consequence of the swap channel $\mathcal{S}_{AB}^{d}$ in
\eqref{eq:ideal-swap-channel} having the following symmetry:%
\begin{equation}
\mathcal{S}_{AB}^{d}\left(  \mathcal{U}_{A}\otimes\mathcal{V}_{B}\right)
=\left(  \mathcal{V}_{A}\otimes\mathcal{U}_{B}\right)  \mathcal{S}_{AB}^{d},
\label{eq:swap-ch-symmetry-for-LOCC}%
\end{equation}
holding for all unitary channels $\mathcal{U}_{A}$ and $\mathcal{V}_{B}$. An
additional symmetry of the swap channel $\mathcal{S}_{AB}^{d}$ is that it
commutes with itself:%
\begin{equation}
\mathcal{S}_{AB}^{d}\circ\mathcal{S}_{AB}^{d}=\mathcal{S}_{AB}^{d}%
\circ\mathcal{S}_{AB}^{d}. \label{eq:swap-ch-add-symmetry-for-LOCC}%
\end{equation}
Although at first glance this latter symmetry might seem trivial, it is
actually helpful in further simplifying the optimization problem for
bidirectional teleportation. More generally, a unitary $U$ is a symmetry of a
channel $\mathcal{N}$ if it commutes with the action of the channel:
$\mathcal{N}\circ U=U\circ\mathcal{N}$. Clearly, the equalities in
\eqref{eq:swap-ch-symmetry-for-LOCC} and
\eqref{eq:swap-ch-add-symmetry-for-LOCC} represent symmetries of the swap
channel $\mathcal{S}_{AB}^{d}$.

By exploiting the symmetries in \eqref{eq:swap-ch-symmetry-for-LOCC} and
\eqref{eq:swap-ch-add-symmetry-for-LOCC}, we prove in
Appendix~\ref{app:simplify-LOCC-err-swap}\ that it suffices to optimize both
the normalized diamond distance and the channel infidelity with respect to
LOCC\ channels $\mathcal{L}_{AB\hat{A}\hat{B}\rightarrow AB}$ having the
following form:%
\begin{multline}
\mathcal{L}_{AB\hat{A}\hat{B}\rightarrow AB}(\omega_{AB}\otimes\rho_{\hat
{A}\hat{B}})=\mathcal{S}_{AB}^{d}(\omega_{AB})\operatorname{Tr}[K_{\hat{A}%
\hat{B}}\rho_{\hat{A}\hat{B}}]\label{eq:sim-ch-LOCC-opt}\\
+\frac{1}{2}\left(  \operatorname{id}_{A\rightarrow B}\otimes\mathcal{D}%
_{B\rightarrow A}+\mathcal{D}_{A\rightarrow B}\otimes\operatorname{id}%
_{B\rightarrow A}\right)  (\omega_{AB})\operatorname{Tr}[L_{\hat{A}\hat{B}%
}\rho_{\hat{A}\hat{B}}]\\
+\left(  \mathcal{D}_{A\rightarrow B}\otimes\mathcal{D}_{B\rightarrow
A}\right)  (\omega_{AB})\operatorname{Tr}[N_{\hat{A}\hat{B}}\rho_{\hat{A}%
\hat{B}}],
\end{multline}
where $\mathcal{D}$ denotes the following generalized Pauli channel:%
\begin{equation}
\mathcal{D}(\sigma)\coloneqq\frac{1}{d^{2}-1}\sum_{\left(  x,z\right)
\neq\left(  0,0\right)  }W^{z,x}\sigma(W^{z,x})^{\dag}.
\end{equation}
Thus, the interpretation of the simulating channel is that it measures the
resource state $\rho_{\hat{A}\hat{B}}$ according to the POVM\ $\{K_{\hat
{A}\hat{B}},L_{\hat{A}\hat{B}},N_{\hat{A}\hat{B}}\}$, which is subject to the
constraint that the overall channel $\mathcal{L}_{AB\hat{A}\hat{B}\rightarrow
A^{\prime}B^{\prime}}$ is LOCC. After that, it takes the following action:

\begin{enumerate}
\item If the first outcome $K_{\hat{A}\hat{B}}$ occurs, then apply the ideal
swap channel to the input state $\omega_{AB}$.

\item If the second outcome $L_{\hat{A}\hat{B}}$ occurs, then with probability
1/2, apply the identity channel $\operatorname{id}_{A\rightarrow B}$ to
transfer Alice's input system $A$ to Bob, but then garble Bob's input system
$B$ by applying the channel $\mathcal{D}$ and transfer the resulting system to
Alice; with probability 1/2, apply the identity channel $\operatorname{id}%
_{B\rightarrow A}$ to transfer Bob's input system $B$ to Alice, but then
garble Alice's input system $A$ by applying the channel $\mathcal{D}$ and
transfer the resulting system to Bob.

\item If the third outcome $N_{\hat{A}\hat{B}}$ occurs, then apply the
garbling channel $\mathcal{D}$ to both Alice and Bob's systems individually
and exchange them.
\end{enumerate}

\noindent We again stress that the constraint on the set $\{K_{\hat{A}\hat{B}%
},L_{\hat{A}\hat{B}},N_{\hat{A}\hat{B}}\}$ is that the overall channel
$\mathcal{L}_{AB\hat{A}\hat{B}\rightarrow AB}$ is LOCC.

We state the equality of the simulation errors as follows and prove this
result in Appendix~\ref{app:simplify-LOCC-err-swap}:

\begin{proposition}
\label{prop:err-collapse-LOCC}The optimization problems in
\eqref{eq:sim-err-swap-DD} and \eqref{eq:sim-err-swap-ch-infid},\ for the
error in simulating the unitary SWAP channel $\mathcal{S}_{AB}^{d}$ in
\eqref{eq:ideal-swap-channel}, simplify as follows:%
\begin{align}
&  e_{\operatorname{LOCC}}(\mathcal{S}_{AB}^{d},\rho_{\hat{A}\hat{B}%
})\nonumber\\
&  =e_{\operatorname{LOCC}}^{F}(\mathcal{S}_{AB}^{d},\rho_{\hat{A}\hat{B}})\\
&  =1-\sup_{K_{\hat{A}\hat{B}},L_{\hat{A}\hat{B}},N_{\hat{A}\hat{B}}\geq
0}\operatorname{Tr}[\rho_{\hat{A}\hat{B}}K_{\hat{A}\hat{B}}],
\end{align}
subject to $K_{\hat{A}\hat{B}}+L_{\hat{A}\hat{B}}+N_{\hat{A}\hat{B}}%
=I_{\hat{A}\hat{B}}$ and the following channel $\mathcal{L}_{AB\hat{A}\hat
{B}\rightarrow AB}$ being LOCC:%
\begin{align}
&  \mathcal{L}_{AB\hat{A}\hat{B}\rightarrow AB}(\omega_{AB\hat{A}\hat{B}%
})\nonumber\\
&  =\mathcal{S}_{AB}^{d}(\operatorname{Tr}_{\hat{A}\hat{B}}[K_{\hat{A}\hat{B}%
}\omega_{AB\hat{A}\hat{B}}])\nonumber\\
&  \qquad+\frac{1}{2}\left(
\begin{array}
[c]{c}%
\operatorname{id}_{A\rightarrow B}\otimes\mathcal{D}_{B\rightarrow A}\\
+\mathcal{D}_{A\rightarrow B}\otimes\operatorname{id}_{B\rightarrow A}%
\end{array}
\right)  (\operatorname{Tr}_{\hat{A}\hat{B}}[L_{\hat{A}\hat{B}}\omega
_{AB\hat{A}\hat{B}}])\nonumber\\
&  \qquad+\left(  \mathcal{D}_{A\rightarrow B}\otimes\mathcal{D}_{B\rightarrow
A}\right)  (\operatorname{Tr}_{\hat{A}\hat{B}}[N_{\hat{A}\hat{B}}%
\omega_{AB\hat{A}\hat{B}}]).
\end{align}

\end{proposition}

As a consequence of Proposition~\ref{prop:err-collapse-LOCC}, there is no need
for two different notions of simulation error when considering the simulation
of the swap channel.

\begin{remark}
If one had to pick one error metric over the other, we think the diamond
distance is preferable for comparing general channels. It captures a notion of
error that makes physical sense as the largest deviation in outcome
probabilities that could be observed by performing the most general physical
procedure to distinguish an ideal channel from its simulation. Related to
this, it has an operational interpretation in terms of hypothesis testing of
channels. It also has nice properties like the triangle inequality, data
processing under the action of a superchannel, and stability under tensoring
with the identity. For these reasons, it is the standard theoretical tool used
in the study of fault tolerant quantum computation, and one can consult
\cite{Watrous2018} and find it used to define quantum channel capacities. Thus
we are using it here also.

The channel fidelity has a sensible operational interpretation if the target
channel is a unitary channel, as the probability with which the simulation
channel can pass a test for being the unitary channel. It also possesses the
properties of stability and data processing mentioned above, and if one takes
the square root of the infidelity (often called sine distance), then it also
obeys the triangle inequality.

So we view both of these error metrics as being important and thus we have
considered them both here. In light of the fact that these error metrics are
generally different, we find it an interesting conclusion that the normalized
diamond distance and the infidelity give the same value when considering
simulation of the SWAP channel.
\end{remark}

\subsection{LOCC\ simulation of general bipartite
channels\label{sec:LOCC-sim-gen-bi}}

In the previous sections, we discussed how to quantify the simulation error
for unideal bidirectional teleportation. In this section, we generalize the
task to the LOCC\ simulation of an arbitrary bipartite channel.

A bipartite channel $\mathcal{N}_{AB\rightarrow A^{\prime}B^{\prime}}$,
depicted in Figure~\ref{fig: Bipartite Channel N}, is a quantum channel with
input systems $A$ and $B$ and output systems $A^{\prime}$ and $B^{\prime}$
\cite{BHLS03,CLL06}. Alice has control of the systems $A$ and $A^{\prime}$,
and Bob has control of the output systems $B$ and $B^{\prime}$. The swap
channel in \eqref{eq:ideal-swap-channel} is a particular example of a
bipartite channel, but of course there are many other interesting examples,
such as the controlled-NOT gate.

We can thus generalize the simulation task in the previous section to be about
simulating a general bipartite channel $\mathcal{N}_{AB\rightarrow A^{\prime
}B^{\prime}}$. This was considered for point-to-point channels in
\cite{BDSW96,HHH99} and for bipartite channels in \cite{BHLS03,BDWW19,GMS19}.
In this case, the simulating channel is defined similarly to
\eqref{eq:swap-sim}:%
\begin{equation}
\widetilde{\mathcal{N}}_{AB\rightarrow A^{\prime}B^{\prime}}(\omega
_{AB})\coloneqq \mathcal{L}_{AB\hat{A}\hat{B}\rightarrow A^{\prime}B^{\prime}%
}(\omega_{AB}\otimes\rho_{\hat{A}\hat{B}}),
\end{equation}
where $\rho_{\hat{A}\hat{B}}$ is a resource state and $\mathcal{L}_{AB\hat
{A}\hat{B}\rightarrow A^{\prime}B^{\prime}}$ is an LOCC channel. The
simulation error when employing a specific LOCC channel $\mathcal{L}%
_{AB\hat{A}\hat{B}\rightarrow A^{\prime}B^{\prime}}$ is quantified as follows:%
\begin{equation}
e_{\operatorname{LOCC}}(\mathcal{N}_{AB\rightarrow A^{\prime}B^{\prime}}%
,\rho_{\hat{A}\hat{B}},\mathcal{L}_{AB\hat{A}\hat{B}\rightarrow A^{\prime
}B^{\prime}})\coloneqq \frac{1}{2}\left\Vert \mathcal{N}-\widetilde
{\mathcal{N}}\right\Vert _{\diamond}, \label{eq:err-sim-bi-ch-fixed-locc}%
\end{equation}
and the simulation error minimized over all possible LOCC channels is%
\begin{multline}
e_{\operatorname{LOCC}}(\mathcal{N}_{AB\rightarrow A^{\prime}B^{\prime}}%
,\rho_{\hat{A}\hat{B}})\coloneqq\label{eq:sim-err-bipartite}\\
\inf_{\mathcal{L}\in\operatorname{LOCC}}e_{\operatorname{LOCC}}(\mathcal{N}%
_{AB\rightarrow A^{\prime}B^{\prime}},\rho_{\hat{A}\hat{B}},\mathcal{L}%
_{AB\hat{A}\hat{B}\rightarrow A^{\prime}B^{\prime}}).
\end{multline}
Just as before, the simulation error is difficult to compute because it
involves an optimization over LOCC channels.

\begin{figure}[ptb]
\begin{center}
\includegraphics[scale=1]{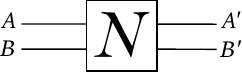}
\end{center}
\caption{An arbitary bipartite channel $\mathcal{N}$ with input systems $A$,
$B$ and output systems $A^{\prime}$, $B^{\prime}$.}%
\label{fig: Bipartite Channel N}%
\end{figure}


We could alternatively employ the infidelity to quantify the simulation error
(note that for channels other than the swap channel, the simulation errors
resulting from diamond distance and infidelity are generally different).\ For
this case, the infidelity simulation error is defined similarly to
\eqref{eq:err-sim-bi-ch-fixed-locc} and \eqref{eq:sim-err-bipartite} as
follows:%
\begin{multline}
e_{\operatorname{LOCC}}^{F}(\mathcal{N}_{AB\rightarrow A^{\prime}B^{\prime}%
},\rho_{\hat{A}\hat{B}},\mathcal{L}_{AB\hat{A}\hat{B}\rightarrow A^{\prime
}B^{\prime}})\coloneqq\\
1-F(\mathcal{N},\widetilde{\mathcal{N}}),
\end{multline}%
\begin{multline}
e_{\operatorname{LOCC}}^{F}(\mathcal{N}_{AB\rightarrow A^{\prime}B^{\prime}%
},\rho_{\hat{A}\hat{B}})\coloneqq\\
\inf_{\mathcal{L}\in\operatorname{LOCC}}e_{\operatorname{LOCC}}^{F}%
(\mathcal{N}_{AB\rightarrow A^{\prime}B^{\prime}},\rho_{\hat{A}\hat{B}%
},\mathcal{L}_{AB\hat{A}\hat{B}\rightarrow A^{\prime}B^{\prime}}).
\label{eq:sim-err-bipartite-infid}%
\end{multline}

\section{Semi-definite programming lower bounds}

\subsection{Semi-definite programming lower bound on the error in
LOCC\ simulation of bipartite channels\label{sec:SDP-lower-bound-general}}

As discussed above, it is challenging to compute the simulation error
$e_{\operatorname{LOCC}}(\mathcal{N}_{AB\rightarrow A^{\prime}B^{\prime}}%
,\rho_{\hat{A}\hat{B}})$ in \eqref{eq:sim-err-bipartite}\ because it is
difficult to optimize over the set of LOCC\ channels \cite{Gur04,Ghar10}. For
this reason, we follow the approach of \cite{Rai99,Rai01} and enlarge the set
LOCC\ to the set of completely positive-partial-transpose-preserving channels
(denoted as C-PPT-P for short). Then we can optimize with respect to this
larger set of channels and obtain a lower bound on the error in \eqref{eq:sim-err-bipartite}.

In more detail, recall that a bipartite channel $\mathcal{P}_{AB\rightarrow
A^{\prime}B^{\prime}}$\ is defined to be C-PPT-P \cite{Rai99,Rai01} if the map%
\begin{equation}
T_{B^{\prime}}\circ\mathcal{P}_{AB\rightarrow A^{\prime}B^{\prime}}\circ
T_{B}\text{ is completely positive}, \label{eq:C-PPT-P-condition}%
\end{equation}
where $T_{B}$ denotes the transpose map, defined by%
\begin{equation}
T_{B}(\omega_{B})=\sum_{i,j}|i\rangle\!\langle j|_{B}\omega_{B}|i\rangle
\!\langle j|_{B}, \label{eq:transpose-map}%
\end{equation}
and with $T_{B^{\prime}}$ defined similarly on the system $B^{\prime}$. It is
well known that every LOCC\ channel is a C-PPT-P channel \cite{Rai99,Rai01},
but the reverse containment does not hold. Thus,%
\begin{equation}
\text{LOCC\ }\subset\ \text{C-PPT-P}, \label{eq:LOCC-in-CPPTP}%
\end{equation}
as depicted in Figure~\ref{fig:Relax LOCC}.

Having defined this set of channels, we define the simulation error under
C-PPT-P\ channels as follows:%
\begin{multline}
e_{\operatorname{PPT}}(\mathcal{N}_{AB\rightarrow A^{\prime}B^{\prime}}%
,\rho_{\hat{A}\hat{B}})\coloneqq\label{eq:c-ppt-p-error}\\
\frac{1}{2}\inf_{\mathcal{P}\in\text{C-PPT-P}}\left\Vert \mathcal{N}%
_{AB\rightarrow A^{\prime}B^{\prime}}-\widetilde{\mathcal{N}}_{AB\rightarrow
A^{\prime}B^{\prime}}\right\Vert _{\diamond},
\end{multline}
where the optimization is with respect to C-PPT-P\ channels $\mathcal{P}%
_{AB\hat{A}\hat{B}\rightarrow A^{\prime}B^{\prime}}$ and%
\begin{equation}
\widetilde{\mathcal{N}}_{AB\rightarrow A^{\prime}B^{\prime}}(\omega
_{AB})\coloneqq \mathcal{P}_{AB\hat{A}\hat{B}\rightarrow A^{\prime}B^{\prime}%
}(\omega_{AB}\otimes\rho_{\hat{A}\hat{B}}).
\end{equation}
We note here that both of the optimization problems given in
\eqref{eq:sim-err-bipartite} and \eqref{eq:c-ppt-p-error}\ are special cases
of the optimization discussed in \cite[Section~II]{FWTB18}.

Furthermore, the following bound holds due to the containment in
\eqref{eq:LOCC-in-CPPTP}:%
\begin{equation}
e_{\operatorname{PPT}}(\mathcal{N}_{AB\rightarrow A^{\prime}B^{\prime}}%
,\rho_{AB})\leq e_{\operatorname{LOCC}}(\mathcal{N}_{AB\rightarrow A^{\prime
}B^{\prime}},\rho_{AB}). \label{eq:ppt-err-locc-err}%
\end{equation}

We now discuss how the error in \eqref{eq:c-ppt-p-error} can be computed by
means of a semi-definite program. To do so, we need to review a few key
concepts. First, recall from \cite{Wat09} that the normalized diamond distance
between channels $\mathcal{N}_{C\rightarrow D}$ and $\mathcal{M}_{C\rightarrow
D}$ can be computed by means of the following semi-definite program:%
\begin{multline}
\frac{1}{2}\left\Vert \mathcal{N}_{C\rightarrow D}-\mathcal{M}_{C\rightarrow
D}\right\Vert _{\diamond}=\label{eq:SDP-diamond}\\
\inf_{\mu\geq0,Z_{RD}\geq0}\left\{  \mu:\mu I_{R}\geq Z_{R},\ Z_{RD}\geq
\Gamma_{RD}^{\mathcal{N}}-\Gamma_{RD}^{\mathcal{M}}\right\}  ,
\end{multline}
where $\Gamma_{RD}^{\mathcal{N}}$ is the Choi operator of the channel
$\mathcal{N}_{C\rightarrow D}$, defined as%
\begin{align}
\Gamma_{RD}^{\mathcal{N}}  &  \coloneqq \mathcal{N}_{C\rightarrow D}%
(\Gamma_{RC}),\\
\Gamma_{RC}  &  \coloneqq |\Gamma\rangle\!\langle\Gamma|_{RC},\\
|\Gamma\rangle_{RC}  &  \coloneqq \sum_{i}|i\rangle_{R}|i\rangle_{C},
\end{align}
and system $R$ is isomorphic to the channel input system $A$. The Choi
operator $\Gamma_{RD}^{\mathcal{M}}$ is defined similarly. Thus, when
calculating the error in \eqref{eq:c-ppt-p-error}, we can employ this
semi-definite program, as well as the semi-definite constraints corresponding
to the optimization over C-PPT-P\ channels.

To arrive at the desired conclusion, let us recall some facts about quantum
channels and their Choi operators. A quantum channel has two properties:\ it
should be completely positive and trace preserving. The relation of these
properties to the Choi operator is given by the following:

\begin{enumerate}
\item a linear map $\mathcal{N}_{C\rightarrow D}$\ is completely positive if
and only if its Choi operator $\Gamma_{RD}^{\mathcal{N}}$ is positive
semi-definite, and

\item it is trace preserving if and only if $\operatorname{Tr}_{D}[\Gamma
_{RD}^{\mathcal{N}}]=I_{R}$.
\end{enumerate}

Each of these conditions is semi-definite and can be incorporated into an
optimization over C-PPT-P channels. The condition in
\eqref{eq:C-PPT-P-condition} corresponds to the Choi operator being PPT
\cite{Rai99,Rai01}; that is,

\begin{enumerate}
\item A bipartite channel satisfies \eqref{eq:C-PPT-P-condition} if and only
if its Choi operator $\Gamma_{ABA^{\prime}B^{\prime}}^{\mathcal{N}}$ satisfies
$T_{BB^{\prime}}(\Gamma_{ABA^{\prime}B^{\prime}}^{\mathcal{N}})\geq0$.
\end{enumerate}

Finally, suppose that we have a bipartite channel $\mathcal{N}_{AB\rightarrow
CD}$ and a channel $\mathcal{M}_{E\rightarrow A}$. Then the serial composition
of them leads to the channel%
\begin{equation}
\mathcal{R}_{EB\rightarrow CD}=\mathcal{N}_{AB\rightarrow CD}\circ
\mathcal{M}_{E\rightarrow A}.
\end{equation}
It is natural then to express the Choi operator of $\mathcal{R}_{EB\rightarrow
CD}$ in terms of those for $\mathcal{M}_{E\rightarrow A}$ and $\mathcal{N}%
_{AB\rightarrow CD}$. It is known that%
\begin{equation}
\Gamma_{EB\rightarrow CD}^{\mathcal{R}}=\operatorname{Tr}_{A}[T_{A}%
(\Gamma_{EA}^{\mathcal{M}})\Gamma_{ABCD}^{\mathcal{N}}],
\label{eq:Choi-op-comp}%
\end{equation}
where $T_{A}$ is the transpose map from \eqref{eq:transpose-map}.

Combining all of the above, we arrive at the following:

\begin{proposition}
\label{prop:gen-bi-SDP}The simulation error in \eqref{eq:c-ppt-p-error} can be
computed by means of the following semi-definite program:%
\begin{equation}
e_{\operatorname{PPT}}(\mathcal{N}_{AB\rightarrow A^{\prime}B^{\prime}}%
,\rho_{\hat{A}\hat{B}})=\inf_{\substack{\mu\geq0,Z_{ABA^{\prime}B^{\prime}%
}\geq0,\\P_{AB\hat{A}\hat{B}A^{\prime}B^{\prime}}\geq0}}\mu,
\label{eq:SDP-sim-error}%
\end{equation}
subject to%
\begin{align}
\mu I_{AB}  &  \geq Z_{AB},\\
T_{B\hat{B}B^{\prime}}(P_{AB\hat{A}\hat{B}A^{\prime}B^{\prime}})  &  \geq0,\\
\operatorname{Tr}_{A^{\prime}B^{\prime}}[P_{AB\hat{A}\hat{B}A^{\prime
}B^{\prime}}]  &  =I_{AB\hat{A}\hat{B}},\\
Z_{ABA^{\prime}B^{\prime}}  &  \geq\Gamma_{ABA^{\prime}B^{\prime}%
}^{\mathcal{N}}\nonumber\\
&  \qquad-\operatorname{Tr}_{\hat{A}\hat{B}}[T_{\hat{A}\hat{B}}(\rho_{\hat
{A}\hat{B}})P_{AB\hat{A}\hat{B}A^{\prime}B^{\prime}}].
\label{eq:main-SDP-Z-constraint-diamond-norm}%
\end{align}

\end{proposition}

The objective function and the first two constraints come from the
semi-definite program for the normalized diamond distance in
\eqref{eq:SDP-diamond}. The quantity $\operatorname{Tr}_{\hat{A}\hat{B}%
}[T_{\hat{A}\hat{B}}(\rho_{\hat{A}\hat{B}})P_{AB\hat{A}\hat{B}A^{\prime
}B^{\prime}}]$ comes about from the Choi operator of the composition of two
channels in \eqref{eq:Choi-op-comp}, while noting that a state $\rho_{\hat
{A}\hat{B}}$ is a particular kind of channel from a trivial system to the
systems $\hat{A}\hat{B}$. The third constraint comes from the fact that
$P_{AB\hat{A}\hat{B}A^{\prime}B^{\prime}}$ should be the Choi operator for a
C-PPT-P\ map, and the fourth constraint comes from the fact that $P_{AB\hat
{A}\hat{B}A^{\prime}B^{\prime}}$ should be the Choi operator for a
trace-preserving map.

As we show in Appendix~\ref{sec:SDP-dual-gen-bi}, the SDP\ dual to
\eqref{eq:SDP-sim-error} is as follows:%
\begin{equation}
\sup_{\substack{X_{AB}^{1},X_{ABA^{\prime}B^{\prime}}^{2},\\X_{AB\hat{A}%
\hat{B}A^{\prime}B^{\prime}}^{3}\geq0,\\W_{AB\hat{A}\hat{B}}\in
\operatorname{Herm}}}\operatorname{Tr}[\Gamma_{ABA^{\prime}B^{\prime}%
}^{\mathcal{N}}X_{ABA^{\prime}B^{\prime}}^{2}]-\operatorname{Tr}[W_{AB\hat
{A}\hat{B}}], \label{eq:SDP-dual}%
\end{equation}
subject to%
\begin{equation}
\operatorname{Tr}[X_{AB}^{1}]\leq1,\qquad X_{ABA^{\prime}B^{\prime}}^{2}\leq
X_{AB}^{1}\otimes I_{A^{\prime}B^{\prime}},
\end{equation}%
\begin{multline}
X_{ABA^{\prime}B^{\prime}}^{2}\otimes T_{\hat{A}\hat{B}}(\rho_{\hat{A}\hat{B}%
})+T_{B\hat{B}B^{\prime}}(X_{AB\hat{A}\hat{B}A^{\prime}B^{\prime}}^{3})\\
\leq W_{AB\hat{A}\hat{B}}\otimes I_{A^{\prime}B^{\prime}}.
\end{multline}

\subsection{Semi-definite programming lower bound on the simulation error of
unideal bidirectional teleportation\label{sec:SDP-lower-bound-BTP}}

The semi-definite program in Proposition~\ref{prop:gen-bi-SDP} can be
evaluated for our bipartite channel of interest, i.e., the unitary swap
channel $\mathcal{S}_{AB}^{d}$ in \eqref{eq:ideal-swap-channel}. Even though
the semi-definite program is efficiently computable with respect to the
dimensions of systems $A$, $B$, $A^{\prime}$, $B^{\prime}$, $\hat{A}$, and
$\hat{B}$, one finds that it can take some time in practice when evaluated for
states of interest. Thus, we are interested in ways of reducing its complexity.

To reduce its complexity, we recall that the unitary swap channel
$\mathcal{S}_{AB}^{d}$ in \eqref{eq:ideal-swap-channel} has the following
symmetry:%
\begin{equation}
\mathcal{S}_{AB}^{d}\left(  \mathcal{U}_{A}\otimes\mathcal{V}_{B}\right)
=\left(  \mathcal{V}_{A}\otimes\mathcal{U}_{B}\right)  \mathcal{S}_{AB}^{d},
\label{eq:swap-ch-symmetry}%
\end{equation}
where $\mathcal{U}_{A}$ and $\mathcal{V}_{B}$ are arbitrary unitary channels.
Furthermore, it commutes with itself, i.e.,%
\begin{equation}
\mathcal{S}_{AB}^{d}\circ\mathcal{S}_{AB}^{d}=\mathcal{S}_{AB}^{d}%
\circ\mathcal{S}_{AB}^{d}. \label{eq:swap-ch-add-symmetry}%
\end{equation}
By exploiting these symmetries, we arrive at a semi-definite program for
evaluating $e_{\operatorname{PPT}}(\mathcal{S}_{AB}^{d},\rho_{\hat{A}\hat{B}%
})$, which has significantly lower complexity than the generic semi-definite
program in Proposition~\ref{prop:gen-bi-SDP}. In particular, its complexity is
polynomial in the dimension of the systems $\hat{A}$ and $\hat{B}$. We provide
a proof of Proposition~\ref{prop:swap-sdp-simplify}\ in
Appendix~\ref{app:SDP-simplify-swap}.

\begin{proposition}
\label{prop:swap-sdp-simplify}The semi-definite program in
Proposition~\ref{prop:gen-bi-SDP},\ for the error in simulating the unitary
SWAP channel $\mathcal{S}_{AB}^{d}$ in \eqref{eq:ideal-swap-channel},
simplifies as follows:%
\begin{equation}
e_{\operatorname{PPT}}(\mathcal{S}_{AB}^{d},\rho_{\hat{A}\hat{B}}%
)=1-\sup_{\substack{K_{\hat{A}\hat{B}},L_{\hat{A}\hat{B}},\\N_{\hat{A}\hat{B}%
}\geq0}}\operatorname{Tr}[\rho_{\hat{A}\hat{B}}K_{\hat{A}\hat{B}}],
\label{eq:main-SDP-paper}%
\end{equation}
subject to
\begin{align}
T_{\hat{B}}\!\left(  K_{\hat{A}\hat{B}}+\frac{L_{\hat{A}\hat{B}}}{d+1}%
+\frac{N_{\hat{A}\hat{B}}}{\left(  d+1\right)  ^{2}}\right)   &
\geq0,\nonumber\\
\frac{1}{d^{2}-1}T_{\hat{B}}\!\left(  L_{\hat{A}\hat{B}}+N_{\hat{A}\hat{B}%
}\right)   &  \geq T_{\hat{B}}\!\left(  K_{\hat{A}\hat{B}}\right)
,\nonumber\\
T_{\hat{B}}\!\left(  K_{\hat{A}\hat{B}}+\frac{N_{\hat{A}\hat{B}}}{\left(
d-1\right)  ^{2} }\right)   &  \geq\frac{1}{d-1}T_{\hat{B}}\!\left(
L_{\hat{A}\hat{B}}\right)  ,\nonumber\\
K_{\hat{A}\hat{B}}+L_{\hat{A}\hat{B}}+N_{\hat{A}\hat{B}}  &  =I_{\hat{A}%
\hat{B}}.
\end{align}

\end{proposition}

\begin{remark}
\label{rem:channel-form}The proof of Proposition~\ref{prop:swap-sdp-simplify}%
\ demonstrates that an optimal C-PPT-P\ channel for simulating the unitary
swap channel has the following structure:%
\begin{multline}
\mathcal{P}_{AB\hat{A}\hat{B}\rightarrow AB}(\omega_{AB}\otimes\rho_{\hat
{A}\hat{B}})=\mathcal{S}_{AB}^{d}(\omega_{AB})\operatorname{Tr}[K_{\hat{A}%
\hat{B}}\rho_{\hat{A}\hat{B}}]\\
+\frac{1}{2}\left(
\begin{array}
[c]{c}%
\operatorname{id}_{A\rightarrow B}\otimes\mathcal{D}_{B\rightarrow A}\\
+\mathcal{D}_{A\rightarrow B}\otimes\operatorname{id}_{B\rightarrow A}%
\end{array}
\right)  (\omega_{AB})\operatorname{Tr}[L_{\hat{A}\hat{B}}\rho_{\hat{A}\hat
{B}}]\\
+\left(  \mathcal{D}_{A\rightarrow B}\otimes\mathcal{D}_{B\rightarrow
A}\right)  (\omega_{AB})\operatorname{Tr}[N_{\hat{A}\hat{B}}\rho_{\hat{A}%
\hat{B}}],
\end{multline}
where $\mathcal{D}$ denotes the following generalized Pauli channel:%
\begin{equation}
\mathcal{D}(\sigma)\coloneqq\frac{1}{d^{2}-1}\sum_{\left(  x,z\right)
\neq\left(  0,0\right)  }W^{z,x}\sigma(W^{z,x})^{\dag}%
.\label{eq:gen-Pauli-channel}%
\end{equation}
Thus, the interpretation of the simulating channel is that it measures the
resource state $\rho_{\hat{A}\hat{B}}$ according to the POVM\ $\{K_{\hat
{A}\hat{B}},L_{\hat{A}\hat{B}},N_{\hat{A}\hat{B}}\}$, which is subject to the
inequality constraints in Proposition~\ref{prop:swap-sdp-simplify}. After
that, it takes the following action:

\begin{enumerate}
\item If the first outcome $K_{\hat{A}\hat{B}}$ occurs, then apply the ideal
swap channel to the input state $\omega_{AB}$.

\item If the second outcome $L_{\hat{A}\hat{B}}$ occurs, then with probability
1/2, apply the identity channel $\operatorname{id}_{A\rightarrow B}$ to
transfer Alice's input system $A$ to Bob, but then garble Bob's input system
$B$ by applying the channel $\mathcal{D}$ and transfer the resulting system to
Alice; with probability 1/2, apply the identity channel $\operatorname{id}%
_{B\rightarrow A}$ to transfer Bob's input system $B$ to Alice, but then
garble Alice's input system $A$ by applying the channel $\mathcal{D}$ and
transfer the resulting system to Bob.

\item If the third outcome $N_{\hat{A}\hat{B}}$ occurs, then apply the
garbling channel $\mathcal{D}$ to both Alice and Bob's systems individually
and exchange them.
\end{enumerate}

\noindent The fact that the measurement operators obey the inequality
constraints in Proposition~\ref{prop:swap-sdp-simplify} implies that the
quantum channel $\mathcal{P}_{AB\hat{A}\hat{B}\rightarrow AB}$ is C-PPT-P.
\end{remark}

\subsection{Semi-definite programming lower bounds when using channel
infidelity}

\label{sec:SDP-ch-infid}As mentioned at the end of
Section~\ref{sec:LOCC-sim-gen-bi}, we can also employ infidelity to quantify
the simulation error. In this section, we briefly detail the semi-definite
programming lower bound that results when using the infidelity to measure
simulation error.

First, consider that the simulation error is defined as follows, when
optimizing over C-PPT-P channels and using the infidelity error measure:%
\begin{multline}
e_{\operatorname{PPT}}^{F}(\mathcal{N}_{AB\rightarrow A^{\prime}B^{\prime}%
},\rho_{\hat{A}\hat{B}})\coloneqq\label{eq:infid-sim-err-CPPTP}\\
\inf_{\mathcal{P}\in\text{C-PPT-P}}1-F(\mathcal{N}_{AB\rightarrow A^{\prime
}B^{\prime}},\widetilde{\mathcal{N}}_{AB\rightarrow A^{\prime}B^{\prime}}),
\end{multline}
where the optimization is with respect to C-PPT-P\ channels $\mathcal{P}%
_{AB\hat{A}\hat{B}\rightarrow A^{\prime}B^{\prime}}$ and%
\begin{equation}
\widetilde{\mathcal{N}}_{AB\rightarrow A^{\prime}B^{\prime}}(\omega
_{AB})\coloneqq \mathcal{P}_{AB\hat{A}\hat{B}\rightarrow A^{\prime}B^{\prime}%
}(\omega_{AB}\otimes\rho_{\hat{A}\hat{B}}).
\end{equation}
It is clear that%
\begin{equation}
e_{\operatorname{PPT}}^{F}(\mathcal{N}_{AB\rightarrow A^{\prime}B^{\prime}%
},\rho_{\hat{A}\hat{B}})\leq e_{\operatorname{LOCC}}^{F}(\mathcal{N}%
_{AB\rightarrow A^{\prime}B^{\prime}},\rho_{\hat{A}\hat{B}}),
\end{equation}
for the same reason that \eqref{eq:ppt-err-locc-err} holds.

Then we find the following result, from combining the semi-definite program
for channel fidelity in \eqref{eq:SDP-ch-fid-1}--\eqref{eq:SDP-ch-fid-3},
along with reasoning similar to that used to justify
Proposition~\ref{prop:gen-bi-SDP}:

\begin{proposition}
\label{prop:gen-bi-SDP-infidelity}The simulation error in
\eqref{eq:infid-sim-err-CPPTP} can be computed by means of the following
semi-definite program:%
\begin{multline}
e_{\operatorname{PPT}}^{F}(\mathcal{N}_{AB\rightarrow A^{\prime}B^{\prime}%
},\rho_{\hat{A}\hat{B}})=\\
1-\left[  \sup_{\lambda\geq0,P_{AB\hat{A}\hat{B}A^{\prime}B^{\prime}}%
\geq0,Q_{ABA^{\prime}B^{\prime}}}\lambda\right]  ^{2},
\end{multline}
subject to%
\begin{align}
\lambda I_{AB}  &  \leq\operatorname{Re}[\operatorname{Tr}_{A^{\prime
}B^{\prime}}[Q_{ABA^{\prime}B^{\prime}}]],\\
T_{B\hat{B}B^{\prime}}(P_{AB\hat{A}\hat{B}A^{\prime}B^{\prime}})  &  \geq0,\\
\operatorname{Tr}_{A^{\prime}B^{\prime}}[P_{AB\hat{A}\hat{B}A^{\prime
}B^{\prime}}]  &  =I_{AB\hat{A}\hat{B}},
\end{align}%
\begin{equation}%
\begin{bmatrix}
\Gamma_{ABA^{\prime}B^{\prime}}^{\mathcal{N}} & Q_{ABA^{\prime}B^{\prime}%
}^{\dag}\\
Q_{ABA^{\prime}B^{\prime}} & \operatorname{Tr}_{\hat{A}\hat{B}}[T_{\hat{A}%
\hat{B}}(\rho_{\hat{A}\hat{B}})P_{AB\hat{A}\hat{B}A^{\prime}B^{\prime}}]
\end{bmatrix}
\geq0.
\end{equation}

\end{proposition}

\begin{remark}
There are important implications of
Proposition~\ref{prop:gen-bi-SDP-infidelity}\ beyond the problems considered
in this paper and which are relevant for quantum resource theories of channels
(see, e.g., \cite{LW19,LY19,Wang2019a}). In prior work, the normalized diamond
distance was used to measure simulation error
\cite{FWTB18,LW19,LY19,Wang2019a}, and if the class of channels being
considered for the simulation obeys semi-definite constraints, then it follows
that the simulation error can be computed by means of a semi-definite program.
What Proposition~\ref{prop:gen-bi-SDP-infidelity}\ demonstrates is that the
same is true when using channel infidelity for the simulation error. As a key
example, Proposition~\ref{prop:gen-bi-SDP-infidelity}\ demonstrates that the
channel box transformation optimization problem from \cite{Wang2019a}\ can be
computed by means of a semi-definite program when using channel infidelity for
approximation error. We show this explicitly in
Appendix~\ref{app:channel-box-transform-infid}.
\end{remark}

Applying similar reasoning used to arrive at
Propositions~\ref{prop:err-collapse-LOCC} and \ref{prop:swap-sdp-simplify}, we
find that the PPT\ simulation error of the unitary swap channel is equal to
the expression from Proposition~\ref{prop:swap-sdp-simplify}.

\begin{proposition}
\label{prop:sim-err-infid-swap-simplify}The semi-definite program in
Proposition~\ref{prop:gen-bi-SDP-infidelity},\ for the error in simulating the
unitary SWAP channel $\mathcal{S}_{AB}^{d}$ in \eqref{eq:ideal-swap-channel},
simplifies to the expression from Proposition~\ref{prop:swap-sdp-simplify}.
That is,
\begin{equation}
e_{\operatorname{PPT}}^{F}(\mathcal{S}_{AB}^{d},\rho_{\hat{A}\hat{B}%
})=e_{\operatorname{PPT}}(\mathcal{S}_{AB}^{d},\rho_{\hat{A}\hat{B}}).
\end{equation}

\end{proposition}

As a consequence of Proposition~\ref{prop:sim-err-infid-swap-simplify}, there
is no need for different notions of simulation error when considering the
simulation of the unitary swap channel using C-PPT-P channels.

\section{Examples}

\label{sec:examples}In this section, we consider several examples of resource
states that can be used for bidirectional teleportation, and we evaluate their
performance. For several cases of interest, we establish an exact evaluation
not only for the error when using a PPT\ simulation, but also when using an LOCC\ simulation.

One key example resource state, considered in Section~\ref{sec:no-res-state}%
\ below, is when there is in fact no resource state at all. In such a
situation, Alice and Bob can only employ a PPT or LOCC simulation of
bidirectional teleportation. In doing so, they could prepare a PPT\ or
separable resource state for free, respectively, by means of a C-PPT-P\ or
LOCC channel, and so this situation also captures the case in which the
resource state that they share is a PPT\ or separable state, respectively. We
find that the simulation error in both cases is equal to $1-\frac{1}{d^{2}}$,
where $d$ is the dimension of the swap channel $\mathcal{S}_{AB}^{d}$ to be
simulated. The importance of this result is that it establishes a dividing
line between a classical and quantum implementation of bidirectional
teleportation, which can be used by experimentalists to assess the performance
of an implementation of bidirectional teleportation. That is, in the case that
Alice and Bob share no resource state or in the case that they share a
separable state, the simulation error that they can achieve is no smaller than
$1-\frac{1}{d^{2}}$. If they share an entangled state, then it is possible for
the simulation error to be smaller than this.

Other resource states that we consider are isotropic \cite{Horodecki99}\ and
Werner \cite{W89}\ states, which are states that are simply characterized by a
single parameter, due to the large amount of symmetry that they possess. We
consider these states in Sections~\ref{sec:isotropic-example}\ and
\ref{sec:werner-example}, respectively.

We note here that we arrived at the analytical conclusions in
Sections~\ref{sec:no-res-state}, \ref{sec:isotropic-example},\ and
\ref{sec:werner-example} by employing Matlab and Mathematica. We used Matlab
to implement the linear programs numerically, whose numerical solutions were
subsequently used to infer analytical solutions through the use of
Mathematica. We have included all of these files with the arXiv posting of our
paper. The Mathematica files are especially useful for verifying that the
values of primal and dual variables are feasible for the linear programs given.

\subsection{No resource state:\ Benchmark for classical versus quantum
bidirectional teleportation\label{sec:no-res-state}}

We begin by considering the scenario in which there is no resource state at
all. As mentioned above, the utility of this scenario is that it provides a
dividing line between a classical and quantum implementation of bidirectional
teleportation. The proof is sufficiently short that we provide it below.
Figure~\ref{fig:noResourceError} plots the expression in
\eqref{eq:err-no-res-thm} for the simulation error.

\begin{proposition}
\label{prop:no-res-sim-err}If there is no resource state, then the error in
simulating the unitary SWAP channel $\mathcal{S}_{AB}^{d}$ in
\eqref{eq:ideal-swap-channel} is equal to $1-1/d^{2}$:%
\begin{equation}
e_{\operatorname{PPT}}(\mathcal{S}_{AB}^{d},\emptyset)=e_{\operatorname{LOCC}%
}(\mathcal{S}_{AB}^{d},\emptyset)=1-\frac{1}{d^{2}}, \label{eq:err-no-res-thm}%
\end{equation}
where the notation $\emptyset$ indicates the absence of a resource state.
\end{proposition}

\begin{figure}[ptb]
\begin{center}
\includegraphics[
width=\linewidth
]{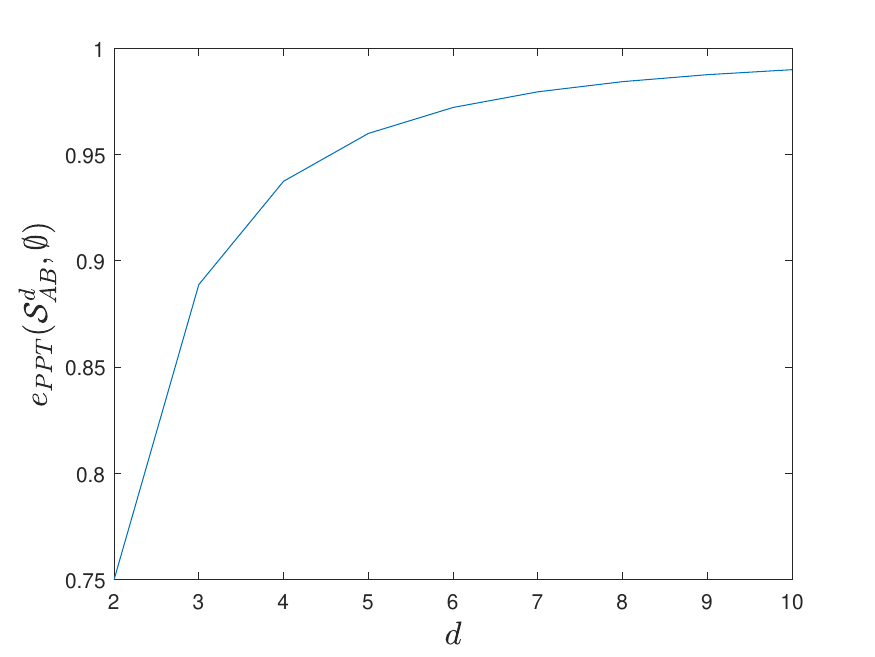}
\end{center}
\caption{Plot of the simulation error of bidirectional teleportation when
using no quantum resource state or, equivalently when LOCC is free, a
separable state. The quantity $d$ is the dimension of the SWAP channel.}%
\label{fig:noResourceError}%
\end{figure}

\begin{proof}
If there is no resource state, this is equivalent to the state $\rho_{\hat
{A}\hat{B}}$ simply being equal to the number one, and the three operators
$K_{\hat{A}\hat{B}}$, $L_{\hat{A}\hat{B}}$, and $N_{\hat{A}\hat{B}}$, from
Proposition~\ref{prop:swap-sdp-simplify}, reduce to real numbers $p_{1}$,
$p_{2}$, and $p_{3}$. Then the semi-definite program from
Proposition~\ref{prop:swap-sdp-simplify}\ simplifies to the following linear
program:%
\begin{equation}
1-\sup_{p_{1},p_{2},p_{3}\geq0}p_{1}, \label{eq:LP-no-res-primal}%
\end{equation}
subject to%
\begin{align}
p_{1}+\frac{p_{2}}{d+1}+\frac{p_{3}}{\left(  d+1\right)  ^{2}}  &
\geq0,\label{eq:redundant-constraint}\\
\frac{1}{d^{2}-1}\left(  p_{2}+p_{3}\right)   &  \geq p_{1},\\
p_{1}+\frac{p_{3}}{\left(  d-1\right)  ^{2}}  &  \geq\frac{p_{2}}{d-1},\\
p_{1}+p_{2}+p_{3}  &  =1. \label{eq:LP-no-res-primal-last}%
\end{align}
Note that the first inequality constraint is redundant (a trivial consequence
of $p_{1},p_{2},p_{3}\geq0$), and so it can be eliminated. A feasible point of
the linear program above is given by%
\begin{equation}
p_{1}=\frac{1}{d^{2}},\qquad p_{2}=0,\qquad p_{3}=1-\frac{1}{d^{2}},
\end{equation}
implying that%
\begin{equation}
e_{\operatorname{PPT}}(\mathcal{S}_{AB}^{d},\emptyset)\leq1-\frac{1}{d^{2}}.
\label{eq:no-res-upper-bnd}%
\end{equation}

To find a matching lower bound, we determine the linear program dual to that
in \eqref{eq:LP-no-res-primal}. Consider that
\eqref{eq:LP-no-res-primal}--\eqref{eq:LP-no-res-primal-last}\ can be written
in the standard form of a linear program as follows \cite{BV04}:%
\begin{equation}
1-\sup_{x\geq0}\left\{  c^{T}x:Ax\leq b\right\}  ,
\label{eq:LP-no-res-primal-standard-form}%
\end{equation}
where%
\begin{align}
x^{T}  &  =%
\begin{bmatrix}
p_{1} & p_{2} & p_{3}%
\end{bmatrix}
,\label{eq:LP-no-res-standard-1}\\
c^{T}  &  =%
\begin{bmatrix}
1 & 0 & 0
\end{bmatrix}
,\\
A  &  =%
\begin{bmatrix}
1 & -\frac{1}{d^{2}-1} & -\frac{1}{d^{2}-1}\\
-1 & \frac{1}{d-1} & -\frac{1}{\left(  d-1\right)  ^{2}}\\
1 & 1 & 1\\
-1 & -1 & -1
\end{bmatrix}
,\\
b^{T}  &  =%
\begin{bmatrix}
0 & 0 & 1 & -1
\end{bmatrix}
. \label{eq:LP-no-res-standard-last}%
\end{align}
Note that, in deriving the matrix $A$, we have eliminated the redundant
constraint in \eqref{eq:redundant-constraint}. The dual of this linear program
is given by \cite{BV04}%
\begin{equation}
1-\inf_{y\geq0}\left\{  b^{T}y:A^{T}y\geq c\right\}  .
\label{eq:LP-dual-no-res}%
\end{equation}
Due to weak duality \cite{BV04}, the optimal value of
\eqref{eq:LP-dual-no-res} does not exceed the optimal value of
\eqref{eq:LP-no-res-primal-standard-form}. Writing
\eqref{eq:LP-dual-no-res}\ out in detail using the definitions in
\eqref{eq:LP-no-res-standard-1}--\eqref{eq:LP-no-res-standard-last}\ gives%
\begin{equation}
1-\inf_{y_{1},\ldots,y_{4}\geq0}y_{3}-y_{4}%
\end{equation}
subject to%
\begin{align}
y_{1}-y_{2}+y_{3}-y_{4}  &  \geq1,\\
-\frac{y_{1}}{d^{2}-1}+\frac{y_{2}}{d-1}+y_{3}-y_{4}  &  \geq0,\\
-\frac{y_{1}}{d^{2}-1}+\frac{y_{2}}{d-1}+y_{3}-y_{4}  &  \geq0,\\
-\frac{y_{1}}{d^{2}-1}-\frac{y_{2}}{\left(  d-1\right)  ^{2}}+y_{3}-y_{4}  &
\geq0.
\end{align}
A solution is given by%
\begin{equation}
y_{1}=1-\frac{1}{d^{2}},\quad y_{2}=y_{4}=0,\quad y_{3}=\frac{1}{d^{2}},
\end{equation}
implying that%
\begin{equation}
e_{\operatorname{PPT}}(\mathcal{S}_{AB}^{d},\emptyset)\geq1-\frac{1}{d^{2}}.
\label{eq:no-res-lower-bnd}%
\end{equation}
Putting together \eqref{eq:no-res-upper-bnd} and \eqref{eq:no-res-lower-bnd},
we conclude the equality%
\begin{equation}
e_{\operatorname{PPT}}(\mathcal{S}_{AB}^{d},\emptyset)=1-\frac{1}{d^{2}}.
\end{equation}

By applying the inequality in \eqref{eq:ppt-err-locc-err}, we conclude that%
\begin{equation}
e_{\operatorname{PPT}}(\mathcal{S}_{AB}^{d},\emptyset)\leq
e_{\operatorname{LOCC}}(\mathcal{S}_{AB}^{d},\emptyset).
\end{equation}
Thus, it remains to establish an upper bound on $e_{\operatorname{LOCC}%
}(\mathcal{S}_{AB}^{d},\emptyset)$, by demonstrating a scheme for
bidirectional teleportation that achieves the simulation error $1-\frac
{1}{d^{2}}$.

A scheme to achieve this error using LOCC\ consists of Alice and Bob preparing
the $d^{2}$-dimensional states $|0\rangle\!\langle0|_{\hat{A}}\otimes
|0\rangle\!\langle0|_{\hat{B}}$, applying the bilateral twirl in
\eqref{eq:bilat-twirl}\ to get the state%
\begin{align}
\omega_{\hat{A}\hat{B}}  &  \coloneqq \widetilde{\mathcal{T}}_{\hat{A}\hat{B}%
}(|0\rangle\!\langle0|_{\hat{A}}\otimes|0\rangle\!\langle0|_{\hat{B}})\\
&  =\frac{1}{d^{2}}\Phi_{\hat{A}\hat{B}}^{d^{2}}+\left(  1-\frac{1}{d^{2}%
}\right)  \frac{\left(  I_{\hat{A}\hat{B}}-\Phi_{\hat{A}\hat{B}}^{d^{2}%
}\right)  }{d^{4}-1},
\end{align}
separating out $\Phi_{\hat{A}\hat{B}}^{d^{2}}$ to two e-dits $\Phi_{\hat
{A}_{1}\hat{B}_{1}}^{d}\otimes\Phi_{\hat{A}_{2}\hat{B}_{2}}^{d}$ via local
isometries, and then performing teleportation in opposite directions.

Let us now prove that this simulation achieves the simulation error
$1-\frac{1}{d^{2}}$. Consider that the action of the bidirectional
teleportation operations is equivalent to an LOCC\ channel $\mathcal{L}%
_{AB\hat{A}\hat{B}\rightarrow AB}$. Furthermore, we apply the same LOCC
channel $\mathcal{L}_{AB\hat{A}\hat{B}\rightarrow AB}$ regardless of whether
we are conducting ideal or unideal bidirectional teleportation. So it follows
that ideal bidirectional teleportation is given by%
\begin{equation}
\mathcal{S}_{AB}^{d}(\cdot)=\mathcal{L}_{AB\hat{A}\hat{B}\rightarrow
AB}((\cdot)\otimes\Phi_{\hat{A}\hat{B}}^{d^{2}}),
\end{equation}
and the channel realized by unideal bidirectional teleportation is%
\begin{equation}
\widetilde{\mathcal{S}}_{AB}^{d}(\cdot)=\mathcal{L}_{AB\hat{A}\hat
{B}\rightarrow AB}((\cdot)\otimes\omega_{\hat{A}\hat{B}}).
\end{equation}
Let $\psi_{RAB}$ be an arbitrary input state for these channels. Then we find
that%
\begin{align}
&  \frac{1}{2}\left\Vert \mathcal{S}_{AB}^{d}(\psi_{RAB})-\widetilde
{\mathcal{S}}_{AB}^{d}(\psi_{RAB})\right\Vert _{1}
\label{eq:err-eval-no-res-1}\\
&  =\frac{1}{2}\left\Vert \mathcal{L}_{AB\hat{A}\hat{B}\rightarrow AB}%
(\psi_{RAB}\otimes(\Phi_{\hat{A}\hat{B}}^{d^{2}}-\omega_{\hat{A}\hat{B}%
}))\right\Vert _{1}\\
&  \leq\frac{1}{2}\left\Vert \psi_{RAB}\otimes(\Phi_{\hat{A}\hat{B}}^{d^{2}%
}-\omega_{\hat{A}\hat{B}})\right\Vert _{1}\\
&  =\frac{1}{2}\left\Vert \Phi_{\hat{A}\hat{B}}^{d^{2}}-\omega_{\hat{A}\hat
{B}}\right\Vert _{1}.
\end{align}
The inequality follows from the data-processing inequality for trace distance
under the action of the quantum channel $\mathcal{L}_{AB\hat{A}\hat
{B}\rightarrow AB}$. The final equality follows from the multiplicativity of
the trace norm and the fact that it is equal to one for the quantum state
$\psi_{RAB}$. Since the state $\psi_{RAB}$ is arbitrary, we conclude that%
\begin{equation}
\frac{1}{2}\left\Vert \mathcal{S}_{AB}^{d}-\widetilde{\mathcal{S}}_{AB}%
^{d}\right\Vert _{\diamond}\leq\frac{1}{2}\left\Vert \Phi_{\hat{A}\hat{B}%
}^{d^{2}}-\omega_{\hat{A}\hat{B}}\right\Vert _{1},
\end{equation}
by applying the equality in \eqref{eq:d-dist-pure-states}. Continuing, we find
that%
\begin{align}
&  \frac{1}{2}\left\Vert \Phi_{\hat{A}\hat{B}}^{d^{2}}-\omega_{\hat{A}\hat{B}%
}\right\Vert _{1}\nonumber\\
&  =\frac{1}{2}\left\Vert \Phi_{\hat{A}\hat{B}}^{d^{2}}-\left(  \frac{1}%
{d^{2}}\Phi_{\hat{A}\hat{B}}^{d^{2}}+\left(  1-\frac{1}{d^{2}}\right)
\frac{\left(  I_{\hat{A}\hat{B}}-\Phi_{\hat{A}\hat{B}}^{d^{2}}\right)  }%
{d^{4}-1}\right)  \right\Vert _{1}\nonumber\\
&  =\frac{1}{2}\left\Vert \left(  1-\frac{1}{d^{2}}\right)  \Phi_{\hat{A}%
\hat{B}}^{d^{2}}-\left(  1-\frac{1}{d^{2}}\right)  \frac{\left(  I_{\hat
{A}\hat{B}}-\Phi_{\hat{A}\hat{B}}^{d^{2}}\right)  }{d^{4}-1}\right\Vert
_{1}\nonumber\\
&  =1-\frac{1}{d^{2}}.
\end{align}
The final equality follows because $\Phi_{\hat{A}\hat{B}}^{d^{2}}$ is
orthogonal to $I_{\hat{A}\hat{B}}-\Phi_{\hat{A}\hat{B}}^{d^{2}}$, and both
operators are positive semi-definite, so that the trace norm is equal to the
sum of the traces of the individual operators. Thus, we have proven that%
\begin{equation}
\frac{1}{2}\left\Vert \mathcal{S}_{AB}^{d}-\widetilde{\mathcal{S}}_{AB}%
^{d}\right\Vert _{\diamond}\leq1-\frac{1}{d^{2}},
\label{eq:err-eval-no-res-last}%
\end{equation}
establishing an upper bound on the LOCC\ simulation error that matches the
lower bound in \eqref{eq:no-res-lower-bnd}.
\end{proof}

\subsection{Isotropic states}

\label{sec:isotropic-example}A general class of bipartite states of interest
in quantum information is the class of isotropic states \cite{Horodecki99}. An
isotropic state of fidelity $F\in\left[  0,1\right]  $ and dimension
$d_{\hat{A}}\in\left\{  2,3,4,\ldots\right\}  $ is defined as follows:%
\begin{equation}
\rho_{\hat{A}\hat{B}}^{(F,d_{\hat{A}})}\coloneqq F\Phi_{\hat{A}\hat{B}%
}+\left(  1-F\right)  \frac{I_{\hat{A}\hat{B}}-\Phi_{\hat{A}\hat{B}}}%
{d_{\hat{A}}^{2}-1}. \label{eq:isotropic-def}%
\end{equation}
The general interest in isotropic states stems from the fact that an arbitrary
state of systems $\hat{A}\hat{B}$ can be twirled to an isotropic state, which
follows from an application of \eqref{eq:simple-formula-bilat-twirl}.

The following proposition establishes a simple expression for the simulation
error when using an isotropic state for bidirectional teleportation. It is
given exclusively in terms of the dimension $d$ of the swap channel that is
being simulated and the two parameters $F$ and $d_{\hat{A}}$ that characterize
the isotropic resource state. A proof is available in
Appendix~\ref{app:isotropic-proof}. The proof exploits the symmetries of an
isotropic state to reduce a variation of the semi-definite program in
Proposition~\ref{prop:swap-sdp-simplify}\ to a linear program, which we then
solve analytically.

\begin{proposition}
\label{prop:isotropic-sim-perf}The simulation error for the unitary swap
channel when using an isotropic resource state $\rho_{\hat{A}\hat{B}%
}^{(F,d_{\hat{A}})}$ is as follows:%
\begin{multline}
\label{eq:isotropic-error}e_{\operatorname{PPT}}(\mathcal{S}_{AB}^{d}%
,\rho_{\hat{A}\hat{B}}^{(F,d_{\hat{A}})})=\\
\left\{
\begin{array}
[c]{cc}%
1-\frac{1}{d^{2}} & \text{if }F\leq\frac{1}{d_{\hat{A}}}\\
1-\frac{Fd_{\hat{A}}}{d^{2}} & \text{if }F>\frac{1}{d_{\hat{A}}}\text{ and
}d_{\hat{A}}\leq d^{2}\\
\frac{\left(  1-\frac{1}{d^{2}}\right)  \left(  1-F\right)  }{1-\frac
{1}{d_{\hat{A}}}} & \text{if }F>\frac{1}{d_{\hat{A}}}\text{ and }d_{\hat{A}%
}>d^{2}%
\end{array}
\right.  .
\end{multline}
We also have that%
\begin{equation}
e_{\operatorname{PPT}}(\mathcal{S}_{AB}^{d},\rho_{\hat{A}\hat{B}}%
^{(F,d_{\hat{A}})})=e_{\operatorname{LOCC}}(\mathcal{S}_{AB}^{d},\rho_{\hat
{A}\hat{B}}^{(F,d_{\hat{A}})}) \label{eq:ppt-err-equals-locc-err-iso}%
\end{equation}
if $F\leq\frac{1}{d_{\hat{A}}}$ or if $F>\frac{1}{d_{\hat{A}}}$ and
$d_{\hat{A}}\leq d^{2}$.
\end{proposition}

It remains open to determine if the equality in
\eqref{eq:ppt-err-equals-locc-err-iso} holds when $F>\frac{1}{d_{\hat{A}}}$
and $d_{\hat{A}}>d^{2}$. The main question is to design an LOCC-assisted
scheme that achieves a simulation error of $(1-\frac{1}{d^{2}})\left(
1-F\right)  /(1-\frac{1}{d_{\hat{A}}})$ when $F>\frac{1}{d_{\hat{A}}}$ and
$d_{\hat{A}}>d^{2}$.

\begin{figure}[ptb]
\begin{center}
\includegraphics[
width=\linewidth
]{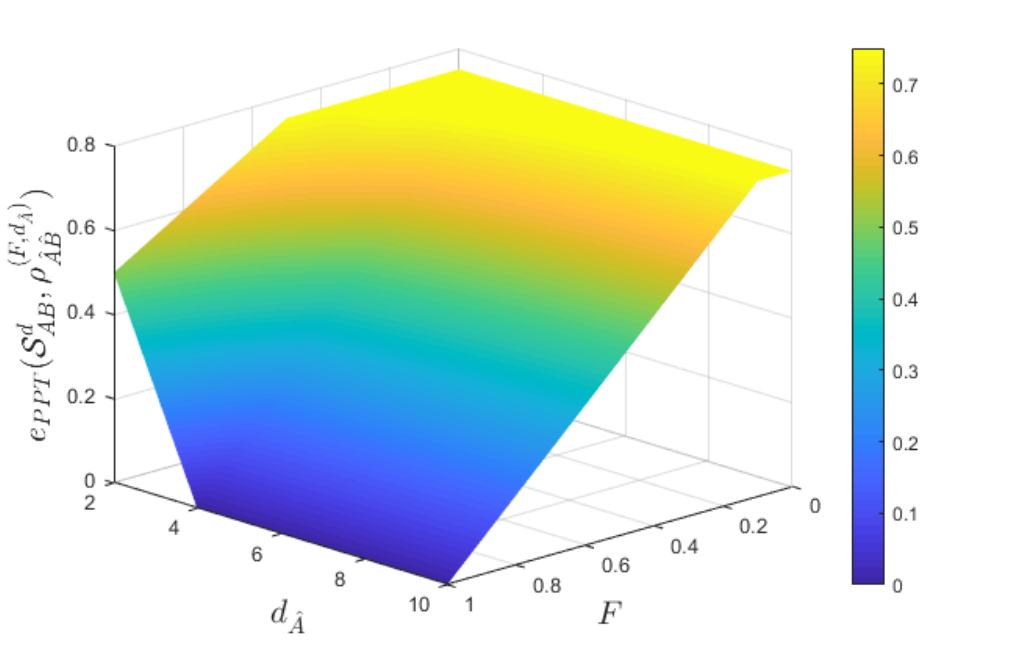}
\end{center}
\caption{Plot of the simulation error of bidirectional teleportation when
using the isotropic resource state defined in $\eqref{eq:isotropic-def}$ where
$F$ is the fidelity parameter and $d_{\hat{A}}$ is the dimension of Alice's
system of the resource state.}%
\label{fig:isotropicError}%
\end{figure}

Figure~\ref{fig:isotropicError} plots the expression given in
\eqref{eq:isotropic-error} for the simulation error.

\begin{remark}
[Optimal strategy with a single e-dit]\label{rem:single-edit-strat}A special
case of Proposition~\ref{prop:isotropic-sim-perf} above occurs when Alice and
Bob share a single e-dit, so that $F=1$ and $d_{\hat{A}}=d$. In this case, the
following equality holds%
\begin{equation}
e_{\operatorname{PPT}}(\mathcal{S}_{AB}^{d},\rho_{\hat{A}\hat{B}}%
^{(1,d)})=e_{\operatorname{LOCC}}(\mathcal{S}_{AB}^{d},\rho_{\hat{A}\hat{B}%
}^{(1,d)})=1-\frac{1}{d}.
\end{equation}
An optimal protocol to achieve this simulation error is for Alice and Bob to
embed the the resource state $\rho_{\hat{A}\hat{B}}^{(1,d)}=\Phi_{\hat{A}%
\hat{B}}^{d}$ in larger Hilbert spaces, each of dimension $d^{2}$ and labeled
by $\hat{A}$ and $\hat{B}$ without loss of generality, twirl it according to
\eqref{eq:bilat-twirl} and \eqref{eq:simple-formula-bilat-twirl}, which
results in the following resource state:%
\begin{equation}
\frac{1}{d}\Phi_{\hat{A}\hat{B}}^{d^{2}}+\left(  1-\frac{1}{d}\right)
\frac{\left(  I_{\hat{A}\hat{B}}-\Phi_{\hat{A}\hat{B}}^{d^{2}}\right)  }%
{d^{4}-1}.
\end{equation}
From there, Alice and Bob can locally separate out $\Phi_{\hat{A}\hat{B}%
}^{d^{2}}$ to two e-dits $\Phi_{A_{1}B_{1}}^{d}\otimes\Phi_{A_{2}B_{2}}^{d}$
via local isometries and each perform a teleportation in opposite directions.
By following an error analysis similar to that in
\eqref{eq:err-eval-no-res-1}--\eqref{eq:err-eval-no-res-last}, we conclude
that this scheme achieves a simulation error equal to $1-\frac{1}{d}$.
\end{remark}

\subsection{Werner states}

\label{sec:werner-example}Another general class of bipartite states of
interest in quantum information is the class of Werner states \cite{W89}. A
Werner state of parameter $p\in\left[  0,1\right]  $ and dimension $d_{\hat
{A}}\in\left\{  2,3,4,\ldots\right\}  $ is defined as follows:%
\begin{multline}
W_{\hat{A}\hat{B}}^{(p,d_{\hat{A}})}\coloneqq\label{eq:Werner-state-def}\\
\left(  1-p\right)  \frac{2}{d_{\hat{A}}\left(  d_{\hat{A}}+1\right)  }%
\Pi_{\hat{A}\hat{B}}^{\mathcal{S}}+p\frac{2}{d_{\hat{A}}\left(  d_{\hat{A}%
}-1\right)  }\Pi_{\hat{A}\hat{B}}^{\mathcal{A}},
\end{multline}
where $\Pi_{\hat{A}\hat{B}}^{\mathcal{S}}\coloneqq (I_{\hat{A}\hat{B}}%
+F_{\hat{A}\hat{B}})/2$, $\Pi_{\hat{A}\hat{B}}^{\mathcal{A}}\coloneqq (I_{\hat
{A}\hat{B}}-F_{\hat{A}\hat{B}})/2$, and $F_{\hat{A}\hat{B}}$ is the unitary
swap operator defined in \eqref{eq:swap-op-def}. The general interest in
Werner states stems from the fact that an arbitrary state of systems $\hat
{A}\hat{B}$ can be twirled in a different way to a Werner state. That is, a
Werner state results from the following bilateral twirl:%
\begin{equation}
\widetilde{\mathcal{W}}_{CD}(X_{CD})\coloneqq \int dU\ (\mathcal{U}_{C}%
\otimes\mathcal{U}_{D})(X_{CD}), \label{eq:werner-twirl-1}%
\end{equation}
where the notation is defined similarly to that in \eqref{eq:bilat-twirl}. It
is well known that the action of a Werner twirl as above, on an arbitrary
input operator $X_{CD}$, is as follows \cite{Watrous2018}:%
\begin{multline}
\widetilde{\mathcal{W}}_{CD}(X_{CD})=\operatorname{Tr}[\Pi_{\hat{A}\hat{B}%
}^{\mathcal{S}}X_{CD}]\frac{2}{d_{\hat{A}}\left(  d_{\hat{A}}+1\right)  }%
\Pi_{\hat{A}\hat{B}}^{\mathcal{S}}\label{eq:werner-twirl-2}\\
+\operatorname{Tr}[\Pi_{\hat{A}\hat{B}}^{\mathcal{A}}X_{CD}]\frac{2}%
{d_{\hat{A}}\left(  d_{\hat{A}}-1\right)  }\Pi_{\hat{A}\hat{B}}^{\mathcal{A}}.
\end{multline}
Just as with isotropic states, the general interest in them stems from the
fact that they result from a twirl of an arbitrary bipartite state according
to $\widetilde{\mathcal{W}}_{CD}$. Also, the states are easily characterized
in terms of just two parameters.

The following proposition establishes a simple expression for the simulation
error when using a Werner state for bidirectional teleportation. It is given
exclusively in terms of the dimension $d$ of the swap channel that is being
simulated and the two parameters $p$ and $d_{\hat{A}}$ that characterize the
Werner resource state. A proof is available in
Appendix~\ref{app:werner-state-proof}. The proof exploits the symmetries of a
Werner state to reduce a variation of the semi-definite program in
Proposition~\ref{prop:swap-sdp-simplify}\ to a linear program, which we then
solve analytically.

\begin{proposition}
\label{prop:werner-sim-perf}The simulation error for the unitary swap channel
when using a Werner resource state $W_{\hat{A}\hat{B}}^{(p,d_{\hat{A}})}$ is
as follows:%
\begin{equation}
\label{eq:werner-error}e_{\operatorname{PPT}}(\mathcal{S}_{AB}^{d},W_{\hat
{A}\hat{B}}^{(p,d_{\hat{A}})})=\left\{
\begin{array}
[c]{cc}%
1-\frac{1}{d^{2}} & \text{if }p\leq\frac{1}{2}\\
1-\frac{4p-2+d_{\hat{A}}}{d^{2}d_{\hat{A}}} & \text{if }p>\frac{1}{2}%
\end{array}
\right.  .
\end{equation}
If $p\leq\frac{1}{2}$, then%
\begin{equation}
e_{\operatorname{PPT}}(\mathcal{S}_{AB}^{d},W_{\hat{A}\hat{B}}^{(p,d_{\hat{A}%
})})=e_{\operatorname{LOCC}}(\mathcal{S}_{AB}^{d},W_{\hat{A}\hat{B}%
}^{(p,d_{\hat{A}})}).
\end{equation}

\end{proposition}

It is an open question to determine if
\begin{equation}
e_{\operatorname{PPT}}(\mathcal{S}_{AB}^{d},W_{\hat{A}\hat{B}}^{(p,d_{\hat{A}%
})})=e_{\operatorname{LOCC}}(\mathcal{S}_{AB}^{d},W_{\hat{A}\hat{B}%
}^{(p,d_{\hat{A}})})
\end{equation}
for $p>\frac{1}{2}$. The main question is to design an LOCC-assisted scheme
that achieves a simulation error of $1-\frac{4p-2+d_{\hat{A}}}{d^{2}d_{\hat
{A}}}$ when $p>\frac{1}{2}$.

\begin{figure}[pt]
\begin{center}
\includegraphics[
width=\linewidth
]{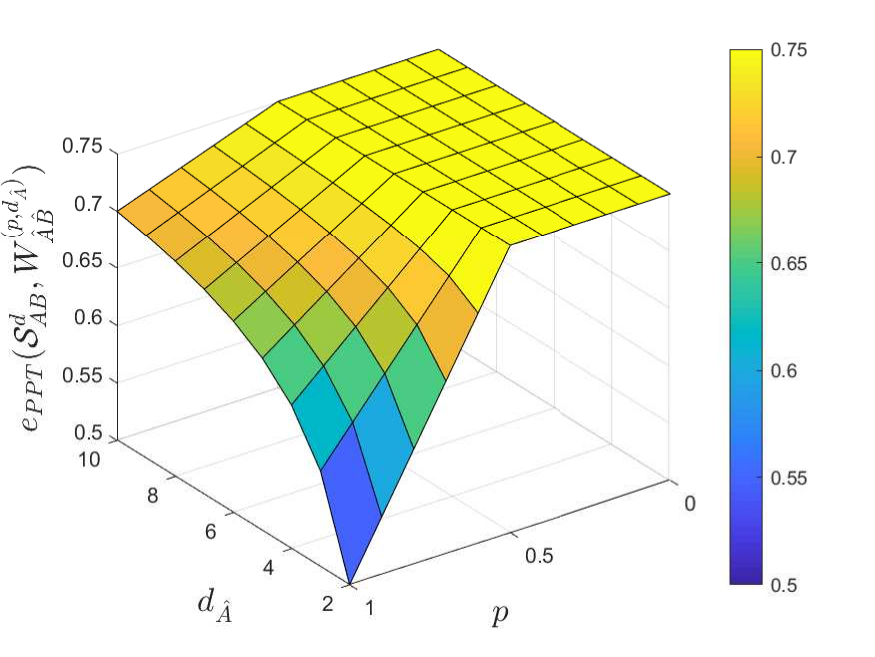}
\end{center}
\caption{Plot of the simulation error of bidirectional teleportation when
using the Werner resource state defined in $\eqref{eq:Werner-state-def}$,
where $p$ is the relative weight parameterizing the state and~$d_{\hat{A}}$ is
the dimension of Alice's system of the resource state.}%
\label{fig:wernerError}%
\end{figure}

Figure~\ref{fig:wernerError} plots the expression given in
\eqref{eq:werner-error} for the simulation error.

\subsection{Resource state resulting from generalized amplitude damping
channel}

\label{sec:gadc}

In this section, we consider a numerical example in which we can apply the
semi-definite program from Proposition~\ref{prop:swap-sdp-simplify}. This
example involves a resource state resulting from two Bell states affected by
noise from a generalized amplitude damping channel (GADC). The GADC can be
understood as a qubit thermal channel, in which the input qubit interacts with
a thermal qubit environment according to a beamsplitter-like interaction,
after which the environment qubit is discarded \cite{KSW19}. In more detail,
recall that the GADC has the following form (see, e.g., \cite{KSW19}):%
\begin{equation}
\mathcal{A}_{\gamma,N}(\rho)\coloneqq \sum_{i=1}^{4}A_{i}\rho A_{i}^{\dag},
\end{equation}
where $\gamma\in\left[  0,1\right]  $ is the damping parameter, $N\in\left[
0,1\right]  $ is the noise parameter, and%
\begin{align}
A_{1}  &  \coloneqq \sqrt{1-N}\left(  |0\rangle\!\langle0|+\sqrt{1-\gamma
}|1\rangle\!\langle1|\right)  ,\\
A_{2}  &  \coloneqq \sqrt{\gamma\left(  1-N\right)  }|0\rangle\!\langle1|,\\
A_{3}  &  \coloneqq \sqrt{N}\left(  \sqrt{1-\gamma}|0\rangle\!\langle
0|+|1\rangle\langle1|\right)  ,\\
A_{4}  &  \coloneqq \sqrt{\gamma N}|1\rangle\!\langle0|.
\end{align}
The resource state we consider is then%
\begin{equation}
\mathcal{A}_{\gamma,N}^{\otimes4}(\Phi^{\otimes2}), \label{eq:GADC-res-state}%
\end{equation}
where the maximally entangled state $\Phi$ is defined as%
\begin{equation}
\Phi\coloneqq \frac{1}{2}\sum_{i,j=0}^{1}|i\rangle\!\langle j|\otimes
|i\rangle\!\langle j|.
\end{equation}
The resource state in \eqref{eq:GADC-res-state} is equivalent to two ebits,
consisting of four qubits in total, each of which is acted upon by a GADC with
the same parameters $\gamma$ and $N$. When $\gamma$ and $N$ are both equal to
zero, the resource state is equivalent to two ebits and perfect bidirectional
teleportation is possible. As the noise parameters increase, the bidirectional
teleportation is imperfect and occurs with some error.

By evaluating the semi-definite program in
Proposition~\ref{prop:swap-sdp-simplify}\ for this resource state, we obtain a
lower bound on the simulation error of bidirectional teleportation. We obtain
an upper bound by demonstrating a protocol that uses this resource state. If
Alice and Bob perform a bilateral twirl on their state, specifically, the
channel in \eqref{eq:bilat-twirl}, where $U$ is a unitary that acts on two
qubits, then the resulting state is an isotropic state of the following form:%
\begin{equation}
F(\gamma,N)\Phi^{\otimes2}+\left(  1-F(\gamma,N)\right)  \frac{I^{\otimes
4}-\Phi^{\otimes2}}{15},
\end{equation}
where%
\begin{align}
F(\gamma,N)  &  \coloneqq \operatorname{Tr}[\Phi^{\otimes2}\mathcal{A}%
_{\gamma,N}^{\otimes4}(\Phi^{\otimes2})]\\
&  =\left[  1+\frac{\gamma}{2}\left(  \gamma-2\left[  1+\gamma N\left(
1-N\right)  \right]  \right)  \right]  ^{2}.
\end{align}
By applying Proposition~\ref{prop:isotropic-sim-perf}\ and noting that
$d_{\hat{A}}=d^{2}=4$ for this example, we find that the simulation error,
when using this protocol, is given by%
\begin{equation}
1-\max\left\{  F(\gamma,N),\frac{1}{16}\right\}  . \label{eq:upper-bnd-GADC}%
\end{equation}
Up to numerical precision, we find that the upper bound in
\eqref{eq:upper-bnd-GADC} and the SDP\ lower bound from
Proposition~\ref{prop:swap-sdp-simplify}\ match, so that
\eqref{eq:upper-bnd-GADC} should in fact be an exact analytical expression for
the simulation error when using this resource state.

Figure~\ref{fig:gadc-values} plots the expression in \eqref{eq:upper-bnd-GADC}
for the simulation error. The simulation error tends to zero as the damping
parameter $\gamma$ approaches zero (so that the channel $\mathcal{A}_{\gamma,
N}$ is converging to an identity channel and thus the resource state to two
ebits). For fixed $\gamma$ and the noise parameter $N$ converging to 1/2, the
simulation error increases.

\begin{figure}[ptb]
\begin{center}
\includegraphics[
width=\linewidth
]{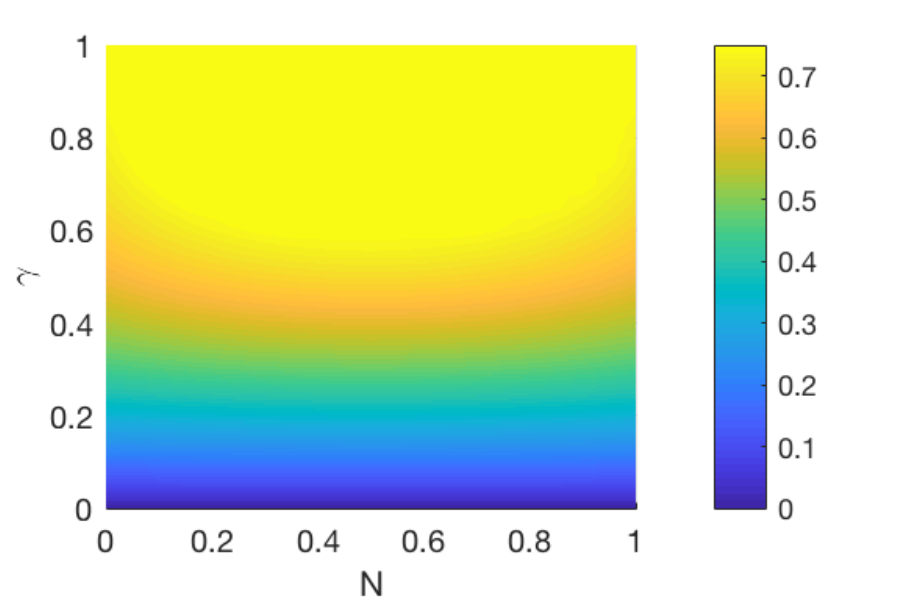}
\end{center}
\caption{Plot of the simulation error of bidirectional teleportation when
using the resource state in \eqref{eq:GADC-res-state}, for all $\gamma
,N\in\left[  0,1\right]  $.}%
\label{fig:gadc-values}%
\end{figure}

\section{On the performance of the KPF16~protocols for bidirectional
teleportation}

\label{sec:KPF16}We now use our framework to evaluate the performance of some
proposals for bidirectional teleportation from \cite{KPF16}. Let us call these
proposals KPF16, after the authors and year of publication of \cite{KPF16}.
Our main conclusion is that the schemes from \cite{KPF16} are suboptimal
according to the performance metric in \eqref{eq:sim-err-bipartite}, whereas
our simple strategy proposed in Remark~\ref{rem:single-edit-strat}\ is optimal.

\subsection{Bipartite channel for the first KPF16 protocol}

\begin{figure}[ptb]
\begin{center}
\includegraphics[scale=0.75]{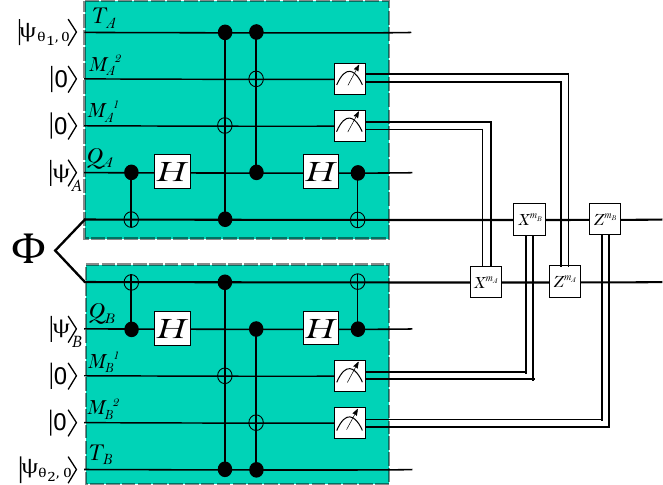}
\end{center}
\caption{Shown above is the first proposal for bidirectional teleportation
from Kiktenko~\emph{et al}. The scheme utilizes ten qubits with subindices $A$
and $B$ to denote Alice and Bob, respectively. $T_{A}$ and $T_{B}$ denote the
trigger qubits of Alice and Bob, respectively, whose states dictate the
actions of both parties. If either party's trigger qubit is in the state~1,
that party performs unidirectional quantum teleportation on their end. There
are four auxiliary qubits, labeled by $M_{A}^{1}$, $M_{A}^{2}$, $M_{B}^{1}$,
and $M_{B}^{2}$, which are each initialized to the state $|0\rangle$ and are
used to store projective measurement outcomes. The states $|\psi\rangle
_{Q_{A}}$ and $|\psi\rangle_{Q_{B}}$ are those that Alice and Bob wish to swap
(more generally, the state can be a joint state $|\phi\rangle_{AB}$ of both
systems). The resource state shared by Alice and Bob to help with the task is
a single Bell state $|\Phi\rangle_{C_{A}C_{B}}\coloneqq (|00\rangle
_{C_{A}C_{B}}+|11\rangle_{C_{A}C_{B}})/\sqrt{2}$. }%
\label{fig: BQT Circuit}%
\end{figure}

To see this, let us recall the first proposal for bidirectional teleportation
from \cite{KPF16} shown in Figure~\ref{fig: BQT Circuit}. The scheme presented
in Figure~2(a) of \cite{KPF16} utilizes ten qubits with subindices $A$ and $B$
to denote Alice and Bob, respectively. The system label $T_{A}$ denotes a
\textquotedblleft trigger qubit\textquotedblright\ of Alice, and the system
label $T_{B}$ denotes a trigger qubit of Bob (the notation here overlaps with
our notation for partial transpose, but it should be clear from the context).
There are four auxiliary qubits, labeled by $M_{A}^{1}$, $M_{A}^{2}$,
$M_{B}^{1}$, and $M_{B}^{2}$, that are each initialized to the state
$|0\rangle$ and are used to store projective measurement outcomes. The states
$|\psi\rangle_{Q_{A}}$ and $|\psi\rangle_{Q_{B}}$ are those that Alice and Bob
wish to swap (more generally, the state can be a joint state $|\phi
\rangle_{AB}$ of both systems). The resource state shared by Alice and Bob to
help with the task is a single Bell state $|\Phi\rangle_{C_{A}C_{B}%
}\coloneqq (|00\rangle_{C_{A}C_{B}}+|11\rangle_{C_{A}C_{B}})/\sqrt{2}$. The
trigger qubits of Alice and Bob are initialized to the states $|\psi
_{\theta_{1}}\rangle$ and $|\psi_{\theta_{2}}\rangle$, respectively, where%
\begin{equation}
|\psi_{\theta}\rangle\coloneqq \cos(\theta/2)|0\rangle+\sin(\theta
/2)|1\rangle.
\end{equation}

The trigger qubits act as controls for local Bell measurements. That is, if
the trigger qubit of Alice is in the state $|0\rangle$, then a Bell
measurement is not performed on her systems $Q_{A}$ and $C_{A}$. If the
trigger qubit of Alice is in the state $|1\rangle$, then a Bell measurement is
performed on her systems $Q_{A}$ and $C_{A}$. The quantum circuit in
Figure~2(a) of \cite{KPF16} indicates that these actions happen coherently.
However, the trigger qubits are discarded at the end of the circuit, and so
they really just end up playing the role of random control bits. A similar
description applies to Bob's side.

Thus, the following actions are taken, depending on the value of the trigger
qubits. If both are zero, then no action is taken and the output of the
circuit is a Bell state on systems $C_{A}$ and $C_{B}$. If Alice's trigger
qubit is equal to one and Bob's is equal to zero, then unidirectional
teleportation is performed from Alice to Bob. The result is that an identity
channel takes system $C_{A}$ to system $C_{B}$, while a completely
depolarizing channel is performed on system $C_{B}$ and takes it to system
$C_{A}$. If Alice's trigger qubit is equal to zero and Bob's is equal to one,
then the actions are exactly opposite of the previous setting. Finally, if
both Alice and Bob's trigger qubits are equal to one, then completely
depolarizing channels act on both systems $C_{A}$ and $C_{B}$, so that the
output in this case is the maximally mixed state of two qubits. In this last
case, the protocol is such that the output states become \textquotedblleft
jammed\textquotedblright\ due to Alice and Bob both trying to teleport at the
same time. The various actions are summarized in Table~\ref{table:nonlin}%
.\begin{table}[t]
\caption{Actions realized by KPF16 protocol for bidirectional teleportation.}%
\label{table:nonlin}
\centering
\begin{tabular}
[c]{ccc}\hline\hline
Trigger Q.'s & Probability & Bipartite Channel Term\\[0.5ex]\hline
00 & $\cos^{2}(\theta_{1}/2)\cos^{2}(\theta_{2}/2)$ & $\rho_{Q_{A}Q_{B}%
}\rightarrow{\Phi}$\\
01 & $\cos^{2}(\theta_{1}/2)\sin^{2}(\theta_{2}/2)$ & ${\rho}_{{Q}_{A}{Q}_{B}%
}\rightarrow{\rho}_{{Q}_{B}}\otimes{\pi}$\\
10 & $\sin^{2}(\theta_{1}/2)\cos^{2}(\theta_{2}/2)$ & ${\rho}_{{Q}_{A}{Q}_{B}%
}\rightarrow{\pi}\otimes{\rho}_{{Q}_{A}}$\\
11 & $\sin^{2}(\theta_{1}/2)\sin^{2}(\theta_{2}/2)$ & ${\pi}\otimes{\pi}%
$\\[1ex]\hline
\end{tabular}
\end{table}

The bipartite channel realized by the first KPF16 protocol is thus as follows:%
\begin{multline}
\mathcal{K}_{Q_{A}Q_{B}\rightarrow C_{A}C_{B}}(\rho_{Q_{A}Q_{B}}%
)\coloneqq\label{eq:KPF16-bi-ch}\\
\left(  1-p_{1}\right)  \left(  1-p_{2}\right)  \operatorname{Tr}[\rho
_{Q_{A}Q_{B}}]\Phi_{C_{A}C_{B}}\\
+\left(  1-p_{1}\right)  p_{2}(\mathcal{R}_{Q_{A}\rightarrow C_{B}}^{\pi
}\otimes\operatorname{id}_{Q_{B}\rightarrow C_{A}})(\rho_{Q_{A}Q_{B}})\\
+p_{1}\left(  1-p_{2}\right)  (\operatorname{id}_{Q_{A}\rightarrow C_{B}%
}\otimes\mathcal{R}_{Q_{B}\rightarrow C_{A}}^{\pi})(\rho_{Q_{A}Q_{B}})\\
+p_{1}p_{2}\mathcal{R}_{Q_{A}Q_{B}\rightarrow C_{A}C_{B}}^{\pi\otimes\pi}%
(\rho_{Q_{A}Q_{B}}),
\end{multline}
where $p_{i}=\sin^{2}(\theta_{i}/2)$ for $i\in\left\{  1,2\right\}  $,
$\mathcal{R}^{\pi}$ is a replacer channel that traces out its input and
replaces it with the maximally mixed qubit state $\pi=I/2$, and $\mathcal{R}%
^{\pi\otimes\pi}$ is a replacer channel that traces out its two-qubit input
and replaces with the maximally mixed state $\pi\otimes\pi$ of two qubits.

One of the first observations that we can make about the bipartite channel in
\eqref{eq:KPF16-bi-ch} is that none of the terms contain the ideal swap
channel. That is, with the first KPF16 protocol, the ideal swap channel does
not occur even probabilistically. Furthermore, the bipartite channel in
\eqref{eq:KPF16-bi-ch} does not obey the symmetry of the swap channel in
\eqref{eq:swap-ch-symmetry}, due to the presence of the first term in
\eqref{eq:KPF16-bi-ch}, whereas we know from the analysis in
\eqref{eq:app-swap-analysis-1}--\eqref{eq:app-swap-analysis-last} that it
should if it is to be an optimal simulation. As such, these are strong
indicators that this protocol will not perform well as an approximation of
bidirectional teleportation, according to the metric defined in
\eqref{eq:err-sim-bi-ch-fixed-locc}. In what follows, we prove that this is
the case.

Let us now calculate the Choi operator of the channel~$\mathcal{K}_{Q_{A}%
Q_{B}\rightarrow C_{A}C_{B}}$, which is defined as%
\begin{equation}
K_{Q_{A}C_{A}C_{B}Q_{B}}\coloneqq \mathcal{K}_{\bar{Q}_{A}\bar{Q}%
_{B}\rightarrow C_{A}C_{B}}(\Gamma_{Q_{A}\bar{Q}_{A}}\otimes\Gamma_{\bar
{Q}_{B}Q_{B}}),
\end{equation}
where%
\begin{equation}
\Gamma_{Q_{A}\bar{Q}_{A}}\coloneqq \sum_{i,j\in\left\{  0,1\right\}
}|i\rangle\!\langle j|_{Q_{A}}\otimes|i\rangle\!\langle j|_{\bar{Q}_{A}},
\end{equation}
and $\Gamma_{\bar{Q}_{B}Q_{B}}$ is similarly defined. To do so, let us
consider the four terms in \eqref{eq:KPF16-bi-ch}\ separately. We find that
the first term is%
\begin{multline}
\operatorname{Tr}_{\bar{Q}_{A}\bar{Q}_{B}}[\Gamma_{Q_{A}\bar{Q}_{A}}%
\otimes\Gamma_{\bar{Q}_{B}Q_{B}}]\Phi_{C_{A}C_{B}}=\\
I_{Q_{A}}\otimes\Phi_{C_{A}C_{B}}\otimes I_{Q_{B}}.
\end{multline}
The second term is%
\begin{multline}
(\mathcal{R}_{\bar{Q}_{A}\rightarrow C_{B}}^{\pi}\otimes\operatorname{id}%
_{\bar{Q}_{B}\rightarrow C_{A}})(\Gamma_{Q_{A}\bar{Q}_{A}}\otimes\Gamma
_{\bar{Q}_{B}Q_{B}})=\\
I_{Q_{A}}\otimes\Gamma_{C_{A}Q_{B}}\otimes\pi_{C_{B}}.
\end{multline}
The third term is%
\begin{multline}
(\operatorname{id}_{\bar{Q}_{A}\rightarrow C_{B}}\otimes\mathcal{R}_{\bar
{Q}_{B}\rightarrow C_{A}}^{\pi})(\Gamma_{Q_{A}\bar{Q}_{A}}\otimes\Gamma
_{\bar{Q}_{B}Q_{B}})=\\
\Gamma_{Q_{A}C_{B}}\otimes\pi_{C_{A}}\otimes I_{Q_{B}}.
\end{multline}
The fourth term is%
\begin{equation}
\mathcal{R}_{Q_{A}Q_{B}\rightarrow C_{A}C_{B}}^{\pi\otimes\pi}(\Gamma
_{Q_{A}\bar{Q}_{A}}\otimes\Gamma_{\bar{Q}_{B}Q_{B}})=I_{Q_{A}}\otimes
\pi_{C_{A}}\otimes\pi_{C_{B}}\otimes I_{Q_{B}}.
\end{equation}
Putting everything together, we conclude that the Choi operator $K_{Q_{A}%
C_{A}C_{B}Q_{B}}$ of the channel $\mathcal{K}_{\bar{Q}_{A}\bar{Q}%
_{B}\rightarrow C_{A}C_{B}}$ is as follows:%
\begin{multline}
K_{Q_{A}C_{A}C_{B}Q_{B}}=\left(  1-p_{1}\right)  \left(  1-p_{2}\right)
I_{Q_{A}}\otimes\Phi_{C_{A}C_{B}}\otimes I_{Q_{B}}\label{eq:Choi-op-KPF16}\\
+\left(  1-p_{1}\right)  p_{2}I_{Q_{A}}\otimes\Gamma_{C_{A}Q_{B}}\otimes
\pi_{C_{B}}\\
+p_{1}\left(  1-p_{2}\right)  \Gamma_{Q_{A}C_{B}}\otimes\pi_{C_{A}}\otimes
I_{Q_{B}}\\
+p_{1}p_{2}I_{Q_{A}}\otimes\pi_{C_{A}}\otimes\pi_{C_{B}}\otimes I_{Q_{B}}.
\end{multline}
This is helpful not only for calculating the normalized diamond distance
between the ideal swap channel and the first KPF16 protocol, but also for
estimating the fidelity when sending in shares of maximally entangled states.

\subsection{Fidelity of the first KPF16 protocol to ideal bidirectional
teleportation}

We now exploit the Choi operator in \eqref{eq:Choi-op-KPF16} in order to
compare the first KPF16 protocol with ideal bidirectional teleportation via
fidelity. As recalled in Section~\ref{sec:ideal-BQT}, ideal bidirectional
teleportation realizes the following state when acting on maximally entangled
inputs:%
\begin{equation}
\mathcal{S}_{\bar{Q}_{A}\bar{Q}_{B}\rightarrow C_{A}C_{B}}(\Phi_{Q_{A}\bar
{Q}_{A}}\otimes\Phi_{\bar{Q}_{B}Q_{B}})=\Phi_{Q_{A}C_{B}}\otimes\Phi
_{C_{A}Q_{B}}.
\end{equation}
Since the state above is a pure state, we can calculate the fidelity by means
of the following formula:%
\begin{equation}
\operatorname{Tr}[(\Phi_{Q_{A}C_{B}}\otimes\Phi_{C_{A}Q_{B}})K_{Q_{A}%
C_{A}C_{B}Q_{B}}/4], \label{eq:fid-KPF16-to-swap}%
\end{equation}
where we have normalized the Choi operator to become the Choi state of the
channel $\mathcal{K}_{\bar{Q}_{A}\bar{Q}_{B}\rightarrow C_{A}C_{B}}$. To be
clear, the Choi state is as follows:%
\begin{multline}
\frac{1}{4}K_{Q_{A}C_{A}C_{B}Q_{B}}=\left(  1-p_{1}\right)  \left(
1-p_{2}\right)  \pi_{Q_{A}}\otimes\Phi_{C_{A}C_{B}}\otimes\pi_{Q_{B}%
}\label{eq:Choi-state-KPF16}\\
+\left(  1-p_{1}\right)  p_{2}\pi_{Q_{A}}\otimes\Phi_{C_{A}Q_{B}}\otimes
\pi_{C_{B}}\\
+p_{1}\left(  1-p_{2}\right)  \Phi_{Q_{A}C_{B}}\otimes\pi_{C_{A}}\otimes
\pi_{Q_{B}}\\
+p_{1}p_{2}\pi_{Q_{A}}\otimes\pi_{C_{A}}\otimes\pi_{C_{B}}\otimes\pi_{Q_{B}}.
\end{multline}

Although the fidelity formula in \eqref{eq:fid-KPF16-to-swap} can be readily
calculated by hand, it is helpful to employ a diagrammatic calculus to do so.
Scalable ZX (SZX) calculus is a low-level graphical language with a written
set of rules utilized for the verification of quantum computations
\cite{Coecke_2011}. This technique, in which each wire in a diagram represents
a qubit, allows for one diagram to be transformed into another if both
represent the same quantum process. This makes SZX-calculus useful for a
variety of applications ranging from error correction to circuit optimization.
The basic rules of the SZX calculus that we need are summarized in
Figure~\ref{fig:ZX Calculus}.

\begin{figure}[ptb]
\begin{center}
\includegraphics[width=\columnwidth]{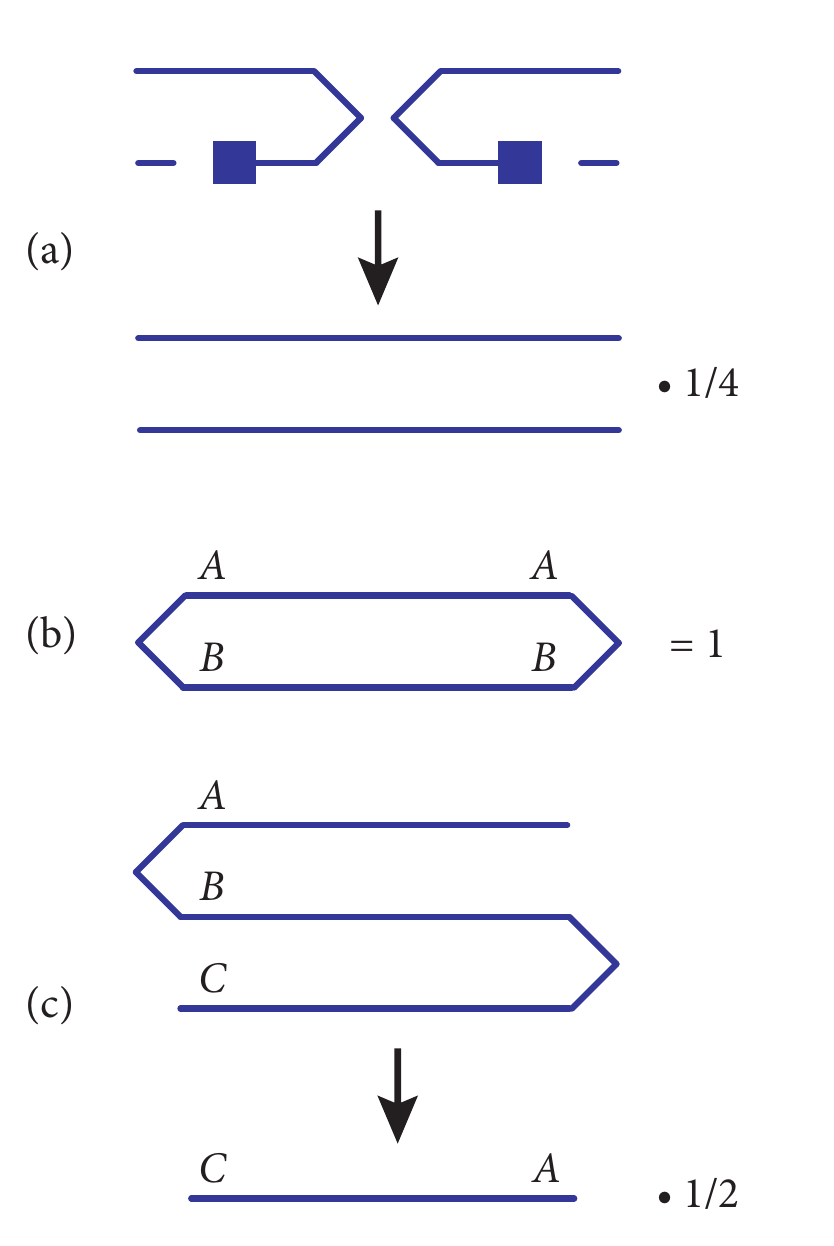}
\end{center}
\caption{(a) Tracing out one share of the maximally entangled state and
replacing with a maximally mixed state leads to two maximally mixed states.
The diagram depicts this transformation and the diagrammatic representations.
The right-hand side of the bottom diagram indicates the normalization factor
needed for two maximally mixed states. (b) The overlap of the maximally
entangled state $|\Phi\rangle$ with itself is 1. (c) Postselected
teleportation transformation, in which the overap of maximally entangled
states is reduced to an identity channel with normalization factor $1/2$.}%
\label{fig:ZX Calculus}%
\end{figure}

In what follows, we utilize the language to calculate the fidelity of the
ideal versus unideal state for each individual scenario in
Table~\ref{table:nonlin}. To be clear, each diagram in
Figures~\ref{fig:FidPoss1}--\ref{fig:FidPoss4} is a visual representation of
one term in the following formula:%
\begin{multline}
\left(  \langle\Phi|_{Q_{A}C_{B}}\otimes\langle\Phi|_{C_{A}Q_{B}}\right)
\frac{1}{4}K\left(  |\Phi\rangle_{Q_{A}C_{B}}\otimes|\Phi\rangle_{C_{A}Q_{B}%
}\right) \\
=(1-p_{1})(1-p_{2})F_{1}+p_{2}(1-p_{1})F_{2}\\
+p_{1}(1-p_{2})F_{3}+p_{1}p_{2}F_{4},
\end{multline}
where we have omitted the system labels $Q_{A}C_{A}C_{B}Q_{B}$ of $K$ for
brevity. That is, our goal is to calculate $F_{1}$, $F_{2}$, $F_{3}$, and
$F_{4}$. Figures~\ref{fig:FidPoss1}--\ref{fig:FidPoss4} accomplish this goal,
from which we conclude that $F_{1}=F_{4}=\frac{1}{16}$ and $F_{2}=F_{3}%
=\frac{1}{4}$. See the captions of Figures~\ref{fig:FidPoss1}%
--\ref{fig:FidPoss4} for explanations of these calculations. Then we find that
the fidelity is equal to
\begin{multline}
(1-p_{1})(1-p_{2})\frac{1}{16}+p_{2}(1-p_{1})\frac{1}{4}\label{eq:fid-KPF16}\\
+p_{1}(1-p_{2})\frac{1}{4}+p_{1}p_{2}\frac{1}{16}\\
=\frac{1}{16}\left(  1+3\left(  p_{1}+p_{2}\right)  -6p_{1}p_{2}\right)  .
\end{multline}
This function has a maximum at either $p_{1}=1$ and $p_{2}=0$ or $p_{1}=0$ and
$p_{1}=1$, with both cases leading to a maximum fidelity of $\frac{1}{4}$ for
the KPF16 protocol. These optimal values of $p_{1}$ and $p_{2}$ correspond to
a unidirectional teleportation strategy for one party and a replacer channel
for the other and vice versa, which clearly do not suffice for bidirectional teleportation.

In contrast, our simple approach from Remark~\ref{rem:single-edit-strat}%
\ achieves a fidelity of $\frac{1}{2}$, which follows from a straightforward
calculation. Thus, our approach provides a significant improvement over the
first KPF16 protocol when only a single ebit is available for bidirectional teleportation.



\begin{figure}[ptb]
\begin{center}
\includegraphics[width=\columnwidth]{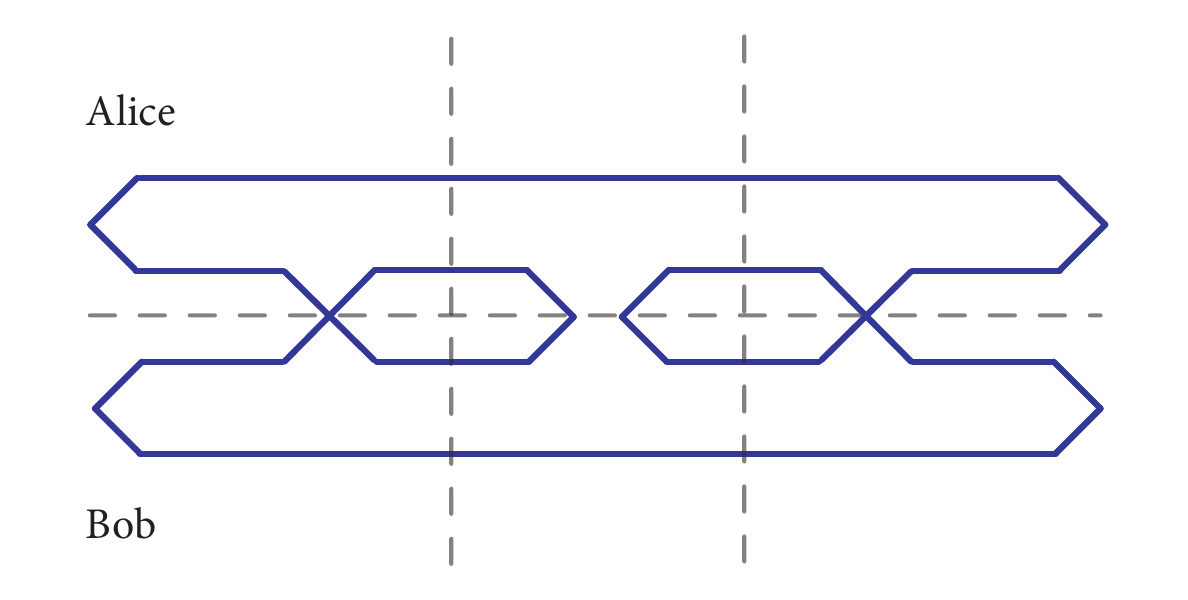}
\end{center}
\caption{The two horizontal lines in the middle section, on the top and
bottom, are each weighted by $\frac{1}{2}$ due to the presence of maximally
mixed states. Twice applying the third rule in Figure~\ref{fig:ZX Calculus}
leads to two factors of $\frac{1}{2}$ and finally applying the second rule in
Figure~\ref{fig:ZX Calculus} completes the calculation, leading to a value of
$\frac{1}{16}$.}%
\label{fig:FidPoss1}%
\end{figure}

\begin{figure}[ptb]
\begin{center}
\includegraphics[width=\columnwidth]{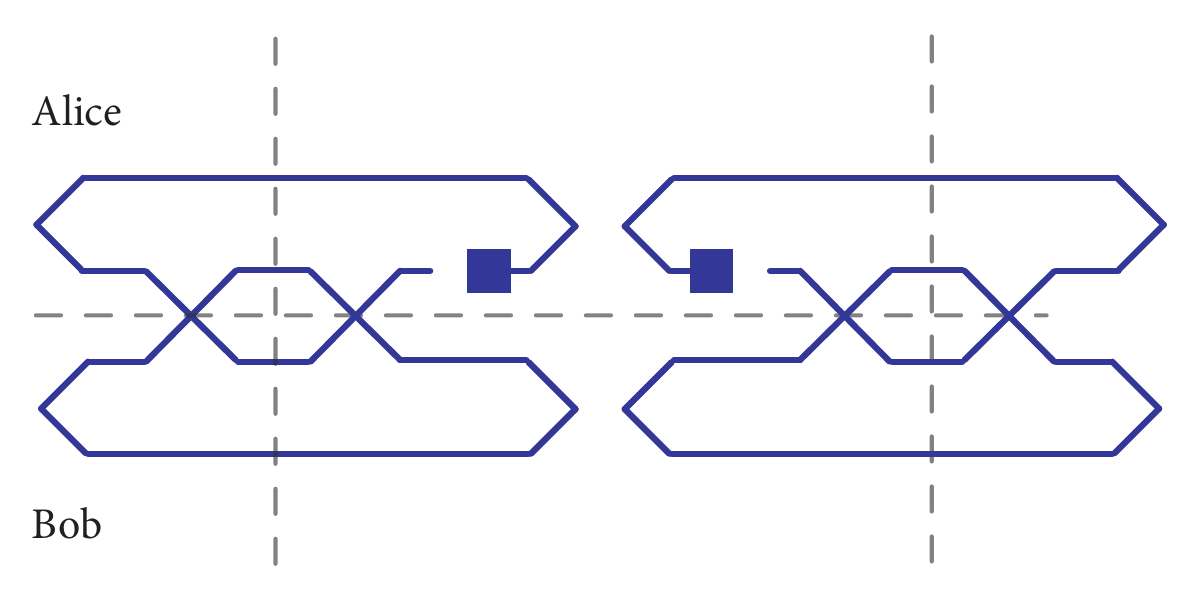}
\end{center}
\caption{ Applying the first rule from Figure~\ref{fig:ZX Calculus} collapses
the upper part of the middle section to two horizonal lines each weighted by
$\frac{1}{2}$. Then we apply the second rule from Figure~\ref{fig:ZX Calculus}
to collapse the remaining three loops, leading to the value of$~\frac{1}{4}$.}%
\label{fig:FidPoss2}%
\end{figure}

\begin{figure}[ptb]
\begin{center}
\includegraphics[width=\columnwidth]{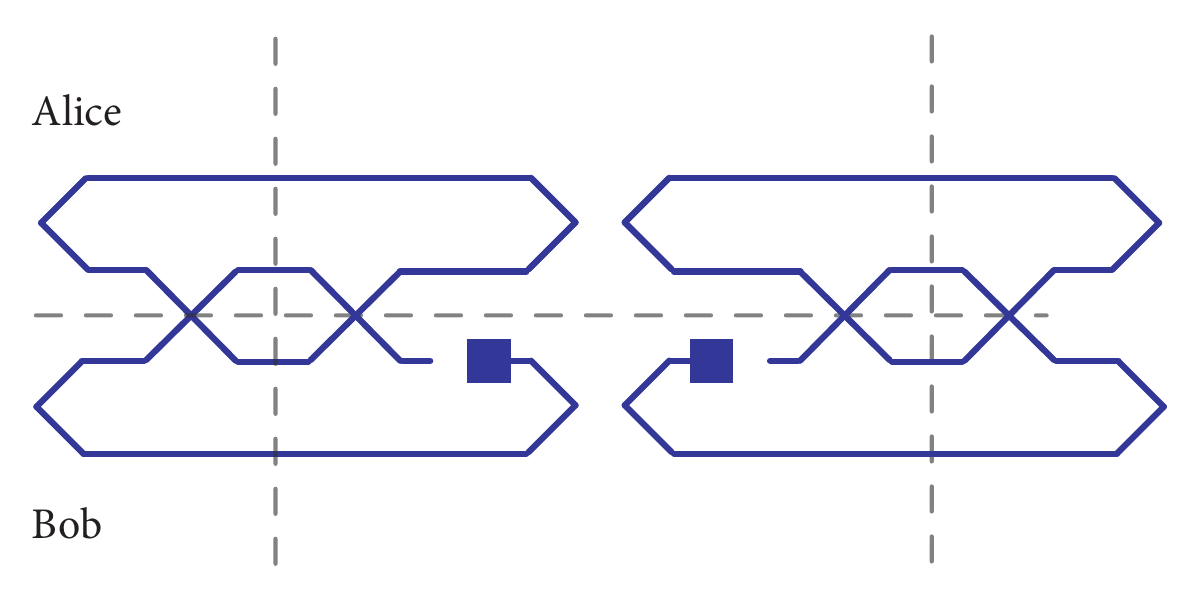}
\end{center}
\caption{ The analysis here is the same as that in Figure~\ref{fig:FidPoss2},
leading to a value of $\frac{1}{4}$.}%
\label{fig:FidPoss3}%
\end{figure}

\begin{figure}[ptb]
\begin{center}
\includegraphics[width=\columnwidth]{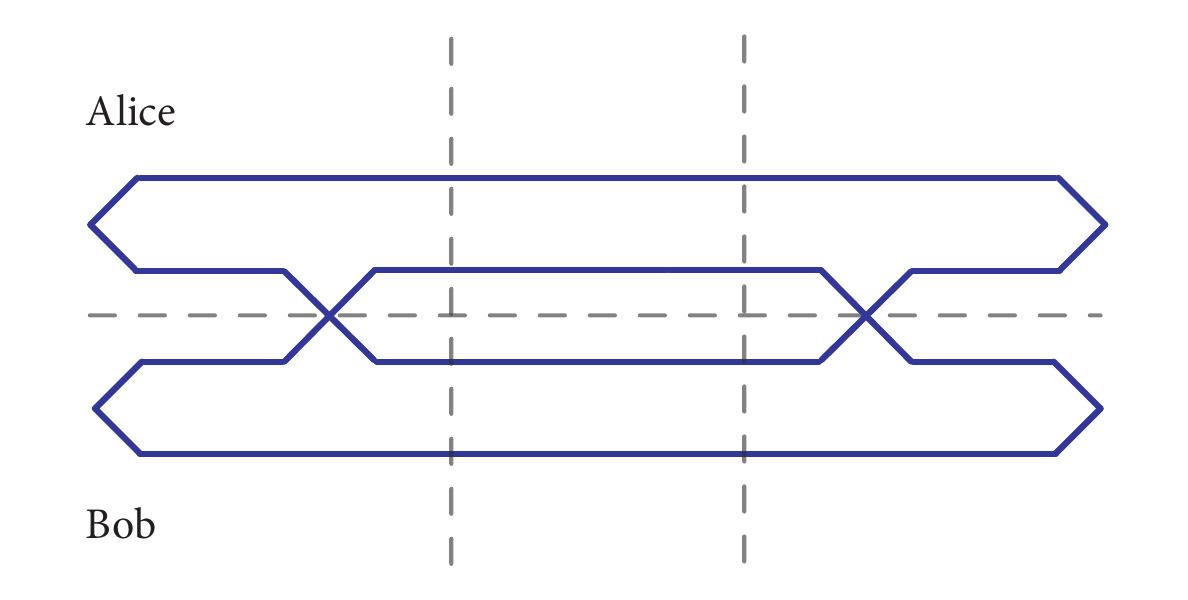}
\end{center}
\caption{ The four horizontal lines in the middle section are each weighted by
$\frac{1}{2}$, due the presence of maximally mixed states. Then we apply the
second rule from Figure~\ref{fig:ZX Calculus}~to collapse the remaining two
loops, leading to a value of $\frac{1}{16}$.}%
\label{fig:FidPoss4}%
\end{figure}

\subsection{Normalized diamond distance of the first KPF16 protocol to ideal
bidirectional teleportation}

We now consider the normalized diamond distance between the KPF16 protocol and
ideal bidirectional teleportation, i.e.,%
\begin{equation}
\frac{1}{2}\left\Vert \mathcal{K}_{\bar{Q}_{A}\bar{Q}_{B}\rightarrow
C_{A}C_{B}}^{p_{1},p_{2}}-\mathcal{S}_{\bar{Q}_{A}\bar{Q}_{B}\rightarrow
C_{A}C_{B}}\right\Vert _{\diamond}, \label{eq:KPF16-sim-err}%
\end{equation}
where we have included the dependence of $\mathcal{K}_{\bar{Q}_{A}\bar{Q}%
_{B}\rightarrow C_{A}C_{B}}^{p_{1},p_{2}}$ on $p_{1},p_{2}\in\left[
0,1\right]  $ for clarity. We have numerically implemented the semi-definite
program to calculate the simulation error above for all values of $p_{1}%
,p_{2}\in\left[  0,1\right]  $. The results are displayed in
Figure~\ref{fig:num-sim-err-KPF16}. We observe that the simulation error is
never below $\frac{3}{4}$ and this lowest value is achieved at either
$p_{1}=1$ and $p_{2}=0$ or $p_{1}=0$ and $p_{2}=1$.

\begin{figure}[ptb]
\begin{center}
\includegraphics[
width=\columnwidth
]{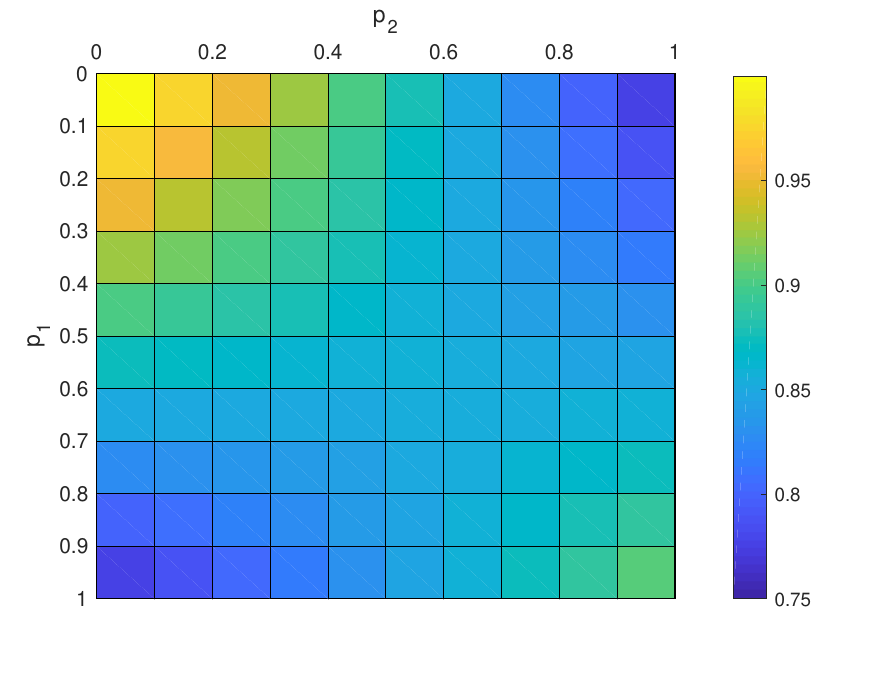}
\end{center}
\caption{Numerical calculation of normalized diamond distance in
\eqref{eq:KPF16-sim-err} for all $p_{1},p_{2}\in\lbrack0,1]$. The lowest
simulation error is achieved at either $p_{1}=1,p_{2}=0$ or $p_{1}=0,p_{2}%
=1$.}%
\label{fig:num-sim-err-KPF16}%
\end{figure}

We now prove that this is the case. That is, we prove that%
\begin{equation}
\min_{p_{1},p_{2}\in\lbrack0,1]}\frac{1}{2}\left\Vert \mathcal{K}_{\bar{Q}%
_{A}\bar{Q}_{B}\rightarrow C_{A}C_{B}}^{p_{1},p_{2}}-\mathcal{S}_{\bar{Q}%
_{A}\bar{Q}_{B}\rightarrow C_{A}C_{B}}\right\Vert _{\diamond}=\frac{3}{4}.
\label{eq:KPF16-error}%
\end{equation}
As such, the first KPF16 protocol does not perform better for bidirectional
teleportation than the classical limit from
Proposition~\ref{prop:no-res-sim-err}, which is contrary to the claim from
\cite{KPF16}. This discrepancy is due to the different performance metrics
being used in \cite{KPF16} and in our paper to quantify the performance of
bidirectional teleportation. One of the main themes of our paper is that
\eqref{eq:sim-err-bipartite} or \eqref{eq:sim-err-bipartite-infid} with
$\mathcal{N}$ set to the unitary swap channel is the correct way of
quantifying the performance of bidirectional teleportation.

To establish \eqref{eq:KPF16-error}, recall from
\eqref{eq:d-dist-pure-states}\ that the simulation error involves an
optimization over all input states to the channels. Let us pick the input
state to be a tensor product of maximally entangled states $\Phi_{Q_{A}\bar
{Q}_{A}}\otimes\Phi_{\bar{Q}_{B}Q_{B}}$. Then we find that%
\begin{align}
&  \frac{1}{2}\left\Vert \mathcal{K}_{\bar{Q}_{A}\bar{Q}_{B}\rightarrow
C_{A}C_{B}}^{p_{1},p_{2}}-\mathcal{S}_{\bar{Q}_{A}\bar{Q}_{B}\rightarrow
C_{A}C_{B}}\right\Vert _{\diamond}\nonumber\\
&  \geq\frac{1}{2}\left\Vert (\mathcal{K}^{p_{1},p_{2}}-\mathcal{S}%
)(\Phi_{Q_{A}\bar{Q}_{A}}\otimes\Phi_{\bar{Q}_{B}Q_{B}})\right\Vert _{1}\\
&  =\frac{1}{2}\left\Vert \frac{1}{4}K_{Q_{A}C_{A}C_{B}Q_{B}}^{p_{1},p_{2}%
}-\Phi_{Q_{A}C_{B}}\otimes\Phi_{C_{A}Q_{B}}\right\Vert _{1},
\end{align}
where we have omitted system labels in the second line. We can now apply the
following measurement channel:%
\begin{multline}
\omega_{Q_{A}C_{A}C_{B}Q_{B}}\rightarrow\\
\operatorname{Tr}[\omega_{Q_{A}C_{A}C_{B}Q_{B}}(\Phi_{Q_{A}C_{B}}\otimes
\Phi_{C_{A}Q_{B}})]|1\rangle\!\langle1|+\\
\operatorname{Tr}[\omega_{Q_{A}C_{A}C_{B}Q_{B}}(I_{Q_{A}C_{B}C_{A}Q_{B}}%
-\Phi_{Q_{A}C_{B}}\otimes\Phi_{C_{A}Q_{B}})]|0\rangle\!\langle0|.
\end{multline}
For the state $\Phi_{Q_{A}C_{B}}\otimes\Phi_{C_{A}Q_{B}}$, the output is
$|1\rangle\!\langle1|$, while for $\frac{1}{4}K_{Q_{A}C_{A}C_{B}Q_{B}}%
^{p_{1},p_{2}}$, the output is $F|1\rangle\!\langle1|+\left(  1-F\right)
|0\rangle\!\langle0|$, where $F$ is the fidelity formula in
\eqref{eq:fid-KPF16}. Since the trace distance does not increase under the
action of a measurement channel, we conclude that%
\begin{align}
&  \frac{1}{2}\left\Vert \frac{1}{4}K_{Q_{A}C_{A}C_{B}Q_{B}}^{p_{1},p_{2}%
}-\Phi_{Q_{A}C_{B}}\otimes\Phi_{C_{A}Q_{B}}\right\Vert _{1}\nonumber\\
&  \geq\frac{1}{2}\left\Vert F|1\rangle\!\langle1|+\left(  1-F\right)
|0\rangle\!\langle0|-|1\rangle\!\langle1|\right\Vert _{1}\\
&  =1-F\\
&  =1-\frac{1}{16}\left(  1+3\left(  p_{1}+p_{2}\right)  -6p_{1}p_{2}\right)
.
\end{align}
Since we already argued that the function in \eqref{eq:fid-KPF16} takes its
maximum value of $1/4$ at either $p_{1}=1,p_{2}=0$ or $p_{1}=0,p_{2}=1$, we
conclude that%
\begin{equation}
\min_{p_{1},p_{2}\in\lbrack0,1]}\frac{1}{2}\left\Vert \mathcal{K}_{\bar{Q}%
_{A}\bar{Q}_{B}\rightarrow C_{A}C_{B}}^{p_{1},p_{2}}-\mathcal{S}_{\bar{Q}%
_{A}\bar{Q}_{B}\rightarrow C_{A}C_{B}}\right\Vert _{\diamond}\geq\frac{3}{4}.
\end{equation}
This value is actually achieved because the optimal input state for
$p_{1}=1,p_{2}=0$ or $p_{1}=0,p_{2}=1$ is a tensor product of maximally
entangled states. This is due to the covariance of the bipartite channel at
these special points and from an application of \cite[Corollary~II.5]{LKDW18}.

In contrast, we know from Proposition~\ref{prop:isotropic-sim-perf} and
Remark~\ref{rem:single-edit-strat} that the optimal simulation error for
bidirectional teleportation of qubits, when only a single ebit is available,
is equal to $\frac{1}{2}$. Furthermore, this simulation error is achievable
using the simple strategy outlined in Remark~\ref{rem:single-edit-strat}.

\subsection{Second protocol of KPF16}

\begin{figure}[ptb]
\begin{center}
\includegraphics[scale = 0.75]{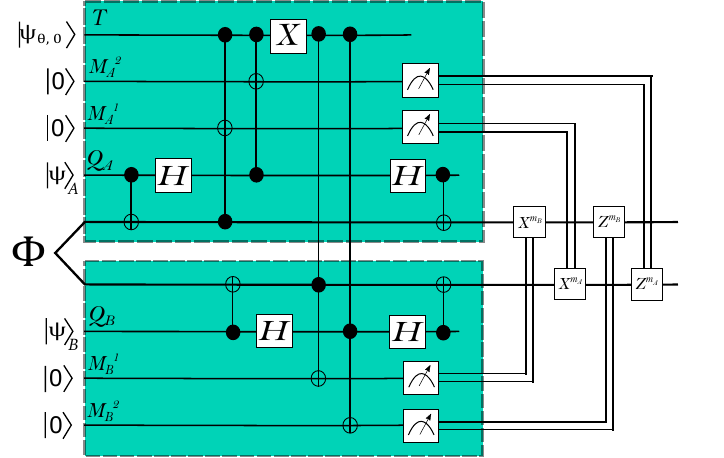}
\end{center}
\caption{Shown above is the second proposal for bidirectional teleportation
from Kiktenko~\emph{et al}. The scheme utilizes nine qubits with subindices
$A$ and $B$ to denote Alice and Bob, respectively. $T$ denotes the trigger
qubit whose state dictates whether Alice teleports to Bob or vice versa. There
are four auxiliary qubits, labeled by $M_{A}^{1}$, $M_{A}^{2}$, $M_{B}^{1}$,
and $M_{B}^{2}$, that are each initialized to the state $|0\rangle$ and are
used to store projective measurement outcomes. The states $|\psi\rangle
_{Q_{A}}$ and $|\psi\rangle_{Q_{B}}$ are those that Alice and Bob wish to swap
(more generally, the state can be a joint state $|\phi\rangle_{AB}$ of both
systems). The resource state shared by Alice and Bob to help with the task is
a single Bell state $|\Phi\rangle_{C_{A}C_{B}}\coloneqq (|00\rangle
_{C_{A}C_{B}}+|11\rangle_{C_{A}C_{B}})/\sqrt{2}$.}%
\label{fig: BQT Circuit 2}%
\end{figure}

We now analyze a second protocol from \cite{KPF16}, shown in
Figure~\ref{fig: BQT Circuit 2}, in which there is a single trigger qubit that
controls whether Alice teleports to Bob or Bob teleports to Alice. As
discussed in \cite{KPF16}, this scheme is equivalent to one in which a
classical random bit selects whether Alice teleports or Bob does, and it can
be communicated via a classical channel. The scheme thus realized a bipartite
channel that is a probabilistic mixture of the two middle terms in
\eqref{eq:KPF16-bi-ch}. The bipartite channel realized by this second protocol
is as follows:%
\begin{multline}
\mathcal{K}_{Q_{A}Q_{B}\rightarrow C_{A}C_{B}}(\rho_{Q_{A}Q_{B}})\coloneqq\\
p(\mathcal{R}_{Q_{A}\rightarrow C_{B}}^{\pi}\otimes\operatorname{id}%
_{Q_{B}\rightarrow C_{A}})(\rho_{Q_{A}Q_{B}})\\
+\left(  1-p\right)  (\operatorname{id}_{Q_{A}\rightarrow C_{B}}%
\otimes\mathcal{R}_{Q_{B}\rightarrow C_{A}}^{\pi})(\rho_{Q_{A}Q_{B}}),
\end{multline}
where $p\in\left[  0,1\right]  $. This channel is covariant and so
\cite[Corollary~II.5]{LKDW18} applies; we thus conclude that the tensor
product of maximally entangled states is an optimal input for either the
normalized diamond distance or the channel infidelity to the ideal swap
channel. By applying the analysis in Figures~\ref{fig:FidPoss2} and
\ref{fig:FidPoss3}, we conclude that both terms above have a fidelity of
$\frac{1}{4}$. Thus, the overall channel infidelity is equal to $\frac{3}{4}$
for all $p\in\left[  0,1\right]  $. One can also check that the normalized
diamond distance is equal to $\frac{3}{4}$. Thus, this second protocol of
\cite{KPF16} does not go beyond the classical limit from
Proposition~\ref{prop:no-res-sim-err}.

Another protocol suggested in \cite{KPF16} is to mix the two previous
protocols (i.e., to use one or the other probabilistically). However, such a
strategy never achieves a fidelity higher than $\frac{1}{4}$, and by previous
analyses that we have given, such a strategy does not go beyond the classical
limit from Proposition~\ref{prop:no-res-sim-err}.

As indicated previously, an optimal strategy in this scenario is to employ
that given in Remark~\ref{rem:single-edit-strat}.

\section{Generalization to multipartite channel simulation and bidirectional
controlled teleportation}

In this section, we generalize the development in
Section~\ref{sec:LOCC-sim-gen-bi}\ to the case of multipartite channel
simulation, and then we consider the specific case of bidirectional controlled
teleportation. The latter has been considered extensively in the literature
\cite{ZZQS13,Shukla2013,Li2013a,Li2013,Chen2014,Yan2013,Sun2013,Duan2014,Li2016a,ZXL19,Duan2014a,Hong2016,Sang2016,Zhang2015,SadeghiZadeh2017,Li2016}%
.

\subsection{Multipartite channel simulation}

\label{sec:MP-gen}Let us begin by recalling that a multipartite LOCC channel
can be written in the following form \cite{CLM+14}:%
\begin{equation}
\mathcal{L}_{A_{1}\cdots A_{M}\rightarrow A_{1}^{\prime}\cdots A_{M}^{\prime}%
}=\sum_{y}\bigotimes\limits_{i=1}^{M}\mathcal{E}_{A_{i}\rightarrow
A_{i}^{\prime}}^{y}, \label{eq:multip-LOCC}%
\end{equation}
where $M$ is the number of parties and the set $\{\mathcal{E}_{A_{i}%
\rightarrow A_{i}^{\prime}}^{y}\}_{y}$, for each $i\in\left\{  1,\ldots
,M\right\}  $, is a set of completely positive maps such that the sum map in
\eqref{eq:multip-LOCC} is trace preserving. However, just as in the bipartite
case, there exist channels of the form in \eqref{eq:multip-LOCC} that are not
implementable by means of LOCC.

Let us now define LOCC simulation of a multipartite channel. Let
$\mathcal{N}_{A_{1}\cdots A_{M}\rightarrow A_{1}^{\prime}\cdots A_{M}^{\prime
}}$ be a multipartite channel, in the sense that there are $M$ input systems
$A_{1}\cdots A_{M}$ and $M$ output systems $A_{1}^{\prime}\cdots A_{M}%
^{\prime}$. Furthermore, the $i$th party controls the $i$th input $A_{i}$ to
the channel, as well as the $i$th output $A_{i}^{\prime}$, for $i\in\left\{
1,\ldots,M\right\}  $. Suppose now that the $M$ parties share a resource state
$\rho_{\hat{A}_{1}\cdots\hat{A}_{M}}$. Then an LOCC simulation of
$\mathcal{N}_{A_{1}\cdots A_{M}\rightarrow A_{1}^{\prime}\cdots A_{M}^{\prime
}}$ consists of an LOCC channel $\mathcal{L}_{A_{1}\cdots A_{M}\hat{A}%
_{1}\cdots\hat{A}_{M}\rightarrow A_{1}^{\prime}\cdots A_{M}^{\prime}}$ such
that%
\begin{equation}
\widetilde{\mathcal{N}}_{A_{1}^{M}\rightarrow A_{1}^{\prime M}}(\omega
_{A_{1}^{M}})\coloneqq \mathcal{L}_{A_{1}^{M}\hat{A}_{1}^{M}\rightarrow
A_{1}^{\prime M}}(\omega_{A_{1}^{M}}\otimes\rho_{\hat{A}_{1}^{M}}),
\label{eq:LOCC-approx-MP}%
\end{equation}
where we have employed the shorthand%
\begin{align}
A_{1}^{M}  &  \equiv A_{1}\cdots A_{M},\\
A_{1}^{\prime M}  &  \equiv A_{1}^{\prime}\cdots A_{M}^{\prime},\\
\hat{A}_{1}^{M}  &  \equiv\hat{A}_{1}\cdots\hat{A}_{M}.
\end{align}
Note that $\mathcal{L}_{A_{1}^{M}\hat{A}_{1}^{M}\rightarrow A_{1}^{\prime M}}$
takes the following form:%
\begin{equation}
\mathcal{L}_{A_{1}^{M}\hat{A}_{1}^{M}\rightarrow A_{1}^{\prime M}}=\sum
_{y}\bigotimes\limits_{i=1}^{M}\mathcal{E}_{A_{i}\hat{A}_{i}\rightarrow
A_{i}^{\prime}}^{y}.
\end{equation}

The simulation error when using the LOCC channel $\mathcal{L}_{A_{1}^{M}%
\hat{A}_{1}^{M}\rightarrow A_{1}^{\prime M}}$ is then defined as follows:%
\begin{multline}
e_{\operatorname{LOCC}}(\mathcal{N}_{A_{1}^{M}\rightarrow A_{1}^{\prime M}%
},\rho_{\hat{A}_{1}^{M}},\mathcal{L}_{A_{1}^{M}\hat{A}_{1}^{M}\rightarrow
A_{1}^{\prime M}})\coloneqq\\
\frac{1}{2}\left\Vert \mathcal{N}_{A_{1}^{M}\rightarrow A_{1}^{\prime M}%
}-\widetilde{\mathcal{N}}_{A_{1}^{M}\rightarrow A_{1}^{\prime M}}\right\Vert
_{\diamond}.
\end{multline}
The simulation error minimized over all LOCC channels is defined as follows:%
\begin{multline}
e_{\operatorname{LOCC}}(\mathcal{N}_{A_{1}^{M}\rightarrow A_{1}^{\prime M}%
},\rho_{\hat{A}_{1}^{M}})\coloneqq\label{eq:MP-sim-err-LOCC}\\
\inf_{\mathcal{L}\in\operatorname{LOCC}}e_{\operatorname{LOCC}}(\mathcal{N}%
_{A_{1}^{M}\rightarrow A_{1}^{\prime M}},\rho_{\hat{A}_{1}^{M}},\mathcal{L}%
_{A_{1}^{M}\hat{A}_{1}^{M}\rightarrow A_{1}^{\prime M}}).
\end{multline}

Just as in the bipartite case, the LOCC simulation error in
\eqref{eq:MP-sim-err-LOCC} is generally hard to calculate. So we instead seek
ways of estimating or bounding it, and we can make use of an idea from
\cite{IP05}, which is helpful for developing a multipartite extension of the
bipartite case presented in Section~\ref{sec:LOCC-sim-gen-bi}. Let us define a
multipartite quantum channel $\mathcal{P}_{A_{1}^{M}\rightarrow A_{1}^{\prime
M}}$ to be completely PPT preserving if the following maps are completely
positive:%
\begin{equation}
T_{S^{\prime}}\circ\mathcal{P}_{A_{1}^{M}\rightarrow A_{1}^{\prime M}}\circ
T_{S},
\end{equation}
for all $S\in\mathbb{P}(\{A_{1},\ldots,A_{M}\})$, where $\mathbb{P}$ denotes
the power set. In the above, $T_{S}$ and $T_{S^{\prime}}$ denote partial
transpose maps acting on all subsystems in $S$ and $S^{\prime}$, respectively,
and the subset $S^{\prime}$ is chosen to correspond to the same systems in $S$
but for the channel output. Note that there is some redundancy in this
specification. There is no need to include the null set or full set from
$\mathbb{P}(\{A_{1},\ldots,A_{M}\})$ because the map $\mathcal{P}_{A_{1}%
^{M}\rightarrow A_{1}^{\prime M}}$ is completely positive. That is, the null
set corresponds to no transposes being taken, and the full set corresponds to
a full transpose on both the input and output, and it is known that the
resulting map is completely positive if and only if the original map is
completely positive. There is further redundancy in the sense that
$T_{S^{\prime}}\circ\mathcal{P}_{A_{1}^{M}\rightarrow A_{1}^{\prime M}}\circ
T_{S}$ is completely positive if and only if $T_{S^{\prime c}}\circ
\mathcal{P}_{A_{1}^{M}\rightarrow A_{1}^{\prime M}}\circ T_{S^{c}}$ is.

Let us illustrate this concept with an example. Suppose that $M=4$ and the
four parties are labeled as $ABCD$. Then a four-partite channel $\mathcal{P}%
_{ABCD\rightarrow A^{\prime}B^{\prime}C^{\prime}D^{\prime}}$ is completely PPT
preserving if the maps%
\begin{equation}
T_{S^{\prime}}\circ\mathcal{P}_{ABCD\rightarrow A^{\prime}B^{\prime}C^{\prime
}D^{\prime}}\circ T_{S}%
\end{equation}
are completely positive, where $S\in\left\{  A,B,C,D,AD,AC,DC\right\}  $ and
$S^{\prime}$ is the corresponding subset on the output systems. As stated
above, there is no need to include all of the elements of the power set.

Extending the bipartite case, we can consider this concept from the
perspective of the Choi operator. The Choi operator of a multipartite linear
map $\mathcal{P}_{A_{1}^{M}\rightarrow A_{1}^{\prime M}}$ is defined as
follows:%
\begin{equation}
P_{A_{1}^{M}A_{1}^{\prime M}}\coloneqq \mathcal{P}_{A_{1}^{M}\rightarrow
A_{1}^{\prime M}}\!\left(  \bigotimes\limits_{i=1}^{M}\Gamma_{A_{i}A_{i}%
}\right)  .
\end{equation}
The multipartite linear map $\mathcal{P}_{A_{1}^{M}\rightarrow A_{1}^{\prime
M}}$ is completely positive if and only if its Choi operator $P_{A_{1}%
^{M}A_{1}^{\prime M}}$ is positive semi-definite%
\begin{equation}
P_{A_{1}^{M}A_{1}^{\prime M}}\geq0, \label{eq:multipartite-choi-ppt-1}%
\end{equation}
and $\mathcal{P}_{A_{1}^{M}\rightarrow A_{1}^{\prime M}}$ is trace preserving
if and only if $P_{A_{1}^{M}A_{1}^{\prime M}}$ satisfies%
\begin{equation}
\operatorname{Tr}_{A_{1}^{\prime M}}[P_{A_{1}^{M}A_{1}^{\prime M}}%
]=I_{A_{1}^{M}}.
\end{equation}
Furthermore, a multipartite channel $\mathcal{P}_{A_{1}^{M}\rightarrow
A_{1}^{\prime M}}$ is completely PPT-preserving if and only if its Choi
operator satisfies%
\begin{equation}
T_{SS^{\prime}}(P_{A_{1}^{M}A_{1}^{\prime M}})\geq0
\label{eq:multipartite-choi-ppt-3}%
\end{equation}
for all $S\in\mathbb{P}(\{A_{1},\ldots,A_{M}\})$ and with $S^{\prime}$
corresponding to $S$.

We then define a PPT simulation when using a multipartite C-PPT-P channel
$\mathcal{P}_{A_{1}^{M}\hat{A}_{1}^{M}\rightarrow A_{1}^{\prime M}}$ along
with a resource state $\rho_{\hat{A}_{1}^{M}}$:%
\begin{equation}
\widetilde{\mathcal{N}}_{A_{1}^{M}\rightarrow A_{1}^{\prime M}}(\omega
_{A_{1}^{M}})\coloneqq \mathcal{P}_{A_{1}^{M}\hat{A}_{1}^{M}\rightarrow
A_{1}^{\prime M}}(\omega_{A_{1}^{M}}\otimes\rho_{\hat{A}_{1}^{M}}),
\label{eq:MP-PPT-sim-channel}%
\end{equation}
the simulation error as%
\begin{multline}
e_{\operatorname{PPT}}(\mathcal{N}_{A_{1}^{M}\rightarrow A_{1}^{\prime M}%
},\rho_{\hat{A}_{1}^{M}},\mathcal{P}_{A_{1}^{M}\hat{A}_{1}^{M}\rightarrow
A_{1}^{\prime M}})\coloneqq\\
\frac{1}{2}\left\Vert \mathcal{N}_{A_{1}^{M}\rightarrow A_{1}^{\prime M}%
}-\widetilde{\mathcal{N}}_{A_{1}^{M}\rightarrow A_{1}^{\prime M}}\right\Vert
_{\diamond},
\end{multline}
and the simulation error minimized over all C-PPT-P channels $\mathcal{P}%
_{A_{1}^{M}\hat{A}_{1}^{M}\rightarrow A_{1}^{\prime M}}$ as%
\begin{multline}
e_{\operatorname{PPT}}(\mathcal{N}_{A_{1}^{M}\rightarrow A_{1}^{\prime M}%
},\rho_{\hat{A}_{1}^{M}})\coloneqq\label{eq:sim-err-MP-PPT}\\
\inf_{\mathcal{P}\in\text{C-}\operatorname{PPT}\text{-P}}e_{\operatorname{PPT}%
}(\mathcal{N}_{A_{1}^{M}\rightarrow A_{1}^{\prime M}},\rho_{\hat{A}_{1}^{M}%
},\mathcal{P}_{A_{1}^{M}\hat{A}_{1}^{M}\rightarrow A_{1}^{\prime M}}).
\end{multline}

One of the main observations of \cite{IP05} is that the set of multipartite
LOCC channels is contained in the set of multipartite C-PPT-P channels:%
\begin{equation}
\operatorname{LOCC}\subset\text{C-PPT-P}, \label{eq:LOCC-in-PPT-MP}%
\end{equation}
which is the main reason that this concept is useful for providing bounds on
what is achievable using LOCC. As a direct consequence of
\eqref{eq:LOCC-in-PPT-MP}, we conclude that%
\begin{equation}
e_{\operatorname{PPT}}(\mathcal{N}_{A_{1}^{M}\rightarrow A_{1}^{\prime M}%
},\rho_{\hat{A}_{1}^{M}})\leq e_{\operatorname{LOCC}}(\mathcal{N}_{A_{1}%
^{M}\rightarrow A_{1}^{\prime M}},\rho_{\hat{A}_{1}^{M}}).
\end{equation}

The main advantage of this approach is due to the following generalization of
Proposition~\ref{prop:gen-bi-SDP}:

\begin{proposition}
\label{prop:gen-MP-SDP}The simulation error in \eqref{eq:sim-err-MP-PPT} can
be computed by means of the following semi-definite program:%
\begin{equation}
e_{\operatorname{PPT}}(\mathcal{N}_{A_{1}^{M}\rightarrow A_{1}^{\prime M}%
},\rho_{\hat{A}_{1}^{M}})=\inf_{\substack{\mu,Z_{A_{1}^{M}A_{1}^{\prime M}%
},\\P_{A_{1}^{M}\hat{A}_{1}^{M}A_{1}^{\prime M}}\geq0}}\mu,
\end{equation}
subject to%
\begin{align}
\mu I_{A_{1}^{M}}  &  \geq Z_{A_{1}^{M}},\\
\operatorname{Tr}_{A_{1}^{\prime M}}[P_{A_{1}^{M}\hat{A}_{1}^{M}A_{1}^{\prime
M}}]  &  =I_{A_{1}^{M}\hat{A}_{1}^{M}},
\end{align}%
\begin{equation}
Z_{A_{1}^{M}A_{1}^{\prime M}}\geq\Gamma_{A_{1}^{M}A_{1}^{\prime M}%
}^{\mathcal{N}}-\operatorname{Tr}_{\hat{A}_{1}^{M}}[T_{\hat{A}_{1}^{M}}%
(\rho_{\hat{A}_{1}^{M}})P_{A_{1}^{M}\hat{A}_{1}^{M}A_{1}^{\prime M}}],
\end{equation}
and%
\begin{equation}
T_{S\hat{S}S^{\prime}}(P_{A_{1}^{M}\hat{A}_{1}^{M}A_{1}^{\prime M}})\geq0,
\end{equation}
for all $S\in\mathbb{P}(\{A_{1},\ldots,A_{M}\})$, with $\hat{S}$ and
$S^{\prime}$ corresponding to $S$.
\end{proposition}

The proof of Proposition~\ref{prop:gen-MP-SDP}\ is very similar to the proof
of Proposition~\ref{prop:gen-bi-SDP}. It just combines the semi-definite
program for the diamond distance in \eqref{eq:SDP-diamond}\ along with the
constraints in \eqref{eq:multipartite-choi-ppt-1}--\eqref{eq:multipartite-choi-ppt-3}.

As we did previously, we can also define simulation error in terms of
infidelity. For this case, we define%
\begin{multline}
e_{\operatorname{LOCC}}^{F}(\mathcal{N}_{A_{1}^{M}\rightarrow A_{1}^{\prime
M}},\rho_{\hat{A}_{1}^{M}},\mathcal{L}_{A_{1}^{M}\hat{A}_{1}^{M}\rightarrow
A_{1}^{\prime M}})\coloneqq\\
1-F(\mathcal{N}_{A_{1}^{M}\rightarrow A_{1}^{\prime M}},\widetilde
{\mathcal{N}}_{A_{1}^{M}\rightarrow A_{1}^{\prime M}}),
\end{multline}
where $\widetilde{\mathcal{N}}_{A_{1}^{M}\rightarrow A_{1}^{\prime M}}$ is
defined in \eqref{eq:LOCC-approx-MP}, and the simulation error minimized over
all LOCC channels as%
\begin{multline}
e_{\operatorname{LOCC}}^{F}(\mathcal{N}_{A_{1}^{M}\rightarrow A_{1}^{\prime
M}},\rho_{\hat{A}_{1}^{M}})\coloneqq\\
\inf_{\mathcal{L}\in\operatorname{LOCC}}e_{\operatorname{LOCC}}^{F}%
(\mathcal{N}_{A_{1}^{M}\rightarrow A_{1}^{\prime M}},\rho_{\hat{A}_{1}^{M}%
},\mathcal{L}_{A_{1}^{M}\hat{A}_{1}^{M}\rightarrow A_{1}^{\prime M}}).
\end{multline}
We also define the following simulation error:%
\begin{multline}
e_{\operatorname{PPT}}^{F}(\mathcal{N}_{A_{1}^{M}\rightarrow A_{1}^{\prime M}%
},\rho_{\hat{A}_{1}^{M}},\mathcal{P}_{A_{1}^{M}\hat{A}_{1}^{M}\rightarrow
A_{1}^{\prime M}})\coloneqq\\
1-F(\mathcal{N}_{A_{1}^{M}\rightarrow A_{1}^{\prime M}},\widetilde
{\mathcal{N}}_{A_{1}^{M}\rightarrow A_{1}^{\prime M}}),
\end{multline}
where $\widetilde{\mathcal{N}}_{A_{1}^{M}\rightarrow A_{1}^{\prime M}}$ is
defined in \eqref{eq:MP-PPT-sim-channel}, and the simulation error minimized
over all C-PPT-P channels $\mathcal{P}_{A_{1}^{M}\hat{A}_{1}^{M}\rightarrow
A_{1}^{\prime M}}$ as%
\begin{multline}
e_{\operatorname{PPT}}^{F}(\mathcal{N}_{A_{1}^{M}\rightarrow A_{1}^{\prime M}%
},\rho_{\hat{A}_{1}^{M}})\coloneqq\label{eq:infid-sim-err-MP}\\
\inf_{\mathcal{P}\in\operatorname{PPT}}e_{\operatorname{PPT}}^{F}%
(\mathcal{N}_{A_{1}^{M}\rightarrow A_{1}^{\prime M}},\rho_{\hat{A}_{1}^{M}%
},\mathcal{P}_{A_{1}^{M}\hat{A}_{1}^{M}\rightarrow A_{1}^{\prime M}}).
\end{multline}

By reasoning similar to that used to conclude
Proposition~\ref{prop:gen-bi-SDP-infidelity}, we conclude the following:

\begin{proposition}
\label{prop:gen-MP-SDP-infid}The simulation error in
\eqref{eq:infid-sim-err-MP} can be computed by means of the following
semi-definite program:%
\begin{multline}
e_{\operatorname{PPT}}^{F}(\mathcal{N}_{A_{1}^{M}\rightarrow A_{1}^{\prime M}%
},\rho_{\hat{A}_{1}^{M}})=\\
1-\left[  \sup_{_{\lambda\geq0,P_{A_{1}^{M}\hat{A}_{1}^{M}A_{1}^{\prime M}%
}\geq0,Q_{A_{1}^{M}A_{1}^{\prime M}}}}\lambda\right]  ^{2},
\end{multline}
subject to%
\begin{align}
\lambda I_{A_{1}^{M}}  &  \leq\operatorname{Re}[\operatorname{Tr}%
_{A_{1}^{\prime M}}[Q_{A_{1}^{M}A_{1}^{\prime M}}]],\\
\operatorname{Tr}_{A_{1}^{\prime M}}[P_{A_{1}^{M}\hat{A}_{1}^{M}A_{1}^{\prime
M}}]  &  =I_{A_{1}^{M}\hat{A}_{1}^{M}},
\end{align}%
\begin{equation}%
\begin{bmatrix}
\Gamma_{A_{1}^{M}A_{1}^{\prime M}}^{\mathcal{N}} & Q_{A_{1}^{M}A_{1}^{\prime
M}}^{\dag}\\
Q_{A_{1}^{M}A_{1}^{\prime M}} & \operatorname{Tr}_{\hat{A}_{1}^{M}}[T_{\hat
{A}_{1}^{M}}(\rho_{\hat{A}_{1}^{M}})P_{A_{1}^{M}\hat{A}_{1}^{M}A_{1}^{\prime
M}}]
\end{bmatrix}
\geq0,
\end{equation}
and%
\begin{equation}
T_{S\hat{S}S^{\prime}}(P_{A_{1}^{M}\hat{A}_{1}^{M}A_{1}^{\prime M}})\geq0,
\end{equation}
for all $S\in\mathbb{P}(\{A_{1},\ldots,A_{M}\})$, with $\hat{S}$ and
$S^{\prime}$ corresponding to $S$.
\end{proposition}

\subsection{Bidirectional controlled teleportation}

\begin{figure}[ptb]
\begin{center}
\includegraphics[scale=0.5]{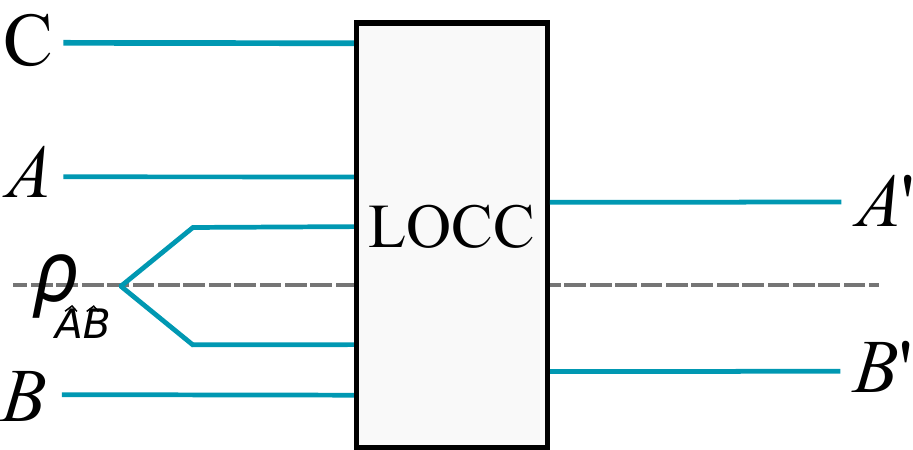}
\end{center}
\caption{Shown above is a diagram for bidirectional controlled teleportation.
The goal of this protocol is the same as the usual BQT protocol, but there is
a third party, Charlie, who helps Alice and Bob with the task. The three
parties share a resource state $\rho_{\hat{A}\hat{B}\hat{C}}$ and perform an
LOCC protocol to approximate a swap of the input systems $A$ and $B$.}%
\label{fig: Controlled BQT}%
\end{figure}

An important special case of LOCC simulation of a multipartite channel is
bidirectional controlled teleportation, depicted in
Figure~\ref{fig: Controlled BQT}. For this problem, the goal is the same as
that of bidirectional teleportation, but there is a third party Charlie who
helps with the task. In more detail, suppose that Alice, Bob, and Charlie
share a resource state $\rho_{\hat{A}\hat{B}\hat{C}}$. They then perform a
multipartite LOCC channel $\mathcal{L}_{AB\hat{A}\hat{B}\hat{C}\rightarrow
AB}$ on input systems $AB$ and the systems $\hat{A}\hat{B}\hat{C}$ of the
resource state, with the goal of simulating an ideal swap channel
$\mathcal{S}_{AB}^{d}$:%
\begin{equation}
\widetilde{\mathcal{S}}_{AB}^{d}(\omega_{AB})\coloneqq \mathcal{L}_{AB\hat
{A}\hat{B}\hat{C}\rightarrow AB}(\omega_{AB}\otimes\rho_{\hat{A}\hat{B}\hat
{C}}).
\end{equation}
Note that the LOCC channel $\mathcal{L}_{AB\hat{A}\hat{B}\hat{C}\rightarrow
AB}$ has a trivial output system for Charlie. In more detail, it can be
written in the following form:%
\begin{multline}
\mathcal{L}_{AB\hat{A}\hat{B}\hat{C}\rightarrow AB}(\tau_{A\hat{A}B\hat{B}%
\hat{C}})=\label{eq:LOCC-BTP-controlled}\\
\sum_{y}(\mathcal{E}_{A\hat{A}\rightarrow A}^{y}\otimes\mathcal{F}_{B\hat
{B}\rightarrow B}^{y})\operatorname{Tr}_{\hat{C}}[\Lambda_{\hat{C}}^{y}%
\tau_{A\hat{A}B\hat{B}\hat{C}}],
\end{multline}
where $\{\mathcal{E}_{A\hat{A}\rightarrow A}^{y}\}_{y}$ and $\{\mathcal{F}%
_{B\hat{B}\rightarrow B}^{y}\}_{y}$ are sets of completely positive maps and
$\{\Lambda_{\hat{C}}^{y}\}_{y}$ is a set of positive semi-definite operators
such that the sum map in \eqref{eq:LOCC-BTP-controlled}\ is a quantum channel.

The simulation error for bidirectional controlled teleportation, when using a
particular LOCC channel $\mathcal{L}_{AB\hat{A}\hat{B}\hat{C}\rightarrow AB}$,
is given by%
\begin{multline}
e_{\operatorname{LOCC}}(\mathcal{S}_{AB}^{d},\rho_{\hat{A}\hat{B}\hat{C}%
},\mathcal{L}_{AB\hat{A}\hat{B}\hat{C}\rightarrow AB})\coloneqq\\
\frac{1}{2}\left\Vert \mathcal{S}_{AB}^{d}-\widetilde{\mathcal{S}}_{AB}%
^{d}\right\Vert _{\diamond},
\end{multline}
and the simulation error minimized over all LOCC channels $\mathcal{L}%
_{AB\hat{A}\hat{B}\hat{C}\rightarrow AB}$ is given by%
\begin{multline}
e_{\operatorname{LOCC}}(\mathcal{S}_{AB}^{d},\rho_{\hat{A}\hat{B}\hat{C}%
})\coloneqq\label{eq:BCQT-sim-err-LOCC}\\
\inf_{\mathcal{L}\in\operatorname{LOCC}}e_{\operatorname{LOCC}}(\mathcal{S}%
_{AB}^{d},\rho_{\hat{A}\hat{B}\hat{C}},\mathcal{L}_{AB\hat{A}\hat{B}\hat
{C}\rightarrow AB}).
\end{multline}
We can also define simulation errors in terms of channel infidelity:%
\begin{equation}
e_{\operatorname{LOCC}}^{F}(\mathcal{S}_{AB}^{d},\rho_{\hat{A}\hat{B}\hat{C}%
},\mathcal{L}_{AB\hat{A}\hat{B}\hat{C}\rightarrow AB}%
)\coloneqq 1-F(\mathcal{S}_{AB}^{d},\widetilde{\mathcal{S}}_{AB}^{d}),
\end{equation}%
\begin{multline}
e_{\operatorname{LOCC}}^{F}(\mathcal{S}_{AB}^{d},\rho_{\hat{A}\hat{B}\hat{C}%
})\coloneqq\\
\inf_{\mathcal{L}\in\operatorname{LOCC}}e_{\operatorname{LOCC}}^{F}%
(\mathcal{S}_{AB}^{d},\rho_{\hat{A}\hat{B}\hat{C}},\mathcal{L}_{AB\hat{A}%
\hat{B}\hat{C}\rightarrow AB}).
\end{multline}

As before, the simulation error in \eqref{eq:BCQT-sim-err-LOCC} is difficult
to compute, and so we bound it from below by the PPT\ simulation error, which
is defined from \eqref{eq:sim-err-MP-PPT} and denoted by
$e_{\operatorname{PPT}}(\mathcal{S}_{AB}^{d},\rho_{\hat{A}\hat{B}\hat{C}})$.
As a generalization of Proposition~\ref{prop:swap-sdp-simplify}, we show that
the semi-definite program for calculating the PPT simulation error
$e_{\operatorname{PPT}}(\mathcal{S}_{AB}^{d},\rho_{\hat{A}\hat{B}\hat{C}})$
simplifies as given in Proposition~\ref{prop:swap-sdp-simplify-BCQT} below. A
proof is available in Appendix~\ref{app:simplified-SDP-BCQT}\ and is similar
to the proof of Proposition~\ref{prop:swap-sdp-simplify}. The interpretation
of an optimal simulating channel is the same as given in
Remark~\ref{rem:channel-form}, except with respect to the measurement
operators $K_{\hat{A}\hat{B}\hat{C}}$, $L_{\hat{A}\hat{B}\hat{C}}$,
$M_{\hat{A}\hat{B}\hat{C}}$, and $N_{\hat{A}\hat{B}\hat{C}}$ being subject to
the conditions in Proposition~\ref{prop:swap-sdp-simplify-BCQT} below.
Furthermore, the two different notions of simulation error based on
$e_{\operatorname{PPT}}(\mathcal{S}_{AB}^{d},\rho_{\hat{A}\hat{B}\hat{C}})$
and $e_{\operatorname{PPT}}^{F}(\mathcal{S}_{AB}^{d},\rho_{\hat{A}\hat{B}%
\hat{C}})$ coincide again.

\begin{proposition}
\label{prop:swap-sdp-simplify-BCQT}The semi-definite programs in
Propositions~\ref{prop:gen-MP-SDP} and \ref{prop:gen-MP-SDP-infid},\ for the
error in simulating the unitary SWAP channel $\mathcal{S}_{AB}^{d}$, using a
resource state $\rho_{\hat{A}\hat{B}\hat{C}}$\ simplifies as follows:%
\begin{align}
&  e_{\operatorname{PPT}}(\mathcal{S}_{AB}^{d},\rho_{\hat{A}\hat{B}\hat{C}%
})\nonumber\\
&  =e_{\operatorname{PPT}}^{F}(\mathcal{S}_{AB}^{d},\rho_{\hat{A}\hat{B}%
\hat{C}})\\
&  =1-\sup_{\substack{K_{\hat{A}\hat{B}\hat{C}},L_{\hat{A}\hat{B}\hat{C}%
},\\N_{\hat{A}\hat{B}\hat{C}}\geq0}}\operatorname{Tr}[\rho_{\hat{A}\hat{B}%
\hat{C}}K_{\hat{A}\hat{B}\hat{C}}],
\end{align}
subject to%
\begin{align}
T_{\hat{S}}\!\left(  K_{\hat{A}\hat{B}\hat{C}}+\frac{L_{\hat{A}\hat{B}\hat{C}%
}}{d+1}+\frac{N_{\hat{A}\hat{B}\hat{C}}}{\left(  d+1\right)  ^{2}}\right)   &
\geq0,\nonumber\\
\frac{1}{d^{2}-1}T_{\hat{S}}\!\left(  L_{\hat{A}\hat{B}\hat{C}}+N_{\hat{A}%
\hat{B}\hat{C}}\right)   &  \geq T_{\hat{S}}\!\left(  K_{\hat{A}\hat{B}\hat
{C}}\right)  ,\\
T_{\hat{S}}\!\left(  K_{\hat{A}\hat{B}\hat{C}}+\frac{N_{\hat{A}\hat{B}\hat{C}%
}}{\left(  d-1\right)  ^{2}}\right)   &  \geq\frac{1}{d-1}T_{\hat{S}}\!\left(
L_{\hat{A}\hat{B}\hat{C}}\right)  ,
\end{align}
for $\hat{S}\in\{\hat{A},\hat{B}\}$ and%
\begin{align}
K_{\hat{A}\hat{B}\hat{C}}+L_{\hat{A}\hat{B}\hat{C}}+N_{\hat{A}\hat{B}\hat{C}}
&  =I_{\hat{A}\hat{B}\hat{C}},\\
T_{\hat{C}}(K_{\hat{A}\hat{B}\hat{C}}),T_{\hat{C}}(L_{\hat{A}\hat{B}\hat{C}%
}),T_{\hat{C}}(N_{\hat{A}\hat{B}\hat{C}})  &  \geq0.
\end{align}

\end{proposition}

\begin{proposition}
If there is no resource state, then the error in \eqref{eq:sim-err-MP-PPT} and
\eqref{eq:MP-sim-err-LOCC}\ for simulating the unitary SWAP channel
$\mathcal{S}_{AB}^{d}$ in \eqref{eq:ideal-swap-channel} is equal to
$1-1/d^{2}$:%
\begin{equation}
e_{\operatorname{PPT}}(\mathcal{S}_{AB}^{d},\emptyset)=e_{\operatorname{LOCC}%
}(\mathcal{S}_{AB}^{d},\emptyset)=1-\frac{1}{d^{2}},
\end{equation}
where the notation $\emptyset$ indicates the absence of a resource state.
\end{proposition}

\begin{proof}
This follows from the same proof given for
Proposition~\ref{prop:no-res-sim-err}. When there is no resource state, the
state $\rho_{\hat{A}\hat{B}\hat{C}}$ collapses to the number one, and the
operators $K_{\hat{A}\hat{B}\hat{C}}$, $L_{\hat{A}\hat{B}\hat{C}}$, and
$N_{\hat{A}\hat{B}\hat{C}}$ collapse to real numbers as well. So the
optimization in Proposition~\ref{prop:swap-sdp-simplify-BCQT} collapses to the
linear program in
\eqref{eq:LP-no-res-primal}--\eqref{eq:LP-no-res-primal-last}, and we conclude
the statement above from the rest of the proof of
Proposition~\ref{prop:no-res-sim-err}.
\end{proof}

\section{Conclusion}

\label{sec:conclusion}In this paper, we have provided a systematic approach
for quantifying the performance of bidirectional teleportation and
bidirectional controlled teleportation. We have established a benchmark for
classical versus quantum bidirectional teleportation, and we have evaluated
semi-definite programming lower bounds on the simulation error for some key
examples of resource states. More generally, we have demonstrated that
semi-definite programs are possible when using the channel infidelity as an
error measure, which addresses an open question from
\cite{WW19states,Wang2019a}\ and should have applications more generally in
quantum resource theories.

Going forward from here, there are several avenues for future work. First, we
can consider other unitary channels besides the swap channel, and the line of
thinking developed here could be useful for related scenarios considered in
\cite{STM11,WSM19}. We can also consider applying the framework used here to
analyze multidirectional teleportation between more than two parties. We also
wonder whether there is an LOCC\ simulation that achieves a performance
matching the lower bound found here, for all parameter values for isotropic
and Werner states. As mentioned in the introduction of our paper, there have
been many proposals for bidirectional (controlled) teleportation. One could
also evaluate the semi-definite programming bounds from this paper for
imperfect versions of the resource states in those works in order to determine
how robust those entangled resource states are to noise.

\begin{acknowledgments}
We would like to thank Jonathan~P.~Dowling for being a catalyst for this
paper. He never stopped believing in us and our potential as researchers. His
strength and perseverance has inspired us and lives within us. May he never be forgotten.

MMW\ acknowledges Moein Sarvaghad-Moghaddam for introducing him to the topic
of bidirectional teleportation. We acknowledge many insightful discussions
with Justin Champagne, Sumeet Khatri, Margarite LaBorde, Soorya Rethinasamy,
and Kunal Sharma. AUS\ acknowledges support from the LSU Discover Research
Grant, the National Science Foundation under Grant No.~OAC-1852454, and the
LSU\ Center for Computation and Technology. We also acknowledge support from
the National Science Foundation under Grant No.~1907615.
\end{acknowledgments}

\bibliographystyle{alpha}
\bibliography{Ref}

\appendix\onecolumngrid

\section{Proof of Proposition~\ref{prop:err-collapse-LOCC}}

\label{app:simplify-LOCC-err-swap}

\subsection{Exploiting symmetries of the unitary swap channel}

The main idea of the proof is to simplify the optimization problems in
\eqref{eq:sim-err-swap-DD} and \eqref{eq:sim-err-swap-ch-infid}\ by exploiting
the symmetries of the unitary swap channel, as given in
\eqref{eq:swap-ch-symmetry-for-LOCC} and
\eqref{eq:swap-ch-add-symmetry-for-LOCC}. To begin with, let us note that the
Choi operator of the SWAP\ channel can be written as
\begin{align}
\Gamma_{ABA^{\prime}B^{\prime}}^{\mathcal{S}^{d}}  &  =\sum_{i,j,k,\ell
}|i\rangle\!\langle j|_{A}\otimes|k\rangle\!\langle\ell|_{B}\otimes
|k\rangle\!\langle\ell|_{A^{\prime}}\otimes|i\rangle\!\langle j|_{B^{\prime}%
}\label{eq:Choi-swap-1}\\
&  =\Gamma_{AB^{\prime}}\otimes\Gamma_{BA^{\prime}}, \label{eq:Choi-swap-2}%
\end{align}
for orthonormal bases $\{|i\rangle_{A}\}_{i}$, $\{|k\rangle_{B}\}_{k}$,
$\{|k\rangle_{A^{\prime}}\}_{k}$, and $\{|i\rangle_{B^{\prime}}\}_{i}$. Let us
define the unitary channel $\mathcal{U}(\cdot)=U(\cdot)U^{\dag}$. To exploit
symmetries of the unitary SWAP\ channel, recall from
\eqref{eq:swap-ch-symmetry-for-LOCC} that the channel $\mathcal{S}_{AB}^{d}%
$\ is covariant in the following way:%
\begin{equation}
\mathcal{S}_{AB}^{d}=(\mathcal{V}_{A}^{\dag}\otimes\mathcal{U}_{B}^{\dag
})\circ\mathcal{S}_{AB}^{d}\circ(\mathcal{U}_{A}\otimes\mathcal{V}_{B}).
\label{eq:swap-symmetry-simplify-1}%
\end{equation}
Let $\mathcal{A}_{\hat{A}\hat{B}}^{\rho}$ denote the channel that appends the
bipartite state $\rho_{\hat{A}\hat{B}}$ to its input:%
\begin{equation}
\mathcal{A}_{\hat{A}\hat{B}}^{\rho}(\omega_{AB})=\omega_{AB}\otimes\rho
_{\hat{A}\hat{B}}.
\end{equation}

Let us start with the diamond distance, but note that the reasoning employed
in the first part of the proof (just below) applies equally well to channel
infidelity. Let $\mathcal{L}_{AB\hat{A}\hat{B}\rightarrow AB}$ be an arbitrary
LOCC\ channel to consider for the optimization problem in
\eqref{eq:sim-err-swap-DD}. Exploiting the unitary invariance of the diamond
distance with respect to input and output unitaries \cite[Proposition~3.44]%
{Watrous2018}, we find the following:%
\begin{align}
&  \left\Vert \mathcal{L}_{AB\hat{A}\hat{B}\rightarrow AB}\circ\mathcal{A}%
_{\hat{A}\hat{B}}^{\rho}-\mathcal{S}_{AB}^{d}\right\Vert _{\diamond
}\nonumber\\
&  =\left\Vert (\mathcal{V}_{A}^{\dag}\otimes\mathcal{U}_{B}^{\dag}%
)\circ\lbrack\mathcal{L}_{AB\hat{A}\hat{B}\rightarrow AB}\circ\mathcal{A}%
_{\hat{A}\hat{B}}^{\rho}-\mathcal{S}_{AB}^{d}]\circ(\mathcal{U}_{A}%
\otimes\mathcal{V}_{B})\right\Vert _{\diamond}\label{eq:app-swap-analysis-1}\\
&  =\left\Vert [(\mathcal{V}_{A}^{\dag}\otimes\mathcal{U}_{B}^{\dag}%
)\circ\mathcal{L}_{AB\hat{A}\hat{B}\rightarrow AB}\circ(\mathcal{U}_{A}%
\otimes\mathcal{V}_{B})]\circ\mathcal{A}_{\hat{A}\hat{B}}^{\rho}%
-(\mathcal{V}_{A}^{\dag}\otimes\mathcal{U}_{B}^{\dag})\circ\mathcal{S}%
_{AB}^{d}\circ(\mathcal{U}_{A}\otimes\mathcal{V}_{B})\right\Vert _{\diamond}\\
&  =\left\Vert [(\mathcal{V}_{A}^{\dag}\otimes\mathcal{U}_{B}^{\dag}%
)\circ\mathcal{L}_{AB\hat{A}\hat{B}\rightarrow AB}\circ(\mathcal{U}_{A}%
\otimes\mathcal{V}_{B})]\circ\mathcal{A}_{\hat{A}\hat{B}}^{\rho}%
-\mathcal{S}_{AB}^{d}\right\Vert _{\diamond}.
\end{align}
Thus, the channels $\mathcal{L}_{AB\hat{A}\hat{B}\rightarrow AB}$ and
$(\mathcal{V}_{A}^{\dag}\otimes\mathcal{U}_{B}^{\dag})\circ\mathcal{L}%
_{AB\hat{A}\hat{B}\rightarrow AB}\circ(\mathcal{U}_{A}\otimes\mathcal{V}_{B})$
perform equally well for the optimization. Now we exploit the convexity of the
diamond distance with respect to one of the channels \cite{Watrous2018}, as
well as the Haar probability measure over the unitary group, to conclude that%
\begin{align}
&  \left\Vert \mathcal{L}_{AB\hat{A}\hat{B}\rightarrow AB}\circ\mathcal{A}%
_{\hat{A}\hat{B}}^{\rho}-\mathcal{S}_{AB}^{d}\right\Vert _{\diamond
}\nonumber\\
&  =\int\int dU\ dV\ \left\Vert [(\mathcal{V}_{A}^{\dag}\otimes\mathcal{U}%
_{B}^{\dag})\circ\mathcal{L}_{AB\hat{A}\hat{B}\rightarrow AB}\circ
(\mathcal{U}_{A}\otimes\mathcal{V}_{B})]\circ\mathcal{A}_{\hat{A}\hat{B}%
}^{\rho}-\mathcal{S}_{AB}^{d}\right\Vert _{\diamond}\\
&  \geq\ \left\Vert \widetilde{\mathcal{L}}_{AB\hat{A}\hat{B}\rightarrow
AB}\circ\mathcal{A}_{\hat{A}\hat{B}}^{\rho}-\mathcal{S}_{AB}^{d}\right\Vert
_{\diamond},
\end{align}
where%
\begin{equation}
\widetilde{\mathcal{L}}_{AB\hat{A}\hat{B}\rightarrow AB}\coloneqq \int\int
dU\ dV\ (\mathcal{V}_{A}^{\dag}\otimes\mathcal{U}_{B}^{\dag})\circ
\mathcal{L}_{AB\hat{A}\hat{B}\rightarrow AB}\circ(\mathcal{U}_{A}%
\otimes\mathcal{V}_{B}). \label{eq:app-swap-analysis-last}%
\end{equation}
Thus, we conclude that it suffices to optimize \eqref{eq:c-ppt-p-error} over
LOCC channels that possess this symmetry. Critical to this argument is the
observation that the channel twirl in \eqref{eq:app-swap-analysis-last} can be
realized by LOCC, so that $\widetilde{\mathcal{L}}_{AB\hat{A}\hat
{B}\rightarrow AB}$ is an LOCC channel if $\mathcal{L}_{AB\hat{A}\hat
{B}\rightarrow AB}$ is.

What is the form of LOCC\ channels possessing this symmetry? Let $L_{AB\hat
{A}\hat{B}A^{\prime}B^{\prime}}$ denote the Choi operator of the channel
$\mathcal{L}_{AB\hat{A}\hat{B}\rightarrow AB}$, where $A^{\prime}$ and
$B^{\prime}$ denote the output systems. Then consider that the Choi operator
$\widetilde{L}_{AB\hat{A}\hat{B}A^{\prime}B^{\prime}}$\ of the double twirled
channel $\widetilde{\mathcal{L}}_{AB\hat{A}\hat{B}\rightarrow AB}$ is as
follows:%
\begin{equation}
\widetilde{L}_{AB\hat{A}\hat{B}A^{\prime}B^{\prime}}\coloneqq \int\int
dU\ dV\ (\mathcal{U}_{A}\otimes\mathcal{V}_{B}\otimes\overline{\mathcal{V}%
}_{A^{\prime}}\otimes\overline{\mathcal{U}}_{B^{\prime}})(L_{AB\hat{A}\hat
{B}A^{\prime}B^{\prime}}), \label{eq:twirled-choi}%
\end{equation}
where $\overline{\mathcal{U}}(\cdot)=\overline{U}(\cdot)U^{T}$ (with the
overbar denoting complex conjugate). Now recall the following identity from
\cite{W89,Horodecki99,Watrous2018}:%
\begin{align}
\widetilde{\mathcal{T}}_{CD}(X_{CD})  &  \coloneqq \int dU\ (\mathcal{U}%
_{C}\otimes\overline{\mathcal{U}}_{D})(X_{CD})\nonumber\\
&  =\Phi_{CD}\operatorname{Tr}_{CD}[\Phi_{CD}X_{CD}]+\frac{I_{CD}-\Phi_{CD}%
}{d^{2}-1}\operatorname{Tr}_{CD}[(I_{CD}-\Phi_{CD})X_{CD}].
\end{align}
We apply this to \eqref{eq:twirled-choi} to write it as%
\begin{align}
\widetilde{L}_{AB\hat{A}\hat{B}A^{\prime}B^{\prime}}  &  =(\widetilde
{\mathcal{T}}_{AB^{\prime}}\otimes\widetilde{\mathcal{T}}_{BA^{\prime}%
})(L_{AB\hat{A}\hat{B}A^{\prime}B^{\prime}})\nonumber\\
&  =(\Phi_{AB^{\prime}}\otimes\Phi_{BA^{\prime}})\operatorname{Tr}%
_{ABA^{\prime}B^{\prime}}[(\Phi_{AB^{\prime}}\otimes\Phi_{BA^{\prime}%
})L_{AB\hat{A}\hat{B}A^{\prime}B^{\prime}}]\nonumber\\
&  \qquad+\left(  \Phi_{AB^{\prime}}\otimes\frac{I_{BA^{\prime}}%
-\Phi_{BA^{\prime}}}{d^{2}-1}\right)  \operatorname{Tr}_{ABA^{\prime}%
B^{\prime}}[\left(  \Phi_{AB^{\prime}}\otimes\lbrack I_{BA^{\prime}}%
-\Phi_{BA^{\prime}}]\right)  L_{AB\hat{A}\hat{B}A^{\prime}B^{\prime}%
}]\nonumber\\
&  \qquad+\left(  \frac{I_{AB^{\prime}}-\Phi_{AB^{\prime}}}{d^{2}-1}%
\otimes\Phi_{BA^{\prime}}\right)  \operatorname{Tr}_{ABA^{\prime}B^{\prime}%
}[\left(  [I_{AB^{\prime}}-\Phi_{AB^{\prime}}]\otimes\Phi_{BA^{\prime}%
}\right)  L_{AB\hat{A}\hat{B}A^{\prime}B^{\prime}}]\nonumber\\
&  \qquad+\left(  \frac{I_{AB^{\prime}}-\Phi_{AB^{\prime}}}{d^{2}-1}%
\otimes\frac{I_{BA^{\prime}}-\Phi_{BA^{\prime}}}{d^{2}-1}\right)
\operatorname{Tr}_{ABA^{\prime}B^{\prime}}[\left(  [I_{AB^{\prime}}%
-\Phi_{AB^{\prime}}]\otimes\lbrack I_{BA^{\prime}}-\Phi_{BA^{\prime}}]\right)
L_{AB\hat{A}\hat{B}A^{\prime}B^{\prime}}].\nonumber
\end{align}
Now defining%
\begin{align}
K_{\hat{A}\hat{B}}^{\prime}  &  \coloneqq \operatorname{Tr}_{ABA^{\prime
}B^{\prime}}[(\Phi_{AB^{\prime}}\otimes\Phi_{BA^{\prime}})L_{AB\hat{A}\hat
{B}A^{\prime}B^{\prime}}]\\
L_{\hat{A}\hat{B}}^{\prime}  &  \coloneqq \operatorname{Tr}_{ABA^{\prime
}B^{\prime}}[\left(  \Phi_{AB^{\prime}}\otimes\lbrack I_{BA^{\prime}}%
-\Phi_{BA^{\prime}}]\right)  L_{AB\hat{A}\hat{B}A^{\prime}B^{\prime}}]\\
M_{\hat{A}\hat{B}}^{\prime}  &  \coloneqq \operatorname{Tr}_{ABA^{\prime
}B^{\prime}}[\left(  [I_{AB^{\prime}}-\Phi_{AB^{\prime}}]\otimes
\Phi_{BA^{\prime}}\right)  L_{AB\hat{A}\hat{B}A^{\prime}B^{\prime}}]\\
N_{\hat{A}\hat{B}}^{\prime}  &  \coloneqq \operatorname{Tr}_{ABA^{\prime
}B^{\prime}}[\left(  [I_{AB^{\prime}}-\Phi_{AB^{\prime}}]\otimes\lbrack
I_{BA^{\prime}}-\Phi_{BA^{\prime}}]\right)  L_{AB\hat{A}\hat{B}A^{\prime
}B^{\prime}}],
\end{align}
we can write%
\begin{multline}
\widetilde{L}_{AB\hat{A}\hat{B}A^{\prime}B^{\prime}}=\Phi_{AB^{\prime}}%
\otimes\Phi_{BA^{\prime}}\otimes K_{\hat{A}\hat{B}}^{\prime}+\Phi_{AB^{\prime
}}\otimes\frac{I_{BA^{\prime}}-\Phi_{BA^{\prime}}}{d^{2}-1}\otimes L_{\hat
{A}\hat{B}}^{\prime}\label{eq:twirled-choi-nice}\\
+\frac{I_{AB^{\prime}}-\Phi_{AB^{\prime}}}{d^{2}-1}\otimes\Phi_{BA^{\prime}%
}\otimes M_{\hat{A}\hat{B}}^{\prime}+\frac{I_{AB^{\prime}}-\Phi_{AB^{\prime}}%
}{d^{2}-1}\otimes\frac{I_{BA^{\prime}}-\Phi_{BA^{\prime}}}{d^{2}-1}\otimes
N_{\hat{A}\hat{B}}^{\prime}.
\end{multline}
What are the conditions on the operators $K_{\hat{A}\hat{B}}^{\prime}$,
$L_{\hat{A}\hat{B}}^{\prime}$, $M_{\hat{A}\hat{B}}^{\prime}$, and $N_{\hat
{A}\hat{B}}^{\prime}$? In order for $\widetilde{L}_{AB\hat{A}\hat{B}A^{\prime
}B^{\prime}}$ to be the Choi operator of a channel, the following conditions
should hold, thus imposing conditions on the operators $K_{\hat{A}\hat{B}%
}^{\prime}$, $L_{\hat{A}\hat{B}}^{\prime}$, $M_{\hat{A}\hat{B}}^{\prime}$, and
$N_{\hat{A}\hat{B}}^{\prime}$:%
\begin{align}
\widetilde{L}_{AB\hat{A}\hat{B}A^{\prime}B^{\prime}}  &  \geq
0,\label{eq:psd-constr}\\
\operatorname{Tr}_{A^{\prime}B^{\prime}}[\widetilde{L}_{AB\hat{A}\hat
{B}A^{\prime}B^{\prime}}]  &  =I_{AB\hat{A}\hat{B}}. \label{eq:ch-constr}%
\end{align}
The first condition in \eqref{eq:psd-constr} imposes that%
\begin{equation}
K_{\hat{A}\hat{B}}^{\prime},L_{\hat{A}\hat{B}}^{\prime},M_{\hat{A}\hat{B}%
}^{\prime},N_{\hat{A}\hat{B}}^{\prime}\geq0. \label{eq:KLMN-PSD}%
\end{equation}
The second condition in \eqref{eq:ch-constr} imposes that%
\begin{equation}
\pi_{AB}\otimes K_{\hat{A}\hat{B}}^{\prime}+\pi_{AB}\otimes L_{\hat{A}\hat{B}%
}^{\prime}+\pi_{AB}\otimes M_{\hat{A}\hat{B}}^{\prime}+\pi_{AB}\otimes
N_{\hat{A}\hat{B}}^{\prime}=I_{AB\hat{A}\hat{B}},
\end{equation}
which is the same as%
\begin{equation}
K_{\hat{A}\hat{B}}^{\prime}+L_{\hat{A}\hat{B}}^{\prime}+M_{\hat{A}\hat{B}%
}^{\prime}+N_{\hat{A}\hat{B}}^{\prime}=d^{2}I_{\hat{A}\hat{B}}.
\label{eq:normalize-KLMN-prime}%
\end{equation}
Thus, we can think of the operators $K_{\hat{A}\hat{B}}^{\prime}$, $L_{\hat
{A}\hat{B}}^{\prime}$, $M_{\hat{A}\hat{B}}^{\prime}$, and $N_{\hat{A}\hat{B}%
}^{\prime}$ normalized by $d^{2}$ as measurement operators, which gives an
interesting physical interpretation to them. Let us then define%
\begin{equation}
K_{\hat{A}\hat{B}}\coloneqq \frac{1}{d^{2}}K_{\hat{A}\hat{B}}^{\prime},\qquad
L_{\hat{A}\hat{B}}\coloneqq \frac{1}{d^{2}}L_{\hat{A}\hat{B}}^{\prime},\qquad
M_{\hat{A}\hat{B}}\coloneqq \frac{1}{d^{2}}M_{\hat{A}\hat{B}}^{\prime},\qquad
N_{\hat{A}\hat{B}}\coloneqq \frac{1}{d^{2}}N_{\hat{A}\hat{B}}^{\prime},
\end{equation}
and note that \eqref{eq:normalize-KLMN-prime} is equivalent to%
\begin{equation}
K_{\hat{A}\hat{B}}+L_{\hat{A}\hat{B}}+M_{\hat{A}\hat{B}}+N_{\hat{A}\hat{B}%
}=I_{\hat{A}\hat{B}},
\end{equation}
\eqref{eq:KLMN-PSD} is equivalent to%
\begin{equation}
K_{\hat{A}\hat{B}},L_{\hat{A}\hat{B}},M_{\hat{A}\hat{B}},N_{\hat{A}\hat{B}%
}\geq0,
\end{equation}
and \eqref{eq:twirled-choi-nice}\ is equivalent to%
\begin{multline}
\widetilde{L}_{AB\hat{A}\hat{B}A^{\prime}B^{\prime}}=\Gamma_{AB^{\prime}%
}\otimes\Gamma_{BA^{\prime}}\otimes K_{\hat{A}\hat{B}}+\Gamma_{AB^{\prime}%
}\otimes\frac{dI_{BA^{\prime}}-\Gamma_{BA^{\prime}}}{d^{2}-1}\otimes
L_{\hat{A}\hat{B}}\label{eq:P-choi-simplified}\\
+\frac{dI_{AB^{\prime}}-\Gamma_{AB^{\prime}}}{d^{2}-1}\otimes\Gamma
_{BA^{\prime}}\otimes M_{\hat{A}\hat{B}}+\frac{dI_{AB^{\prime}}-\Gamma
_{AB^{\prime}}}{d^{2}-1}\otimes\frac{dI_{BA^{\prime}}-\Gamma_{BA^{\prime}}%
}{d^{2}-1}\otimes N_{\hat{A}\hat{B}}.
\end{multline}
Now consider that the Choi operator of the composite channel $\widetilde
{\mathcal{L}}_{AB\hat{A}\hat{B}\rightarrow AB}\circ\mathcal{A}_{\hat{A}\hat
{B}}^{\rho}$ is given by%
\begin{multline}
\operatorname{Tr}_{\hat{A}\hat{B}}[T_{\hat{A}\hat{B}}(\rho_{\hat{A}\hat{B}%
})\widetilde{L}_{AB\hat{A}\hat{B}A^{\prime}B^{\prime}}]=\\
\Gamma_{AB^{\prime}}\otimes\Gamma_{BA^{\prime}}\operatorname{Tr}[T_{\hat
{A}\hat{B}}(\rho_{\hat{A}\hat{B}})K_{\hat{A}\hat{B}}]+\Gamma_{AB^{\prime}%
}\otimes\frac{dI_{BA^{\prime}}-\Gamma_{BA^{\prime}}}{d^{2}-1}\operatorname{Tr}%
[T_{\hat{A}\hat{B}}(\rho_{\hat{A}\hat{B}})L_{\hat{A}\hat{B}}]\\
+\frac{dI_{AB^{\prime}}-\Gamma_{AB^{\prime}}}{d^{2}-1}\otimes\Gamma
_{BA^{\prime}}\operatorname{Tr}[T_{\hat{A}\hat{B}}(\rho_{\hat{A}\hat{B}%
})M_{\hat{A}\hat{B}}]+\frac{dI_{AB^{\prime}}-\Gamma_{AB^{\prime}}}{d^{2}%
-1}\otimes\frac{dI_{BA^{\prime}}-\Gamma_{BA^{\prime}}}{d^{2}-1}%
\operatorname{Tr}[T_{\hat{A}\hat{B}}(\rho_{\hat{A}\hat{B}})N_{\hat{A}\hat{B}%
}].
\end{multline}
By making the substitutions $K_{\hat{A}\hat{B}}\rightarrow T_{\hat{A}\hat{B}%
}(K_{\hat{A}\hat{B}})$, $L_{\hat{A}\hat{B}}\rightarrow T_{\hat{A}\hat{B}%
}(L_{\hat{A}\hat{B}})$, $M_{\hat{A}\hat{B}}\rightarrow T_{\hat{A}\hat{B}%
}(M_{\hat{A}\hat{B}})$, and $N_{\hat{A}\hat{B}}\rightarrow T_{\hat{A}\hat{B}%
}(N_{\hat{A}\hat{B}})$, and using the facts that%
\begin{equation}
K_{\hat{A}\hat{B}}\geq0\qquad\Longleftrightarrow\qquad T_{\hat{A}\hat{B}%
}(K_{\hat{A}\hat{B}})\geq0,
\end{equation}
and the same for $L_{\hat{A}\hat{B}}$, $M_{\hat{A}\hat{B}}$, and $N_{\hat
{A}\hat{B}}$, as well as%
\begin{multline}
K_{\hat{A}\hat{B}}+L_{\hat{A}\hat{B}}+M_{\hat{A}\hat{B}}+N_{\hat{A}\hat{B}%
}=I_{\hat{A}\hat{B}}\qquad\Longleftrightarrow\\
\qquad T_{\hat{A}\hat{B}}(K_{\hat{A}\hat{B}})+T_{\hat{A}\hat{B}}(L_{\hat
{A}\hat{B}})+T_{\hat{A}\hat{B}}(M_{\hat{A}\hat{B}})+T_{\hat{A}\hat{B}}%
(N_{\hat{A}\hat{B}})=T_{\hat{A}\hat{B}}(I_{\hat{A}\hat{B}})=I_{\hat{A}\hat{B}%
},
\end{multline}
and the fact that the channel $\widetilde{\mathcal{L}}_{AB\hat{A}\hat
{B}\rightarrow AB}$ remains LOCC under these changes, we conclude that the
optimization problem does not change if we make these substitions. Thus, we
can take the Choi operator of $\widetilde{\mathcal{L}}_{AB\hat{A}\hat
{B}\rightarrow AB}\circ\mathcal{A}_{\hat{A}\hat{B}}^{\rho}$ to be%
\begin{multline}
\Gamma_{AB^{\prime}}\otimes\Gamma_{BA^{\prime}}\operatorname{Tr}[\rho_{\hat
{A}\hat{B}}K_{\hat{A}\hat{B}}]+\Gamma_{AB^{\prime}}\otimes\frac{dI_{BA^{\prime
}}-\Gamma_{BA^{\prime}}}{d^{2}-1}\operatorname{Tr}[\rho_{\hat{A}\hat{B}%
}L_{\hat{A}\hat{B}}]\label{eq:app-choi-op-LOCC-sim-swap-23}\\
+\frac{dI_{AB^{\prime}}-\Gamma_{AB^{\prime}}}{d^{2}-1}\otimes\Gamma
_{BA^{\prime}}\operatorname{Tr}[\rho_{\hat{A}\hat{B}}M_{\hat{A}\hat{B}}%
]+\frac{dI_{AB^{\prime}}-\Gamma_{AB^{\prime}}}{d^{2}-1}\otimes\frac
{dI_{BA^{\prime}}-\Gamma_{BA^{\prime}}}{d^{2}-1}\operatorname{Tr}[\rho
_{\hat{A}\hat{B}}N_{\hat{A}\hat{B}}].
\end{multline}

Regarding the observation in \eqref{eq:sim-ch-LOCC-opt}, consider from
\eqref{eq:app-choi-op-LOCC-sim-swap-23} that the Choi operator of an optimal
LOCC\ simulating channel acting on the resource state $\rho_{\hat{A}\hat{B}}%
$\ is as follows:%
\begin{multline}
\Gamma_{AB^{\prime}}\otimes\Gamma_{BA^{\prime}}\operatorname{Tr}[\rho_{\hat
{A}\hat{B}}K_{\hat{A}\hat{B}}]+\Gamma_{AB^{\prime}}\otimes\frac{dI_{BA^{\prime
}}-\Gamma_{BA^{\prime}}}{d^{2}-1}\operatorname{Tr}[\rho_{\hat{A}\hat{B}%
}L_{\hat{A}\hat{B}}]\\
+\frac{dI_{AB^{\prime}}-\Gamma_{AB^{\prime}}}{d^{2}-1}\otimes\Gamma
_{BA^{\prime}}\operatorname{Tr}[\rho_{\hat{A}\hat{B}}M_{\hat{A}\hat{B}}%
]+\frac{dI_{AB^{\prime}}-\Gamma_{AB^{\prime}}}{d^{2}-1}\otimes\frac
{dI_{BA^{\prime}}-\Gamma_{BA^{\prime}}}{d^{2}-1}\operatorname{Tr}[\rho
_{\hat{A}\hat{B}}N_{\hat{A}\hat{B}}],
\end{multline}
where the operators $K_{\hat{A}\hat{B}}$, $L_{\hat{A}\hat{B}}$, $M_{\hat
{A}\hat{B}}$, and $N_{\hat{A}\hat{B}}$ obey the constraints%
\begin{align}
K_{\hat{A}\hat{B}},L_{\hat{A}\hat{B}},M_{\hat{A}\hat{B}},N_{\hat{A}\hat{B}}
&  \geq0,\\
K_{\hat{A}\hat{B}}+L_{\hat{A}\hat{B}}+M_{\hat{A}\hat{B}}+N_{\hat{A}\hat{B}}
&  =I_{\hat{A}\hat{B}},
\end{align}
such that%
\begin{multline}
\Gamma_{AB^{\prime}}\otimes\Gamma_{BA^{\prime}}\otimes T_{\hat{A}\hat{B}%
}(K_{\hat{A}\hat{B}})+\Gamma_{AB^{\prime}}\otimes\frac{dI_{BA^{\prime}}%
-\Gamma_{BA^{\prime}}}{d^{2}-1}\otimes T_{\hat{A}\hat{B}}(L_{\hat{A}\hat{B}%
})\label{eq:Choi-op-final-LOCC-app-1}\\
+\frac{dI_{AB^{\prime}}-\Gamma_{AB^{\prime}}}{d^{2}-1}\otimes\Gamma
_{BA^{\prime}}\otimes T_{\hat{A}\hat{B}}(M_{\hat{A}\hat{B}})+\frac
{dI_{AB^{\prime}}-\Gamma_{AB^{\prime}}}{d^{2}-1}\otimes\frac{dI_{BA^{\prime}%
}-\Gamma_{BA^{\prime}}}{d^{2}-1}\otimes T_{\hat{A}\hat{B}}(N_{\hat{A}\hat{B}})
\end{multline}
is the Choi operator of an LOCC channel. Thus, the set $\{K_{\hat{A}\hat{B}%
},L_{\hat{A}\hat{B}},M_{\hat{A}\hat{B}},N_{\hat{A}\hat{B}}\}$ constitutes a
POVM. Given that $\Gamma_{AB^{\prime}}$ is the Choi operator of an identity
channel from $A$ to $B^{\prime}$, $\Gamma_{BA^{\prime}}$ is the Choi operator
of an identity channel from $B$ to $A^{\prime}$, and $\frac{dI_{BA^{\prime}%
}-\Gamma_{BA^{\prime}}}{d^{2}-1}$ is the Choi operator of the generalized
Pauli channel in \eqref{eq:gen-Pauli-channel}, the interpretation in
Remark~\ref{rem:channel-form} follows. The last observation about
$\frac{dI_{BA^{\prime}}-\Gamma_{BA^{\prime}}}{d^{2}-1}$ follows because%
\begin{align}
\frac{dI_{BA^{\prime}}-\Gamma_{BA^{\prime}}}{d^{2}-1}  &  =\frac{d}{d^{2}%
-1}\left(  I_{BA^{\prime}}-\frac{1}{d}\Gamma_{BA^{\prime}}\right) \\
&  =\frac{d}{d^{2}-1}\left(  I_{BA^{\prime}}-\Phi_{BA^{\prime}}\right) \\
&  =\frac{d}{d^{2}-1}\sum_{\left(  x,z\right)  \neq\left(  0,0\right)
}W_{A^{\prime}}^{z,x}\Phi_{BA^{\prime}}(W_{A^{\prime}}^{z,x})^{\dag}\\
&  =\frac{1}{d^{2}-1}\sum_{\left(  x,z\right)  \neq\left(  0,0\right)
}W_{A^{\prime}}^{z,x}\Gamma_{BA^{\prime}}(W_{A^{\prime}}^{z,x})^{\dag}.
\end{align}

Finally, we exploit the symmetry mentioned in
\eqref{eq:swap-ch-add-symmetry-for-LOCC}. Namely, the swap channel commutes
with itself. Then we can follow the same reasoning given in
\eqref{eq:app-swap-analysis-1}--\eqref{eq:app-swap-analysis-last} to conclude
that an optimal LOCC\ channel should obey the following symmetry as well:%
\begin{equation}
\widetilde{\mathcal{L}}_{AB\hat{A}\hat{B}\rightarrow AB}=\frac{1}{2}\left(
\widetilde{\mathcal{L}}_{AB\hat{A}\hat{B}\rightarrow AB}+\mathcal{S}_{AB}%
^{d}\circ\widetilde{\mathcal{L}}_{AB\hat{A}\hat{B}\rightarrow AB}%
\circ\mathcal{S}_{AB}^{d}\right)  .
\end{equation}
Equivalently, its Choi operator $\widetilde{L}_{AB\hat{A}\hat{B}A^{\prime
}B^{\prime}}$ should satisfy%
\begin{equation}
\widetilde{L}_{AB\hat{A}\hat{B}A^{\prime}B^{\prime}}=\frac{1}{2}\left(
\widetilde{L}_{AB\hat{A}\hat{B}A^{\prime}B^{\prime}}+\left(  \mathcal{S}%
_{AB}^{d}\otimes\mathcal{S}_{A^{\prime}B^{\prime}}^{d}\right)  \left(
\widetilde{L}_{AB\hat{A}\hat{B}A^{\prime}B^{\prime}}\right)  \right)  .
\end{equation}
Applying this to \eqref{eq:Choi-op-final-LOCC-app-1}, we conclude that it has
the form%
\begin{multline}
\Gamma_{AB^{\prime}}\otimes\Gamma_{BA^{\prime}}\otimes T_{\hat{A}\hat{B}%
}(K_{\hat{A}\hat{B}})+\frac{1}{2}\left(  \Gamma_{AB^{\prime}}\otimes
\frac{dI_{BA^{\prime}}-\Gamma_{BA^{\prime}}}{d^{2}-1}+\frac{dI_{AB^{\prime}%
}-\Gamma_{AB^{\prime}}}{d^{2}-1}\otimes\Gamma_{BA^{\prime}}\right)  \otimes
T_{\hat{A}\hat{B}}(L_{\hat{A}\hat{B}}+M_{\hat{A}\hat{B}})\\
+\frac{dI_{AB^{\prime}}-\Gamma_{AB^{\prime}}}{d^{2}-1}\otimes\frac
{dI_{BA^{\prime}}-\Gamma_{BA^{\prime}}}{d^{2}-1}\otimes T_{\hat{A}\hat{B}%
}(N_{\hat{A}\hat{B}}).
\end{multline}
After defining $L_{\hat{A}\hat{B}}+M_{\hat{A}\hat{B}}$ as $L_{\hat{A}\hat{B}}%
$, we obtain the form stated in Proposition~\ref{prop:err-collapse-LOCC}.

Note that the swap itself cannot be implemented by LOCC. However, once we have
the reduction of an optimal LOCC\ to the form in
\eqref{eq:Choi-op-final-LOCC-app-1}, this corresponds to a channel of the
following form%
\begin{multline}
\mathcal{S}_{AB}^{d}(\operatorname{Tr}_{\hat{A}\hat{B}}[K_{\hat{A}\hat{B}%
}\omega_{AB\hat{A}\hat{B}}])+\left(  \operatorname{id}_{A\rightarrow B}%
\otimes\mathcal{D}_{B\rightarrow A}\right)  \operatorname{Tr}_{\hat{A}\hat{B}%
}[L_{\hat{A}\hat{B}}\omega_{AB\hat{A}\hat{B}}] +\left(  \mathcal{D}%
_{A\rightarrow B}\otimes\operatorname{id}_{B\rightarrow A}\right)
\operatorname{Tr}_{\hat{A}\hat{B}}[M_{\hat{A}\hat{B}}\omega_{AB\hat{A}\hat{B}%
}]\\
+\left(  \mathcal{D}_{A\rightarrow B}\otimes\mathcal{D}_{B\rightarrow
A}\right)  (\operatorname{Tr}_{\hat{A}\hat{B}}[N_{\hat{A}\hat{B}}%
\omega_{AB\hat{A}\hat{B}}]).
\end{multline}
Thus, the variant of the channel corresponding to $\mathcal{S}_{AB}^{d}%
\circ\widetilde{\mathcal{L}}_{AB\hat{A}\hat{B}\rightarrow AB}\circ
\mathcal{S}_{AB}^{d}$ is as follows:%
\begin{multline}
\mathcal{S}_{AB}^{d}(\operatorname{Tr}_{\hat{A}\hat{B}}[K_{\hat{A}\hat{B}%
}\omega_{AB\hat{A}\hat{B}}]) +\left(  \mathcal{D}_{A\rightarrow B}%
\otimes\operatorname{id}_{B\rightarrow A}\right)  \operatorname{Tr}_{\hat
{A}\hat{B}}[L_{\hat{A}\hat{B}}\omega_{AB\hat{A}\hat{B}}] +\left(
\operatorname{id}_{A\rightarrow B}\otimes\mathcal{D}_{B\rightarrow A}\right)
\operatorname{Tr}_{\hat{A}\hat{B}}[M_{\hat{A}\hat{B}}\omega_{AB\hat{A}\hat{B}%
}]\\
+\left(  \mathcal{D}_{A\rightarrow B}\otimes\mathcal{D}_{B\rightarrow
A}\right)  (\operatorname{Tr}_{\hat{A}\hat{B}}[N_{\hat{A}\hat{B}}%
\omega_{AB\hat{A}\hat{B}}]),
\end{multline}
which is still LOCC\ if the original channel is, because it is related to the
original merely by Alice and Bob flipping their local actions in the case of
the $L_{\hat{A}\hat{B}}$ and $M_{\hat{A}\hat{B}}$ measurement outcomes. When
we randomly apply either of these channels, the random mixture is LOCC and
corresponds to%
\begin{multline}
\mathcal{S}_{AB}^{d}(\operatorname{Tr}_{\hat{A}\hat{B}}[K_{\hat{A}\hat{B}%
}\omega_{AB\hat{A}\hat{B}}]) +\frac{1}{2}\left(  \mathcal{D}_{A\rightarrow
B}\otimes\operatorname{id}_{B\rightarrow A}+\operatorname{id}_{A\rightarrow
B}\otimes\mathcal{D}_{B\rightarrow A}\right)  \operatorname{Tr}_{\hat{A}%
\hat{B}}[\left(  L_{\hat{A}\hat{B}}+M_{\hat{A}\hat{B}}\right)  \omega
_{AB\hat{A}\hat{B}}]\\
+\left(  \mathcal{D}_{A\rightarrow B}\otimes\mathcal{D}_{B\rightarrow
A}\right)  (\operatorname{Tr}_{\hat{A}\hat{B}}[N_{\hat{A}\hat{B}}%
\omega_{AB\hat{A}\hat{B}}]),
\end{multline}
which allows us to lump together the measurement outcomes for $L_{\hat{A}%
\hat{B}}$ and $M_{\hat{A}\hat{B}}$ into a single measurement outcome.

\subsection{Evaluating normalized diamond distance}

We have now reduced the optimization problem in \eqref{eq:sim-err-swap-DD},
for the swap channel, to the following one:%
\begin{equation}
e_{\operatorname{LOCC}}(\mathcal{S}_{AB}^{d},\rho_{\hat{A}\hat{B}}%
)=\inf_{\widetilde{\mathcal{L}}_{AB\hat{A}\hat{B}\rightarrow AB}%
\in\operatorname{LOCC}}\frac{1}{2}\left\Vert \widetilde{\mathcal{L}}%
_{AB\hat{A}\hat{B}\rightarrow AB}\circ\mathcal{A}_{\hat{A}\hat{B}}^{\rho
}-\mathcal{S}_{AB}^{d}\right\Vert _{\diamond}
\label{eq:dd-obj-func-for-swap-sim}%
\end{equation}
subject to%
\begin{multline}
\widetilde{\mathcal{L}}_{AB\hat{A}\hat{B}\rightarrow AB}(\omega_{AB}%
\otimes\rho_{\hat{A}\hat{B}})=\mathcal{S}_{AB}^{d}(\omega_{AB}%
)\operatorname{Tr}[K_{\hat{A}\hat{B}}\rho_{\hat{A}\hat{B}}%
]\label{eq:app:locc-ch-symmetrized-simplified}\\
+\frac{1}{2}\left(  \operatorname{id}_{A\rightarrow B}\otimes\mathcal{D}%
_{B\rightarrow A}+\mathcal{D}_{A\rightarrow B}\otimes\operatorname{id}%
_{B\rightarrow A}\right)  (\omega_{AB})\operatorname{Tr}[L_{\hat{A}\hat{B}%
}\rho_{\hat{A}\hat{B}}]\\
+\left(  \mathcal{D}_{A\rightarrow B}\otimes\mathcal{D}_{B\rightarrow
A}\right)  (\omega_{AB})\operatorname{Tr}[N_{\hat{A}\hat{B}}\rho_{\hat{A}%
\hat{B}}],
\end{multline}%
\begin{align}
K_{\hat{A}\hat{B}},L_{\hat{A}\hat{B}},N_{\hat{A}\hat{B}}  &  \geq0,\\
K_{\hat{A}\hat{B}}+L_{\hat{A}\hat{B}}+N_{\hat{A}\hat{B}}  &  =I_{\hat{A}%
\hat{B}}.
\end{align}
Keep in mind that the operators $K_{\hat{A}\hat{B}}$, $L_{\hat{A}\hat{B}}$,
and $N_{\hat{A}\hat{B}}$ are further constrained so that $\widetilde
{\mathcal{L}}_{AB\hat{A}\hat{B}\rightarrow AB}\in\operatorname{LOCC}$, as
indicated in \eqref{eq:dd-obj-func-for-swap-sim}. We can exploit the form of
the optimization of the diamond distance from \eqref{eq:SDP-diamond}\ to
rewrite the optimization for the simulation error of a unitary swap channel as
follows:%
\begin{equation}
\inf_{\mu,Z_{ABA^{\prime}B^{\prime}},K_{\hat{A}\hat{B}},L_{\hat{A}\hat{B}%
},N_{\hat{A}\hat{B}}\geq0}\mu, \label{eq:final-steps-swap-simp-1}%
\end{equation}
subject to%
\begin{align}
\mu I_{AB}  &  \geq Z_{AB},\\
Z_{ABA^{\prime}B^{\prime}}  &  \geq\Gamma_{AB^{\prime}}\otimes\Gamma
_{BA^{\prime}}\left(  1-\operatorname{Tr}[\rho_{\hat{A}\hat{B}}K_{\hat{A}%
\hat{B}}]\right) \nonumber\\
&  \qquad-\frac{1}{2}\left(  \Gamma_{AB^{\prime}}\otimes\frac{dI_{BA^{\prime}%
}-\Gamma_{BA^{\prime}}}{d^{2}-1}+\frac{dI_{AB^{\prime}}-\Gamma_{AB^{\prime}}%
}{d^{2}-1}\otimes\Gamma_{BA^{\prime}}\right)  \operatorname{Tr}[\rho_{\hat
{A}\hat{B}}L_{\hat{A}\hat{B}}]\nonumber\\
&  \qquad-\frac{dI_{AB^{\prime}}-\Gamma_{AB^{\prime}}}{d^{2}-1}\otimes
\frac{dI_{BA^{\prime}}-\Gamma_{BA^{\prime}}}{d^{2}-1}\operatorname{Tr}%
[\rho_{\hat{A}\hat{B}}N_{\hat{A}\hat{B}}], \label{eq:Z-ineq-SDP}%
\end{align}%
\begin{equation}
K_{\hat{A}\hat{B}}+L_{\hat{A}\hat{B}}+N_{\hat{A}\hat{B}}=I_{\hat{A}\hat{B}},
\end{equation}
subject to the channel $\widetilde{\mathcal{L}}_{AB\hat{A}\hat{B}\rightarrow
AB}$ in \eqref{eq:app:locc-ch-symmetrized-simplified}\ being LOCC. Finally,
since we are trying to minimize with respect to $\mu$ and $Z_{ABA^{\prime
}B^{\prime}}$, we can choose $Z_{ABA^{\prime}B^{\prime}}$ to be the smallest
positive semi-definite operator such that the inequality in
\eqref{eq:Z-ineq-SDP} is satisfied. This is the positive part of the operator
on the right-hand side of the inequality. Since the operator on the right-hand
side has the following Jordan--Hahn decomposition%
\begin{equation}
\Gamma_{AB^{\prime}}\otimes\Gamma_{BA^{\prime}}\left(  1-\operatorname{Tr}%
[\rho_{\hat{A}\hat{B}}K_{\hat{A}\hat{B}}]\right)  -\left[
\begin{array}
[c]{c}%
\frac{1}{2}\left(  \Gamma_{AB^{\prime}}\otimes\frac{dI_{BA^{\prime}}%
-\Gamma_{BA^{\prime}}}{d^{2}-1}+\frac{dI_{AB^{\prime}}-\Gamma_{AB^{\prime}}%
}{d^{2}-1}\otimes\Gamma_{BA^{\prime}}\right)  \operatorname{Tr}[\rho_{\hat
{A}\hat{B}}L_{\hat{A}\hat{B}}]\\
+\frac{dI_{AB^{\prime}}-\Gamma_{AB^{\prime}}}{d^{2}-1}\otimes\frac
{dI_{BA^{\prime}}-\Gamma_{BA^{\prime}}}{d^{2}-1}\operatorname{Tr}[\rho
_{\hat{A}\hat{B}}N_{\hat{A}\hat{B}}]
\end{array}
\right]  ,
\end{equation}
it follows that its positive part is given by%
\begin{equation}
(1-\operatorname{Tr}[\rho_{\hat{A}\hat{B}}K_{\hat{A}\hat{B}}])\Gamma
_{AB^{\prime}}\otimes\Gamma_{BA^{\prime}}.
\end{equation}
Thus, an optimal solution is given by%
\begin{equation}
Z_{ABA^{\prime}B^{\prime}}=\Gamma_{AB^{\prime}}\otimes\Gamma_{BA^{\prime}%
}\left(  1-\operatorname{Tr}[\rho_{\hat{A}\hat{B}}K_{\hat{A}\hat{B}}]\right)
,
\end{equation}
for which the smallest $\mu$ possible is%
\begin{equation}
\mu=1-\operatorname{Tr}[\rho_{\hat{A}\hat{B}}K_{\hat{A}\hat{B}}],
\end{equation}
because%
\begin{equation}
Z_{AB}=\operatorname{Tr}_{A^{\prime}B^{\prime}}[Z_{ABA^{\prime}B^{\prime}%
}]=I_{AB}\left(  1-\operatorname{Tr}[\rho_{\hat{A}\hat{B}}K_{\hat{A}\hat{B}%
}]\right)  . \label{eq:final-steps-swap-simp-last}%
\end{equation}
We then conclude that%
\begin{equation}
e_{\operatorname{LOCC}}(\mathcal{S}_{AB}^{d},\rho_{\hat{A}\hat{B}}%
)=1-\sup_{K_{\hat{A}\hat{B}},L_{\hat{A}\hat{B}},N_{\hat{A}\hat{B}}\geq
0}\operatorname{Tr}[\rho_{\hat{A}\hat{B}}K_{\hat{A}\hat{B}}],
\end{equation}
subject to%
\begin{equation}
K_{\hat{A}\hat{B}}+L_{\hat{A}\hat{B}}+N_{\hat{A}\hat{B}}=I_{\hat{A}\hat{B}}%
\end{equation}
and the following channel is LOCC:%
\begin{multline}
\widetilde{\mathcal{L}}_{AB\hat{A}\hat{B}\rightarrow AB}(\omega_{AB\hat{A}%
\hat{B}})=\mathcal{S}_{AB}^{d}(\operatorname{Tr}_{\hat{A}\hat{B}}[K_{\hat
{A}\hat{B}}\omega_{AB\hat{A}\hat{B}}])+\frac{1}{2}\left(  \operatorname{id}%
_{A\rightarrow B}\otimes\mathcal{D}_{B\rightarrow A}+\mathcal{D}_{A\rightarrow
B}\otimes\operatorname{id}_{B\rightarrow A}\right)  (\operatorname{Tr}%
_{\hat{A}\hat{B}}[L_{\hat{A}\hat{B}}\tau_{\hat{A}\hat{B}}])\\
+\left(  \mathcal{D}_{A\rightarrow B}\otimes\mathcal{D}_{B\rightarrow
A}\right)  (\operatorname{Tr}_{\hat{A}\hat{B}}[N_{\hat{A}\hat{B}}%
\omega_{AB\hat{A}\hat{B}}]).
\end{multline}

\subsection{Evaluating channel infidelity}

\label{app:ch-infid-simplify-swap-LOCC}Let us start from
\eqref{eq:swap-symmetry-simplify-1}\ and recall the symmetries of the unitary
swap channel. This implies the following symmetry for the Choi operator in
\eqref{eq:Choi-swap-1}--\eqref{eq:Choi-swap-2}\ for the swap channel:%
\begin{align}
\Gamma_{ABA^{\prime}B^{\prime}}^{\mathcal{S}^{d}}  &  =\Gamma_{AB^{\prime}%
}\otimes\Gamma_{BA^{\prime}}\\
&  =(\overline{\mathcal{U}}_{A}\otimes\overline{\mathcal{V}}_{B}%
\otimes\mathcal{V}_{A^{\prime}}\otimes\mathcal{U}_{B^{\prime}})(\Gamma
_{ABA^{\prime}B^{\prime}}^{\mathcal{S}^{d}})
\end{align}
for all unitary channels $\mathcal{U}$ and $\mathcal{V}$. This implies that%
\begin{equation}
\Gamma_{ABA^{\prime}B^{\prime}}^{\mathcal{S}^{d}}=\int\int dU\ dV\ (\overline
{\mathcal{U}}_{A}\otimes\overline{\mathcal{V}}_{B}\otimes\mathcal{V}%
_{A^{\prime}}\otimes\mathcal{U}_{B^{\prime}})(\Gamma_{ABA^{\prime}B^{\prime}%
}^{\mathcal{S}^{d}}).
\end{equation}

By exploiting the semi-definite program in
\eqref{eq:SDP-ch-fid-1}--\eqref{eq:SDP-ch-fid-3}\ for channel infidelity, we
find that%
\begin{equation}
e_{\operatorname{LOCC}}^{F}(\mathcal{S}_{AB\rightarrow A^{\prime}B^{\prime}%
}^{d},\rho_{\hat{A}\hat{B}})=1-\left[  \sup_{\lambda\geq0,L_{AB\hat{A}\hat
{B}A^{\prime}B^{\prime}}\geq0,Q_{ABA^{\prime}B^{\prime}}}\lambda\right]  ^{2},
\label{eq:app-obj-func-infid}%
\end{equation}
subject to%
\begin{align}
\lambda I_{AB}  &  \leq\operatorname{Re}[\operatorname{Tr}_{A^{\prime
}B^{\prime}}[Q_{ABA^{\prime}B^{\prime}}]]\label{eq:lam-Q-constr-infid}\\
\operatorname{Tr}_{A^{\prime}B^{\prime}}[L_{AB\hat{A}\hat{B}A^{\prime
}B^{\prime}}]  &  =I_{AB\hat{A}\hat{B}}, \label{eq:TP-constraint-infid}%
\end{align}%
\begin{equation}%
\begin{bmatrix}
\Gamma_{ABA^{\prime}B^{\prime}}^{\mathcal{S}^{d}} & Q_{ABA^{\prime}B^{\prime}%
}^{\dag}\\
Q_{ABA^{\prime}B^{\prime}} & \operatorname{Tr}_{\hat{A}\hat{B}}[T_{\hat{A}%
\hat{B}}(\rho_{\hat{A}\hat{B}})L_{AB\hat{A}\hat{B}A^{\prime}B^{\prime}}]
\end{bmatrix}
\geq0,
\end{equation}
and $L_{AB\hat{A}\hat{B}A^{\prime}B^{\prime}}$ is the Choi operator for an
LOCC channel. Note that we can write the last constraint as%
\begin{equation}
|0\rangle\!\langle0|\otimes\Gamma_{ABA^{\prime}B^{\prime}}^{\mathcal{S}^{d}%
}+|0\rangle\!\langle1|\otimes Q_{ABA^{\prime}B^{\prime}}^{\dag}+|1\rangle
\langle0|\otimes Q_{ABA^{\prime}B^{\prime}}+|1\rangle\!\langle1|\otimes
\operatorname{Tr}_{\hat{A}\hat{B}}[T_{\hat{A}\hat{B}}(\rho_{\hat{A}\hat{B}%
})P_{AB\hat{A}\hat{B}A^{\prime}B^{\prime}}]\geq0.
\label{eq:infid-SDP-constraint-final}%
\end{equation}
Suppose that $\lambda$, $L_{AB\hat{A}\hat{B}A^{\prime}B^{\prime}}$, and
$Q_{ABA^{\prime}B^{\prime}}$ is an optimal solution. Let%
\begin{equation}
\mathcal{W}_{ABA^{\prime}B^{\prime}}\coloneqq\overline{\mathcal{U}}_{A}%
\otimes\overline{\mathcal{V}}_{B}\otimes\mathcal{V}_{A^{\prime}}%
\otimes\mathcal{U}_{B^{\prime}}.
\end{equation}
Then it follows that $\lambda$, $\mathcal{W}_{ABA^{\prime}B^{\prime}%
}(L_{AB\hat{A}\hat{B}A^{\prime}B^{\prime}})$, and $\mathcal{W}_{ABA^{\prime
}B^{\prime}}(Q_{ABA^{\prime}B^{\prime}})$ is an optimal solution. This is
because all of the constraints are satisfied for these choices while still
achieving the same optimal value. Indeed, consider that%
\begin{align}
\lambda I_{AB}  &  \leq\operatorname{Re}[\operatorname{Tr}_{A^{\prime
}B^{\prime}}[Q_{ABA^{\prime}B^{\prime}}]]\\
\Longleftrightarrow\qquad\lambda(\overline{\mathcal{U}}_{A}\otimes
\overline{\mathcal{V}}_{B})(I_{AB})  &  \leq(\overline{\mathcal{U}}_{A}%
\otimes\overline{\mathcal{V}}_{B})(\operatorname{Re}[\operatorname{Tr}%
_{A^{\prime}B^{\prime}}[Q_{ABA^{\prime}B^{\prime}}]])\\
\Longleftrightarrow\qquad\lambda I_{AB}  &  \leq\operatorname{Re}%
[\operatorname{Tr}_{A^{\prime}B^{\prime}}[(\overline{\mathcal{U}}_{A}%
\otimes\overline{\mathcal{V}}_{B})(Q_{ABA^{\prime}B^{\prime}})]]\\
\Longleftrightarrow\qquad\lambda I_{AB}  &  \leq\operatorname{Re}%
[\operatorname{Tr}_{A^{\prime}B^{\prime}}[\mathcal{W}_{ABA^{\prime}B^{\prime}%
}(Q_{ABA^{\prime}B^{\prime}})]],
\end{align}%
\begin{equation}
\operatorname{Tr}_{A^{\prime}B^{\prime}}[L_{AB\hat{A}\hat{B}A^{\prime
}B^{\prime}}]=I_{AB\hat{A}\hat{B}}\qquad\Longleftrightarrow\qquad
\operatorname{Tr}_{A^{\prime}B^{\prime}}[\mathcal{W}_{ABA^{\prime}B^{\prime}%
}(L_{AB\hat{A}\hat{B}A^{\prime}B^{\prime}})]=I_{AB\hat{A}\hat{B}},
\end{equation}
and that%
\begin{align}
|0\rangle\!\langle0|\otimes\Gamma_{ABA^{\prime}B^{\prime}}^{\mathcal{S}^{d}%
}+|0\rangle\!\langle1|\otimes Q_{ABA^{\prime}B^{\prime}}^{\dag}+|1\rangle
\langle0|\otimes Q_{ABA^{\prime}B^{\prime}}+|1\rangle\!\langle1|\otimes
\operatorname{Tr}_{\hat{A}\hat{B}}[T_{\hat{A}\hat{B}}(\rho_{\hat{A}\hat{B}%
})L_{AB\hat{A}\hat{B}A^{\prime}B^{\prime}}]  &  \geq0\\
\Longleftrightarrow\qquad(\operatorname{id}\otimes\mathcal{W}_{ABA^{\prime
}B^{\prime}})\left(  |0\rangle\!\langle0|\otimes\Gamma^{\mathcal{N}}%
+|0\rangle\!\langle1|\otimes Q^{\dag}+|1\rangle\!\langle0|\otimes
Q+|1\rangle\langle1|\otimes\operatorname{Tr}_{\hat{A}\hat{B}}[T_{\hat{A}%
\hat{B}}(\rho_{\hat{A}\hat{B}})L]\right)   &  \geq0\\
\Longleftrightarrow\qquad|0\rangle\!\langle0|\otimes\mathcal{W}(\Gamma
^{\mathcal{S}^{d}})+|0\rangle\!\langle1|\otimes\mathcal{W}(Q^{\dag}%
)+|1\rangle\!\langle0|\otimes\mathcal{W}(Q)+|1\rangle\!\langle1|\otimes
\mathcal{W}(\operatorname{Tr}_{\hat{A}\hat{B}}[T_{\hat{A}\hat{B}}(\rho
_{\hat{A}\hat{B}})L])  &  \geq0\\
\Longleftrightarrow\qquad|0\rangle\!\langle0|\otimes\Gamma^{\mathcal{S}^{d}%
}+|0\rangle\!\langle1|\otimes\lbrack\mathcal{W}(Q)]^{\dag}+|1\rangle
\langle0|\otimes\mathcal{W}(Q)+|1\rangle\!\langle1|\otimes\operatorname{Tr}%
_{\hat{A}\hat{B}}[T_{\hat{A}\hat{B}}(\rho_{\hat{A}\hat{B}})\mathcal{W}(L)]  &
\geq0.
\end{align}
Also, $\mathcal{W}_{ABA^{\prime}B^{\prime}}(L_{AB\hat{A}\hat{B}A^{\prime
}B^{\prime}})$ is the Choi operator for an LOCC channel if $L_{AB\hat{A}%
\hat{B}A^{\prime}B^{\prime}}$ is. Furthermore, due to the fact that the
objective function is linear and the constraints are linear operator
inequalities, it follows that convex combinations of solutions are solutions
as well. So this implies that if $\lambda$, $L_{AB\hat{A}\hat{B}A^{\prime
}B^{\prime}}$, and $Q_{ABA^{\prime}B^{\prime}}$ is an optimal solution, then
so is $\lambda$,%
\begin{align}
\widetilde{L}_{AB\hat{A}\hat{B}A^{\prime}B^{\prime}}  &  \coloneqq\int\int
dU\ dV\ (\overline{\mathcal{U}}_{A}\otimes\overline{\mathcal{V}}_{B}%
\otimes\mathcal{V}_{A^{\prime}}\otimes\mathcal{U}_{B^{\prime}})(L_{AB\hat
{A}\hat{B}A^{\prime}B^{\prime}}),\\
\widetilde{Q}_{ABA^{\prime}B^{\prime}}  &  \coloneqq\int\int
dU\ dV\ (\overline{\mathcal{U}}_{A}\otimes\overline{\mathcal{V}}_{B}%
\otimes\mathcal{V}_{A^{\prime}}\otimes\mathcal{U}_{B^{\prime}})Q_{ABA^{\prime
}B^{\prime}}.
\end{align}
Furthermore, $\widetilde{L}_{AB\hat{A}\hat{B}A^{\prime}B^{\prime}}$ is the
Choi operator for an LOCC channel if $L_{AB\hat{A}\hat{B}A^{\prime}B^{\prime}%
}$ is. As argued in Appendix~\ref{app:SDP-simplify-swap}, $\widetilde
{L}_{AB\hat{A}\hat{B}A^{\prime}B^{\prime}}$ has a simpler form as follows:%
\begin{multline}
\widetilde{L}_{AB\hat{A}\hat{B}A^{\prime}B^{\prime}}=\Gamma_{AB^{\prime}%
}\otimes\Gamma_{BA^{\prime}}\otimes K_{\hat{A}\hat{B}}+\Gamma_{AB^{\prime}%
}\otimes\frac{dI_{BA^{\prime}}-\Gamma_{BA^{\prime}}}{d^{2}-1}\otimes
L_{\hat{A}\hat{B}}\\
+\frac{dI_{AB^{\prime}}-\Gamma_{AB^{\prime}}}{d^{2}-1}\otimes\Gamma
_{BA^{\prime}}\otimes M_{\hat{A}\hat{B}}+\frac{dI_{AB^{\prime}}-\Gamma
_{AB^{\prime}}}{d^{2}-1}\otimes\frac{dI_{BA^{\prime}}-\Gamma_{BA^{\prime}}%
}{d^{2}-1}\otimes N_{\hat{A}\hat{B}},
\end{multline}
and the constraints in
\eqref{eq:app-obj-func-infid}--\eqref{eq:TP-constraint-infid} on
$\widetilde{L}_{AB\hat{A}\hat{B}A^{\prime}B^{\prime}}$ simplify to%
\begin{align}
K_{\hat{A}\hat{B}},L_{\hat{A}\hat{B}},M_{\hat{A}\hat{B}},N_{\hat{A}\hat{B}}
&  \geq0,\\
K_{\hat{A}\hat{B}}+L_{\hat{A}\hat{B}}+M_{\hat{A}\hat{B}}+N_{\hat{A}\hat{B}}
&  =I_{\hat{A}\hat{B}}.
\end{align}
We then find that $\widetilde{Q}_{ABA^{\prime}B^{\prime}}$ simplifies to%
\begin{multline}
\widetilde{Q}_{ABA^{\prime}B^{\prime}}=q_{1}\Gamma_{AB^{\prime}}\otimes
\Gamma_{BA^{\prime}}+q_{2}\Gamma_{AB^{\prime}}\otimes\frac{dI_{BA^{\prime}%
}-\Gamma_{BA^{\prime}}}{d^{2}-1}\\
+q_{3}\frac{dI_{AB^{\prime}}-\Gamma_{AB^{\prime}}}{d^{2}-1}\otimes
\Gamma_{BA^{\prime}}+q_{4}\frac{dI_{AB^{\prime}}-\Gamma_{AB^{\prime}}}%
{d^{2}-1}\otimes\frac{dI_{BA^{\prime}}-\Gamma_{BA^{\prime}}}{d^{2}-1},
\end{multline}
where $q_{1},q_{2},q_{3},q_{4}\in\mathbb{C}$. The constraint in
\eqref{eq:lam-Q-constr-infid} reduces to the following:%
\begin{equation}
\lambda\leq\operatorname{Re}[q_{1}+q_{2}+q_{3}+q_{4}],
\end{equation}
because%
\begin{equation}
\operatorname{Tr}_{A^{\prime}B^{\prime}}[\widetilde{Q}_{ABA^{\prime}B^{\prime
}}]=(q_{1}+q_{2}+q_{3}+q_{4})I_{AB}.
\end{equation}
Furthermore, the constraint in \eqref{eq:infid-SDP-constraint-final} then
reduces to the following:%
\begin{multline}
|0\rangle\!\langle0|\otimes\Gamma_{AB^{\prime}}\otimes\Gamma_{BA^{\prime}%
}+|0\rangle\!\langle1|\otimes\left[
\begin{array}
[c]{c}%
q_{1}\Gamma_{AB^{\prime}}\otimes\Gamma_{BA^{\prime}}+q_{2}\Gamma_{AB^{\prime}%
}\otimes\frac{dI_{BA^{\prime}}-\Gamma_{BA^{\prime}}}{d^{2}-1}\\
+q_{3}\frac{dI_{AB^{\prime}}-\Gamma_{AB^{\prime}}}{d^{2}-1}\otimes
\Gamma_{BA^{\prime}}+q_{4}\frac{dI_{AB^{\prime}}-\Gamma_{AB^{\prime}}}%
{d^{2}-1}\otimes\frac{dI_{BA^{\prime}}-\Gamma_{BA^{\prime}}}{d^{2}-1}%
\end{array}
\right]  ^{\dag}\\
+|1\rangle\!\langle0|\otimes\left[
\begin{array}
[c]{c}%
q_{1}\Gamma_{AB^{\prime}}\otimes\Gamma_{BA^{\prime}}+q_{2}\Gamma_{AB^{\prime}%
}\otimes\frac{dI_{BA^{\prime}}-\Gamma_{BA^{\prime}}}{d^{2}-1}\\
+q_{3}\frac{dI_{AB^{\prime}}-\Gamma_{AB^{\prime}}}{d^{2}-1}\otimes
\Gamma_{BA^{\prime}}+q_{4}\frac{dI_{AB^{\prime}}-\Gamma_{AB^{\prime}}}%
{d^{2}-1}\otimes\frac{dI_{BA^{\prime}}-\Gamma_{BA^{\prime}}}{d^{2}-1}%
\end{array}
\right] \\
+|1\rangle\!\langle1|\otimes\left[
\begin{array}
[c]{c}%
\Gamma_{AB^{\prime}}\otimes\Gamma_{BA^{\prime}}\otimes\operatorname{Tr}%
[T_{\hat{A}\hat{B}}(\rho_{\hat{A}\hat{B}})K_{\hat{A}\hat{B}}]+\Gamma
_{AB^{\prime}}\otimes\frac{dI_{BA^{\prime}}-\Gamma_{BA^{\prime}}}{d^{2}%
-1}\otimes\operatorname{Tr}[T_{\hat{A}\hat{B}}(\rho_{\hat{A}\hat{B}}%
)L_{\hat{A}\hat{B}}]\\
\frac{dI_{AB^{\prime}}-\Gamma_{AB^{\prime}}}{d^{2}-1}\otimes\Gamma
_{BA^{\prime}}\otimes\operatorname{Tr}[T_{\hat{A}\hat{B}}(\rho_{\hat{A}\hat
{B}})M_{\hat{A}\hat{B}}]+\frac{dI_{AB^{\prime}}-\Gamma_{AB^{\prime}}}{d^{2}%
-1}\otimes\frac{dI_{BA^{\prime}}-\Gamma_{BA^{\prime}}}{d^{2}-1}\otimes
\operatorname{Tr}[T_{\hat{A}\hat{B}}(\rho_{\hat{A}\hat{B}})N_{\hat{A}\hat{B}}]
\end{array}
\right]  \geq0.
\end{multline}
Now exploiting the orthogonality of the operators $\Gamma_{AB^{\prime}}%
\otimes\Gamma_{BA^{\prime}}$, $\Gamma_{AB^{\prime}}\otimes\frac{dI_{BA^{\prime
}}-\Gamma_{BA^{\prime}}}{d^{2}-1}$, $\frac{dI_{AB^{\prime}}-\Gamma
_{AB^{\prime}}}{d^{2}-1}\otimes\Gamma_{BA^{\prime}}$, and $\frac
{dI_{AB^{\prime}}-\Gamma_{AB^{\prime}}}{d^{2}-1}\otimes\frac{dI_{BA^{\prime}%
}-\Gamma_{BA^{\prime}}}{d^{2}-1}$, we conclude that the single constraint
above is equivalent to the following four constraints:%
\begin{align}
|0\rangle\!\langle0|+q_{1}^{\ast}|0\rangle\!\langle1|+q_{1}|1\rangle
\langle0|+\operatorname{Tr}[T_{\hat{A}\hat{B}}(\rho_{\hat{A}\hat{B}}%
)K_{\hat{A}\hat{B}}]|1\rangle\!\langle1|  &  \geq0,\\
q_{2}^{\ast}|0\rangle\!\langle1|+q_{2}|1\rangle\!\langle0|+\operatorname{Tr}%
[T_{\hat{A}\hat{B}}(\rho_{\hat{A}\hat{B}})L_{\hat{A}\hat{B}}]|1\rangle
\langle1|  &  \geq0,\\
q_{3}^{\ast}|0\rangle\!\langle1|+q_{3}|1\rangle\!\langle0|+\operatorname{Tr}%
[T_{\hat{A}\hat{B}}(\rho_{\hat{A}\hat{B}})M_{\hat{A}\hat{B}}]|1\rangle
\langle1|  &  \geq0,\\
q_{4}^{\ast}|0\rangle\!\langle1|+q_{4}|1\rangle\!\langle0|+\operatorname{Tr}%
[T_{\hat{A}\hat{B}}(\rho_{\hat{A}\hat{B}})N_{\hat{A}\hat{B}}]|1\rangle
\langle1|  &  \geq0.
\end{align}
These constraints can in turn be expressed as the following four matrix
inequalities:%
\begin{align}%
\begin{bmatrix}
1 & q_{1}^{\ast}\\
q_{1} & \operatorname{Tr}[T_{\hat{A}\hat{B}}(\rho_{\hat{A}\hat{B}})K_{\hat
{A}\hat{B}}]
\end{bmatrix}
&  \geq0,\\%
\begin{bmatrix}
0 & q_{2}^{\ast}\\
q_{2} & \operatorname{Tr}[T_{\hat{A}\hat{B}}(\rho_{\hat{A}\hat{B}})L_{\hat
{A}\hat{B}}]
\end{bmatrix}
&  \geq0,\\%
\begin{bmatrix}
0 & q_{3}^{\ast}\\
q_{3} & \operatorname{Tr}[T_{\hat{A}\hat{B}}(\rho_{\hat{A}\hat{B}})M_{\hat
{A}\hat{B}}]
\end{bmatrix}
&  \geq0,\\%
\begin{bmatrix}
0 & q_{4}^{\ast}\\
q_{4} & \operatorname{Tr}[T_{\hat{A}\hat{B}}(\rho_{\hat{A}\hat{B}})N_{\hat
{A}\hat{B}}]
\end{bmatrix}
&  \geq0.
\end{align}
Since $\operatorname{Tr}[T_{\hat{A}\hat{B}}(\rho_{\hat{A}\hat{B}})K_{\hat
{A}\hat{B}}],\operatorname{Tr}[T_{\hat{A}\hat{B}}(\rho_{\hat{A}\hat{B}%
})L_{\hat{A}\hat{B}}],\operatorname{Tr}[T_{\hat{A}\hat{B}}(\rho_{\hat{A}%
\hat{B}})M_{\hat{A}\hat{B}}],\operatorname{Tr}[T_{\hat{A}\hat{B}}(\rho
_{\hat{A}\hat{B}})N_{\hat{A}\hat{B}}]\geq0$, the inequalities above hold if
and only if%
\begin{equation}
\operatorname{Tr}[T_{\hat{A}\hat{B}}(\rho_{\hat{A}\hat{B}})K_{\hat{A}\hat{B}%
}]\geq\left\vert q_{1}\right\vert ^{2},\quad q_{2}=q_{3}=q_{4}=0.
\end{equation}
Note that we can perform the following substitions as we did
previously:$\ K_{\hat{A}\hat{B}}\rightarrow T_{\hat{A}\hat{B}}(K_{\hat{A}%
\hat{B}})$, $L_{\hat{A}\hat{B}}\rightarrow T_{\hat{A}\hat{B}}(L_{\hat{A}%
\hat{B}})$, $M_{\hat{A}\hat{B}}\rightarrow T_{\hat{A}\hat{B}}(M_{\hat{A}%
\hat{B}})$, and $N_{\hat{A}\hat{B}}\rightarrow T_{\hat{A}\hat{B}}(N_{\hat
{A}\hat{B}})$. The objective function does not change under these
substitutions. Thus, the optimization problem in \eqref{eq:app-obj-func-infid}
reduces to the following:%
\begin{equation}
1-\left[  \sup_{\lambda\geq0,K_{\hat{A}\hat{B}},L_{\hat{A}\hat{B}},M_{\hat
{A}\hat{B}},N_{\hat{A}\hat{B}}\geq0,q_{1}\in\mathbb{C}}\lambda\right]  ^{2},
\end{equation}
subject to%
\begin{align}
\operatorname{Tr}[\rho_{\hat{A}\hat{B}}K_{\hat{A}\hat{B}}]  &  \geq\left\vert
q_{1}\right\vert ^{2},\\
\lambda &  \leq\operatorname{Re}[q_{1}],
\end{align}%
\begin{equation}
K_{\hat{A}\hat{B}}+L_{\hat{A}\hat{B}}+M_{\hat{A}\hat{B}}+N_{\hat{A}\hat{B}%
}=I_{\hat{A}\hat{B}},
\end{equation}
and the following channel is LOCC:%
\begin{multline}
\widetilde{\mathcal{L}}_{AB\hat{A}\hat{B}\rightarrow AB}(\omega_{AB\hat{A}%
\hat{B}})=\mathcal{S}_{AB}^{d}(\operatorname{Tr}_{\hat{A}\hat{B}}[K_{\hat
{A}\hat{B}}\omega_{AB\hat{A}\hat{B}}])+\left(  \operatorname{id}_{A\rightarrow
B}\otimes\mathcal{D}_{B\rightarrow A}\right)  (\operatorname{Tr}_{\hat{A}%
\hat{B}}[L_{\hat{A}\hat{B}}\tau_{\hat{A}\hat{B}}])\\
+\left(  \mathcal{D}_{A\rightarrow B}\otimes\operatorname{id}_{B\rightarrow
A}\right)  (\operatorname{Tr}_{\hat{A}\hat{B}}[M_{\hat{A}\hat{B}}%
\omega_{AB\hat{A}\hat{B}}])+\left(  \mathcal{D}_{A\rightarrow B}%
\otimes\mathcal{D}_{B\rightarrow A}\right)  (\operatorname{Tr}_{\hat{A}\hat
{B}}[N_{\hat{A}\hat{B}}\omega_{AB\hat{A}\hat{B}}]).
\end{multline}
Since we are trying to maximize the value of $\lambda$ subject to these
constraints, it is clear that we should pick $\lambda=q_{1}=\sqrt
{\operatorname{Tr}[\rho_{\hat{A}\hat{B}}K_{\hat{A}\hat{B}}]}$. We can
furthermore apply the other symmetry in
\eqref{eq:swap-ch-add-symmetry-for-LOCC} to get a similar reduced form for the
optimization problem. This concludes the proof.

\section{ Proof of Eq.~\eqref{eq:SDP-dual}}

\label{sec:SDP-dual-gen-bi}

Let $A$ and $B$ be Hermitian operators, and let $\Phi$ be a
Hermiticity-preserving map. Recall from \cite{Watrous2018}\ that if a primal
semi-definite program (SDP)\ is given by%
\begin{equation}
\inf_{Y\geq0}\left\{  \operatorname{Tr}[BY]:\Phi(Y)\geq A\right\}  ,
\label{eq:SDP-standard-form-primal}%
\end{equation}
then its dual is given by%
\begin{equation}
\sup_{X\geq0}\left\{  \operatorname{Tr}[AX]:\Phi^{\dag}(X)\leq B\right\}  ,
\label{eq:dual-sdp-generic}%
\end{equation}
where $\Phi^{\dag}$ is the Hilbert--Schmidt adjoint of $\Phi$. For our
semi-definite program of interest in \eqref{eq:SDP-sim-error}, we find that
\eqref{eq:SDP-sim-error} can be written in the standard form in
\eqref{eq:SDP-standard-form-primal} with%
\begin{align}
Y  &  =\text{diag}(\mu,Z_{ABA^{\prime}B^{\prime}},P_{AB\hat{A}\hat{B}%
A^{\prime}B^{\prime}}),\\
B  &  =\text{diag}(1,0,0),\\
A  &  =\text{diag}(0,\Gamma_{ABA^{\prime}B^{\prime}}^{\mathcal{N}}%
,0,I_{AB\hat{A}\hat{B}},-I_{AB\hat{A}\hat{B}}),\\
\Phi(Y)  &  =\text{diag}(\mu I_{AB}-Z_{AB},\nonumber\\
&  \qquad\qquad Z_{ABA^{\prime}B^{\prime}}+\operatorname{Tr}_{\hat{A}\hat{B}%
}[T_{\hat{A}\hat{B}}(\rho_{\hat{A}\hat{B}}P_{AB\hat{A}\hat{B}A^{\prime
}B^{\prime}})],\nonumber\\
&  \qquad\qquad T_{B\hat{B}B^{\prime}}(P_{AB\hat{A}\hat{B}A^{\prime}B^{\prime
}}),\operatorname{Tr}_{A^{\prime}B^{\prime}}[P_{AB\hat{A}\hat{B}A^{\prime
}B^{\prime}}],\nonumber\\
&  \qquad\qquad-\operatorname{Tr}_{A^{\prime}B^{\prime}}[P_{AB\hat{A}\hat
{B}A^{\prime}B^{\prime}}]).
\end{align}
So we need to compute the Hilbert--Schmidt adjoint of $\Phi$, which is defined
for a Hermiticity-preserving map by the equation%
\begin{equation}
\operatorname{Tr}[X\Phi(Y)]=\operatorname{Tr}[\Phi^{\dag}(X)Y],
\end{equation}
holding for all Hermitian $X$ and $Y$. By defining
\begin{equation}
X=\text{diag}(X_{AB}^{1},X_{ABA^{\prime}B^{\prime}}^{2},X_{AB\hat{A}\hat
{B}A^{\prime}B^{\prime}}^{3},X_{AB\hat{A}\hat{B}}^{4},X_{AB\hat{A}\hat{B}}%
^{5}),
\end{equation}
consider that%
\begin{align}
\operatorname{Tr}[X\Phi(Y)]  &  =\operatorname{Tr}[X_{AB}^{1}(\mu
I_{AB}-Z_{AB})]\nonumber\\
&  \qquad\qquad+\operatorname{Tr}[X_{ABA^{\prime}B^{\prime}}^{2}%
(Z_{ABA^{\prime}B^{\prime}}+\operatorname{Tr}_{\hat{A}\hat{B}}[T_{\hat{A}%
\hat{B}}(\rho_{\hat{A}\hat{B}}P_{AB\hat{A}\hat{B}A^{\prime}B^{\prime}%
})])]\nonumber\\
&  \qquad\qquad+\operatorname{Tr}[X_{AB\hat{A}\hat{B}A^{\prime}B^{\prime}}%
^{3}T_{B\hat{B}B^{\prime}}(P_{AB\hat{A}\hat{B}A^{\prime}B^{\prime}%
})]\nonumber\\
&  \qquad\qquad+\operatorname{Tr}[X_{AB\hat{A}\hat{B}}^{4}\operatorname{Tr}%
_{A^{\prime}B^{\prime}}[P_{AB\hat{A}\hat{B}A^{\prime}B^{\prime}}]]\nonumber\\
&  \qquad\qquad-\operatorname{Tr}[X_{AB\hat{A}\hat{B}}^{5}\operatorname{Tr}%
_{A^{\prime}B^{\prime}}[P_{AB\hat{A}\hat{B}A^{\prime}B^{\prime}}]]\\
&  =\mu\operatorname{Tr}[X_{AB}^{1}]+\operatorname{Tr}[(-X_{AB}^{1}\otimes
I_{A^{\prime}B^{\prime}}+X_{ABA^{\prime}B^{\prime}}^{2})Z_{ABA^{\prime
}B^{\prime}}]\nonumber\\
&  \qquad\qquad+\operatorname{Tr}[X_{AB\hat{A}\hat{B}A^{\prime}B^{\prime}%
}^{\prime}P_{AB\hat{A}\hat{B}A^{\prime}B^{\prime}}],
\end{align}
where%
\begin{equation}
X_{AB\hat{A}\hat{B}A^{\prime}B^{\prime}}^{\prime}\coloneqq X_{ABA^{\prime
}B^{\prime}}^{2}\otimes T_{\hat{A}\hat{B}}(\rho_{\hat{A}\hat{B}})+T_{B\hat
{B}B^{\prime}}(X_{AB\hat{A}\hat{B}A^{\prime}B^{\prime}}^{3})+(X_{AB\hat{A}%
\hat{B}}^{4}-X_{AB\hat{A}\hat{B}}^{5})\otimes I_{A^{\prime}B^{\prime}}.
\end{equation}
So this means that%
\begin{multline}
\Phi^{\dag}(X)=\text{diag}(\operatorname{Tr}[X_{AB}^{1}],-X_{AB}^{1}\otimes
I_{A^{\prime}B^{\prime}}+X_{ABA^{\prime}B^{\prime}}^{2},\\
X_{ABA^{\prime}B^{\prime}}^{2}\otimes T_{\hat{A}\hat{B}}(\rho_{\hat{A}\hat{B}%
})+T_{B\hat{B}B^{\prime}}(X_{AB\hat{A}\hat{B}A^{\prime}B^{\prime}}^{3})\\
+(X_{AB\hat{A}\hat{B}}^{4}-X_{AB\hat{A}\hat{B}}^{5})\otimes I_{A^{\prime
}B^{\prime}}).
\end{multline}
Then the dual SDP\ is found by plugging into \eqref{eq:dual-sdp-generic}:%
\begin{equation}
\sup_{X^{1},\ldots,X^{5}\geq0}\operatorname{Tr}[\Gamma_{ABA^{\prime}B^{\prime
}}^{\mathcal{N}}X_{ABA^{\prime}B^{\prime}}^{2}]+\operatorname{Tr}[X_{AB\hat
{A}\hat{B}}^{4}-X_{AB\hat{A}\hat{B}}^{5}],
\end{equation}
subject to%
\begin{align}
\operatorname{Tr}[X_{AB}^{1}]  &  \leq1,\\
X_{ABA^{\prime}B^{\prime}}^{2}  &  \leq X_{AB}^{1}\otimes I_{A^{\prime
}B^{\prime}},\\
X_{ABA^{\prime}B^{\prime}}^{2}\otimes T_{\hat{A}\hat{B}}(\rho_{\hat{A}\hat{B}%
})+T_{B\hat{B}B^{\prime}}(X_{AB\hat{A}\hat{B}A^{\prime}B^{\prime}}^{3})  &
\leq-(X_{AB\hat{A}\hat{B}}^{4}-X_{AB\hat{A}\hat{B}}^{5})\otimes I_{A^{\prime
}B^{\prime}}.
\end{align}
This can then be rewritten as%
\begin{equation}
\sup_{\substack{X^{1},X^{2},X^{3}\geq0,\\W\in\text{Herm}}}\operatorname{Tr}%
[\Gamma_{ABA^{\prime}B^{\prime}}^{\mathcal{N}}X_{ABA^{\prime}B^{\prime}}%
^{2}]-\operatorname{Tr}[W_{AB\hat{A}\hat{B}}],
\end{equation}
subject to%
\begin{align}
\operatorname{Tr}[X_{AB}^{1}]  &  \leq1,\\
X_{ABA^{\prime}B^{\prime}}^{2}  &  \leq X_{AB}^{1}\otimes I_{A^{\prime
}B^{\prime}},\\
X_{ABA^{\prime}B^{\prime}}^{2}\otimes T_{\hat{A}\hat{B}}(\rho_{\hat{A}\hat{B}%
})+T_{B\hat{B}B^{\prime}}(X_{AB\hat{A}\hat{B}A^{\prime}B^{\prime}}^{3})  &
\leq W_{AB\hat{A}\hat{B}}\otimes I_{A^{\prime}B^{\prime}},
\end{align}
concluding the proof.

\section{Proof of Proposition~\ref{prop:swap-sdp-simplify}%
\label{app:SDP-simplify-swap}}

By following precisely the same reasoning given in
Appendix~\ref{app:simplify-LOCC-err-swap}, we conclude that the desired
optimization can be written as%
\begin{equation}
e_{\text{PPT}}(\mathcal{S}_{AB}^{d},\rho_{\hat{A}\hat{B}})=1-\sup_{K_{\hat
{A}\hat{B}},L_{\hat{A}\hat{B}},N_{\hat{A}\hat{B}}\geq0}\operatorname{Tr}%
[\rho_{\hat{A}\hat{B}}K_{\hat{A}\hat{B}}], \label{eq:opt-C-PPT-P-SDP-start-1}%
\end{equation}
subject to%
\begin{equation}
K_{\hat{A}\hat{B}}+L_{\hat{A}\hat{B}}+N_{\hat{A}\hat{B}}=I_{\hat{A}\hat{B}}%
\end{equation}
and the following channel is C-PPT-P:%
\begin{multline}
\widetilde{\mathcal{P}}_{AB\hat{A}\hat{B}\rightarrow AB}(\omega_{AB\hat{A}%
\hat{B}})=\mathcal{S}_{AB}^{d}(\operatorname{Tr}_{\hat{A}\hat{B}}[K_{\hat
{A}\hat{B}}\omega_{AB\hat{A}\hat{B}}])+\frac{1}{2}\left(  \operatorname{id}%
_{A\rightarrow B}\otimes\mathcal{D}_{B\rightarrow A}+\mathcal{D}_{A\rightarrow
B}\otimes\operatorname{id}_{B\rightarrow A}\right)  (\operatorname{Tr}%
_{\hat{A}\hat{B}}[L_{\hat{A}\hat{B}}\tau_{\hat{A}\hat{B}}%
])\label{eq:opt-C-PPT-P-SDP-start-3}\\
+\left(  \mathcal{D}_{A\rightarrow B}\otimes\mathcal{D}_{B\rightarrow
A}\right)  (\operatorname{Tr}_{\hat{A}\hat{B}}[N_{\hat{A}\hat{B}}%
\omega_{AB\hat{A}\hat{B}}]).
\end{multline}
Recall that the Choi operator of $\widetilde{\mathcal{P}}_{AB\hat{A}\hat
{B}\rightarrow AB}$ is given by%
\begin{multline}
\widetilde{P}_{AB\hat{A}\hat{B}A^{\prime}B^{\prime}}=\Gamma_{AB^{\prime}%
}\otimes\Gamma_{BA^{\prime}}\otimes T_{\hat{A}\hat{B}}(K_{\hat{A}\hat{B}%
})+\frac{1}{2}\left(  \Gamma_{AB^{\prime}}\otimes\frac{dI_{BA^{\prime}}%
-\Gamma_{BA^{\prime}}}{d^{2}-1}+\frac{dI_{AB^{\prime}}-\Gamma_{AB^{\prime}}%
}{d^{2}-1}\otimes\Gamma_{BA^{\prime}}\right)  \otimes T_{\hat{A}\hat{B}%
}(L_{\hat{A}\hat{B}})\\
+\frac{dI_{AB^{\prime}}-\Gamma_{AB^{\prime}}}{d^{2}-1}\otimes\frac
{dI_{BA^{\prime}}-\Gamma_{BA^{\prime}}}{d^{2}-1}\otimes T_{\hat{A}\hat{B}%
}(N_{\hat{A}\hat{B}}).
\end{multline}
The conditions $K_{\hat{A}\hat{B}},L_{\hat{A}\hat{B}},N_{\hat{A}\hat{B}}\geq0$
and $K_{\hat{A}\hat{B}}+L_{\hat{A}\hat{B}}+N_{\hat{A}\hat{B}}=I_{\hat{A}%
\hat{B}}$ guarantee that $\widetilde{P}_{AB\hat{A}\hat{B}A^{\prime}B^{\prime}%
}$ is the Choi operator of a channel. In order for $\widetilde{P}_{AB\hat
{A}\hat{B}A^{\prime}B^{\prime}}$ to be the Choi operator of a C-PPT-P channel,
the following condition should hold%
\begin{equation}
T_{B\hat{B}B^{\prime}}(\widetilde{P}_{AB\hat{A}\hat{B}A^{\prime}B^{\prime}%
})\geq0. \label{eq:ppt-psd-constr}%
\end{equation}
In what follows, we explore what this condition imposes on the operators
$K_{\hat{A}\hat{B}},L_{\hat{A}\hat{B}},N_{\hat{A}\hat{B}}$. Consider that%
\begin{align}
&  T_{B\hat{B}B^{\prime}}(\widetilde{P}_{AB\hat{A}\hat{B}A^{\prime}B^{\prime}%
})\nonumber\\
&  =T_{B^{\prime}}(\Gamma_{AB^{\prime}})\otimes T_{B}(\Gamma_{BA^{\prime}%
})\otimes T_{\hat{A}}(K_{\hat{A}\hat{B}})\nonumber\\
&  \qquad+\frac{1}{2}\left(  T_{B^{\prime}}(\Gamma_{AB^{\prime}})\otimes
T_{B}\left(  \frac{dI_{BA^{\prime}}-\Gamma_{BA^{\prime}}}{d^{2}-1}\right)
+T_{B^{\prime}}\left(  \frac{dI_{AB^{\prime}}-\Gamma_{AB^{\prime}}}{d^{2}%
-1}\right)  \otimes T_{B}(\Gamma_{BA^{\prime}})\right)  \otimes T_{\hat{A}%
}(L_{\hat{A}\hat{B}})\nonumber\\
&  \qquad+T_{B^{\prime}}\left(  \frac{dI_{AB^{\prime}}-\Gamma_{AB^{\prime}}%
}{d^{2}-1}\right)  \otimes T_{B}\left(  \frac{dI_{BA^{\prime}}-\Gamma
_{BA^{\prime}}}{d^{2}-1}\right)  \otimes T_{\hat{A}}(N_{\hat{A}\hat{B}%
}).\label{eq:partial-transpose-cond-1-swap-simp}\\
&  =F_{AB^{\prime}}\otimes F_{BA^{\prime}}\otimes T_{\hat{A}}(K_{\hat{A}%
\hat{B}})\nonumber\\
&  \qquad+\frac{1}{2}\left(  F_{AB^{\prime}}\otimes\frac{dI_{BA^{\prime}%
}-F_{BA^{\prime}}}{d^{2}-1}+\frac{dI_{AB^{\prime}}-F_{AB^{\prime}}}{d^{2}%
-1}\otimes F_{BA^{\prime}}\right)  \otimes T_{\hat{A}}(L_{\hat{A}\hat{B}%
})\nonumber\\
&  \qquad+\frac{dI_{AB^{\prime}}-F_{AB^{\prime}}}{d^{2}-1}\otimes
\frac{dI_{BA^{\prime}}-F_{BA^{\prime}}}{d^{2}-1}\otimes T_{\hat{A}}(N_{\hat
{A}\hat{B}}). \label{eq:partial-transpose-cond-1-swap-simp-2}%
\end{align}
Now consider that the SWAP\ operator $F_{CD}$ and the identity $I_{CD}$ can be
written in terms of the projections $\Pi_{CD}^{\mathcal{S}}$ and $\Pi
_{CD}^{\mathcal{A}}$ onto the respective symmetric and antisymmetric subspaces
as%
\begin{align}
F_{CD}  &  =\Pi_{CD}^{\mathcal{S}}-\Pi_{CD}^{\mathcal{A}},\\
I_{CD}  &  =\Pi_{CD}^{\mathcal{S}}+\Pi_{CD}^{\mathcal{A}},
\end{align}
which implies that%
\begin{align}
\frac{dI_{CD}-F_{CD}}{d^{2}-1}  &  =\frac{1}{d^{2}-1}\left(  d\Pi
_{CD}^{\mathcal{S}}+d\Pi_{CD}^{\mathcal{A}}-\left(  \Pi_{CD}^{\mathcal{S}}%
-\Pi_{CD}^{\mathcal{A}}\right)  \right) \\
&  =\frac{d-1}{d^{2}-1}\Pi_{CD}^{\mathcal{S}}+\frac{d+1}{d^{2}-1}\Pi
_{CD}^{\mathcal{A}}\\
&  =\frac{1}{d+1}\Pi_{CD}^{\mathcal{S}}+\frac{1}{d-1}\Pi_{CD}^{\mathcal{A}}.
\end{align}
Continuing, we find that%
\begin{align}
\text{Eq.~\eqref{eq:partial-transpose-cond-1-swap-simp-2}}  &  =\left(
\Pi_{AB^{\prime}}^{\mathcal{S}}-\Pi_{AB^{\prime}}^{\mathcal{A}}\right)
\otimes\left(  \Pi_{BA^{\prime}}^{\mathcal{S}}-\Pi_{BA^{\prime}}^{\mathcal{A}%
}\right)  \otimes T_{\hat{A}}(K_{\hat{A}\hat{B}})\nonumber\\
&  \qquad+\frac{1}{2}\left[
\begin{array}
[c]{c}%
\left(  \Pi_{AB^{\prime}}^{\mathcal{S}}-\Pi_{AB^{\prime}}^{\mathcal{A}%
}\right)  \otimes\left(  \frac{1}{d+1}\Pi_{BA^{\prime}}^{\mathcal{S}}+\frac
{1}{d-1}\Pi_{BA^{\prime}}^{\mathcal{A}}\right) \\
+\left(  \frac{1}{d+1}\Pi_{AB^{\prime}}^{\mathcal{S}}+\frac{1}{d-1}%
\Pi_{AB^{\prime}}^{\mathcal{A}}\right)  \otimes\left(  \Pi_{BA^{\prime}%
}^{\mathcal{S}}-\Pi_{BA^{\prime}}^{\mathcal{A}}\right)
\end{array}
\right]  \otimes T_{\hat{A}}(L_{\hat{A}\hat{B}})\nonumber\\
&  \qquad+\left(  \frac{1}{d+1}\Pi_{AB^{\prime}}^{\mathcal{S}}+\frac{1}%
{d-1}\Pi_{AB^{\prime}}^{\mathcal{A}}\right)  \otimes\left(  \frac{1}{d+1}%
\Pi_{BA^{\prime}}^{\mathcal{S}}+\frac{1}{d-1}\Pi_{BA^{\prime}}^{\mathcal{A}%
}\right)  \otimes T_{\hat{A}}(N_{\hat{A}\hat{B}})\\
&  =\Pi_{AB^{\prime}}^{\mathcal{S}}\otimes\Pi_{BA^{\prime}}^{\mathcal{S}%
}\otimes\left[  T_{\hat{A}}(K_{\hat{A}\hat{B}})+\frac{T_{\hat{A}}(L_{\hat
{A}\hat{B}})}{d+1}+\frac{T_{\hat{A}}(N_{\hat{A}\hat{B}})}{\left(  d+1\right)
^{2}}\right] \nonumber\\
&  \qquad+\Pi_{AB^{\prime}}^{\mathcal{S}}\otimes\Pi_{BA^{\prime}}%
^{\mathcal{A}}\otimes\left[  -T_{\hat{A}}(K_{\hat{A}\hat{B}})+\frac{T_{\hat
{A}}(L_{\hat{A}\hat{B}})}{2\left(  d-1\right)  }-\frac{T_{\hat{A}}(L_{\hat
{A}\hat{B}})}{2\left(  d+1\right)  }+\frac{T_{\hat{A}}(N_{\hat{A}\hat{B}}%
)}{\left(  d+1\right)  \left(  d-1\right)  }\right] \nonumber\\
&  \qquad+\Pi_{AB^{\prime}}^{\mathcal{A}}\otimes\Pi_{BA^{\prime}}%
^{\mathcal{S}}\otimes\left[  -T_{\hat{A}}(K_{\hat{A}\hat{B}})-\frac{T_{\hat
{A}}(L_{\hat{A}\hat{B}})}{2\left(  d+1\right)  }+\frac{T_{\hat{A}}(L_{\hat
{A}\hat{B}})}{2\left(  d-1\right)  }+\frac{T_{\hat{A}}(N_{\hat{A}\hat{B}}%
)}{\left(  d+1\right)  \left(  d-1\right)  }\right] \nonumber\\
&  \qquad+\Pi_{AB^{\prime}}^{\mathcal{A}}\otimes\Pi_{BA^{\prime}}%
^{\mathcal{A}}\otimes\left[  T_{\hat{A}}(K_{\hat{A}\hat{B}})-\frac{T_{\hat{A}%
}(L_{\hat{A}\hat{B}})}{d-1}+\frac{T_{\hat{A}}(N_{\hat{A}\hat{B}})}{\left(
d-1\right)  ^{2}}\right]  .
\end{align}
Since the operators in the first two factors of the tensor product are
orthogonal projectors, we then see that the condition in
\eqref{eq:ppt-psd-constr} is equivalent to the following four conditions:%
\begin{align}
T_{\hat{B}}(K_{\hat{A}\hat{B}})+\frac{T_{\hat{B}}(L_{\hat{A}\hat{B}})}%
{d+1}+\frac{T_{\hat{B}}(N_{\hat{A}\hat{B}})}{\left(  d+1\right)  ^{2}}  &
\geq0,\\
-T_{\hat{B}}(K_{\hat{A}\hat{B}})+\frac{T_{\hat{B}}(L_{\hat{A}\hat{B}}%
)}{2\left(  d-1\right)  }-\frac{T_{\hat{B}}(L_{\hat{A}\hat{B}})}{2\left(
d+1\right)  }+\frac{T_{\hat{B}}(N_{\hat{A}\hat{B}})}{\left(  d+1\right)
\left(  d-1\right)  }  &  \geq0,\\
-T_{\hat{B}}(K_{\hat{A}\hat{B}})-\frac{T_{\hat{B}}(L_{\hat{A}\hat{B}}%
)}{2\left(  d+1\right)  }+\frac{T_{\hat{B}}(L_{\hat{A}\hat{B}})}{2\left(
d-1\right)  }+\frac{T_{\hat{B}}(N_{\hat{A}\hat{B}})}{\left(  d+1\right)
\left(  d-1\right)  }  &  \geq0,\\
T_{\hat{B}}(K_{\hat{A}\hat{B}})-\frac{T_{\hat{B}}(L_{\hat{A}\hat{B}})}%
{d-1}+\frac{T_{\hat{B}}(N_{\hat{A}\hat{B}})}{\left(  d-1\right)  ^{2}}  &
\geq0.
\end{align}
It is clear that the third one is redundant. These in turn are equivalent to
the following three conditions:%
\begin{align}
T_{\hat{A}}(K_{\hat{A}\hat{B}})+\frac{T_{\hat{A}}(L_{\hat{A}\hat{B}})}%
{d+1}+\frac{T_{\hat{A}}(N_{\hat{A}\hat{B}})}{\left(  d+1\right)  ^{2}}  &
\geq0,\\
\frac{1}{d-1}\left(  \frac{T_{\hat{A}}(L_{\hat{A}\hat{B}})}{2}+\frac
{T_{\hat{A}}(N_{\hat{A}\hat{B}})}{d+1}\right)   &  \geq T_{\hat{A}}(K_{\hat
{A}\hat{B}})+\frac{T_{\hat{A}}(L_{\hat{A}\hat{B}})}{2\left(  d+1\right)  },\\
T_{\hat{A}}(K_{\hat{A}\hat{B}})+\frac{T_{\hat{A}}(N_{\hat{A}\hat{B}})}{\left(
d-1\right)  ^{2}}  &  \geq\frac{1}{d-1}\left[  T_{\hat{A}}(L_{\hat{A}\hat{B}%
})\right]  .
\end{align}
Using the linearity of the transpose operation, we can rewrite these
conditions a final time as%
\begin{align}
T_{\hat{A}}\!\left(  K_{\hat{A}\hat{B}}+\frac{L_{\hat{A}\hat{B}}}{d+1}%
+\frac{N_{\hat{A}\hat{B}}}{\left(  d+1\right)  ^{2}}\right)   &  \geq0,\\
\frac{1}{d-1}T_{\hat{A}}\!\left(  \frac{L_{\hat{A}\hat{B}}}{2}+\frac
{N_{\hat{A}\hat{B}}}{d+1}\right)   &  \geq T_{\hat{A}}\!\left(  K_{\hat{A}%
\hat{B}}+\frac{L_{\hat{A}\hat{B}}}{2\left(  d+1\right)  }\right)
,\label{eq:constraint-to-rewrite}\\
T_{\hat{A}}\!\left(  K_{\hat{A}\hat{B}}+\frac{N_{\hat{A}\hat{B}}}{\left(
d-1\right)  ^{2}}\right)   &  \geq\frac{1}{d-1}T_{\hat{A}}\!\left(  L_{\hat
{A}\hat{B}}\right)  . \label{eq:partial-transpose-cond-swap-simp-last}%
\end{align}
Since $T_{\hat{A}}(G_{\hat{A}\hat{B}})\geq0$ if and only $T_{\hat{B}}%
(G_{\hat{A}\hat{B}})\geq0$, we can equivalently write the partial transpose
constraints with respect to $T_{\hat{B}}$. Some simple algebra reduces the
inequality in \eqref{eq:constraint-to-rewrite} to
\begin{equation}
\frac{1}{d^{2}-1}T_{\hat{A}}\!\left(  L_{\hat{A}\hat{B}}+N_{\hat{A}\hat{B}%
}\right)  \geq T_{\hat{A}}\!\left(  K_{\hat{A}\hat{B}}\right)  .
\end{equation}

By combining with
\eqref{eq:opt-C-PPT-P-SDP-start-1}--\eqref{eq:opt-C-PPT-P-SDP-start-3}, we
conclude the proof of Proposition~\ref{prop:swap-sdp-simplify}.

We note here that, if we do not employ the symmetry in
\eqref{eq:swap-ch-add-symmetry-for-LOCC}, then we arrive at the following
SDP:
\begin{equation}
e_{\operatorname{PPT}}(\mathcal{S}_{AB}^{d},\rho_{\hat{A}\hat{B}}%
)=1-\sup_{\substack{K_{\hat{A}\hat{B}},L_{\hat{A}\hat{B}},\\M_{\hat{A}\hat{B}%
},N_{\hat{A}\hat{B}}\geq0}}\operatorname{Tr}[\rho_{\hat{A}\hat{B}}K_{\hat
{A}\hat{B}}], \label{eq:SDP-no-extra-symmetry}%
\end{equation}
subject to%
\begin{align}
T_{\hat{B}}\!\left(  K_{\hat{A}\hat{B}}+\frac{L_{\hat{A}\hat{B}}}{d+1}%
+\frac{M_{\hat{A}\hat{B}}}{d+1}+\frac{N_{\hat{A}\hat{B}}}{\left(  d+1\right)
^{2}}\right)   &  \geq0,\\
\frac{1}{d-1}T_{\hat{B}}\!\left(  L_{\hat{A}\hat{B}}+\frac{N_{\hat{A}\hat{B}}%
}{d+1}\right)   &  \geq T_{\hat{B}}\!\left(  K_{\hat{A}\hat{B}}+\frac
{M_{\hat{A}\hat{B}}}{d+1}\right)  ,\\
\frac{1}{d-1}T_{\hat{B}}\!\left(  M_{\hat{A}\hat{B}}+\frac{N_{\hat{A}\hat{B}}%
}{d+1}\right)   &  \geq T_{\hat{B}}\!\left(  K_{\hat{A}\hat{B}}+\frac
{L_{\hat{A}\hat{B}}}{d+1}\right)  ,\label{eq:extra-constraint}\\
T_{\hat{B}}\!\left(  K_{\hat{A}\hat{B}}+\frac{N_{\hat{A}\hat{B}}}{\left(
d-1\right)  ^{2}}\right)   &  \geq\frac{1}{d-1}T_{\hat{B}}\!\left(  L_{\hat
{A}\hat{B}}+M_{\hat{A}\hat{B}}\right)  ,\\
K_{\hat{A}\hat{B}}+L_{\hat{A}\hat{B}}+M_{\hat{A}\hat{B}}+N_{\hat{A}\hat{B}}
&  =I_{\hat{A}\hat{B}},
\end{align}
where $K_{\hat{A}\hat{B}}, L_{\hat{A}\hat{B}}, M_{\hat{A}\hat{B}}$, and
$N_{\hat{A}\hat{B}}$ are positive semi-definite Hermitian matrices and
elements of a POVM and $d$ is the dimension of the SWAP channel. This SDP was
derived in an earlier version of our paper, available at \cite{siddiqui2020},
before we noticed the extra symmetry in
\eqref{eq:swap-ch-add-symmetry-for-LOCC}. We have used this latter SDP to
evaluate some examples in our paper. The main difference between the SDP in
\eqref{eq:main-SDP-paper} and that in \eqref{eq:SDP-no-extra-symmetry} is
that, in the latter, by exploiting the symmetry in
\eqref{eq:swap-ch-add-symmetry-for-LOCC}, we can set $M_{\hat{A}\hat{B}} =
L_{\hat{A}\hat{B}}$ and eliminate the constraint in
\eqref{eq:extra-constraint}. Then we set $L_{\hat{A}\hat{B}}:= 2 L_{\hat
{A}\hat{B}}$.

\section{Channel box transformation using infidelity}

\label{app:channel-box-transform-infid}Although the topic of this appendix is
tangential to the main theme of this paper, we think it is nevertheless
important to go through a different example in which the channel infidelity as
a measure of error leads to a semi-definite program. We suspect that this
observation will have wide application in the context of quantum resource
theories \cite{CG18}.

Recall from \cite{Wang2019a} that the channel box transformation problem
refers to the task of converting a pair $(\mathcal{N}_{A\rightarrow
B},\mathcal{M}_{A\rightarrow B})$\ of channels to another pair $(\mathcal{K}%
_{C\rightarrow D},\mathcal{L}_{C\rightarrow D})$ by means of a superchannel
$\Theta_{\left(  A\rightarrow B\right)  \rightarrow\left(  C\rightarrow
D\right)  }$. In particular, the goal is to minimize the error $\varepsilon$
in the following:%
\begin{align}
\Theta(\mathcal{N}_{A\rightarrow B})  &  \approx_{\varepsilon}\mathcal{K}%
_{C\rightarrow D},\label{eq:approx-convert-ch-box}\\
\Theta(\mathcal{M}_{A\rightarrow B})  &  =\mathcal{L}_{C\rightarrow D}.
\end{align}
See \cite{Wang2019a} for a full exposition of this problem. The problem can be
phrased in a more mathematical way as follows:%
\begin{equation}
\varepsilon((\mathcal{N},\mathcal{M})\rightarrow(\mathcal{K},\mathcal{L}%
))\coloneqq \inf_{\Theta\in\text{SC}}\left\{  \varepsilon\in\left[
0,1\right]  :\Theta(\mathcal{N}_{A\rightarrow B})\approx_{\varepsilon
}\mathcal{K}_{C\rightarrow D},\Theta(\mathcal{M}_{A\rightarrow B}%
)=\mathcal{L}_{C\rightarrow D}\right\}  ,
\end{equation}
where SC\ is the set of superchannels.

In \cite{Wang2019a}, the approximation error in
\eqref{eq:approx-convert-ch-box} was taken to be with respect to normalized
diamond distance. In addition to various operational motivations for doing so,
an additional implicit motivation was that doing so led to a semi-definite program.

What we show here is that if we take the approximation error in
\eqref{eq:approx-convert-ch-box} to be with respect to channel infidelity, so
that%
\begin{equation}
\Theta(\mathcal{N}_{A\rightarrow B})\approx_{\varepsilon}\mathcal{K}%
_{C\rightarrow D}\qquad\Longleftrightarrow\qquad1-F(\Theta(\mathcal{N}%
_{A\rightarrow B}),\mathcal{K}_{C\rightarrow D})\leq\varepsilon,
\end{equation}
then we still arrive at a semi-definite program to perform the optimization.
The semi-definite program is as follows:%
\begin{equation}
\varepsilon((\mathcal{N},\mathcal{M})\rightarrow(\mathcal{K},\mathcal{L}%
))=1-\left[  \sup_{\lambda\geq0,\Gamma_{CBAD}^{\Theta}\geq0,Q_{CD}}%
\lambda\right]  ^{2}%
\end{equation}
subject to%
\begin{align}
\Gamma_{CB}^{\Theta}  &  =I_{CB},\\
\Gamma_{CBA}^{\Theta}  &  =\Gamma_{CA}^{\Theta}\otimes\frac{I_{B}}{d_{B}},\\
\Gamma_{CD}^{\mathcal{L}}  &  =\operatorname{Tr}_{AB}[T_{AB}(\Gamma
_{AB}^{\mathcal{M}})\Gamma_{CBAD}^{\Theta}],\\
\lambda I_{C}  &  \leq\operatorname{Re}[\operatorname{Tr}_{D}[Q_{CD}]],\\%
\begin{bmatrix}
\Gamma_{CD}^{\mathcal{K}} & Q_{CD}^{\dag}\\
Q_{CD} & \operatorname{Tr}_{AB}[T_{AB}(\Gamma_{AB}^{\mathcal{N}})\Gamma
_{CBAD}^{\Theta}]
\end{bmatrix}
&  \geq0.
\end{align}
The constraints $\Gamma_{CBAD}^{\Theta}\geq0$, $\Gamma_{CB}^{\Theta}=I_{CB}$,
and $\Gamma_{CBA}^{\Theta}=\Gamma_{CA}^{\Theta}\otimes\frac{I_{B}}{d_{B}}$
correspond to $\Gamma_{CBAD}^{\Theta}$ being a Choi operator for a
superchannel. The constraint $\Gamma_{CD}^{\mathcal{L}}=\operatorname{Tr}%
_{AB}[T_{AB}(\Gamma_{AB}^{\mathcal{M}})\Gamma_{CBAD}^{\Theta}]$ ensures that
the superchannel $\Theta$ transforms $\mathcal{M}_{A\rightarrow B}$\ exactly
to $\mathcal{L}_{C\rightarrow D}$. The other constraints ensure that the
superchannel $\Theta$ transforms $\mathcal{N}_{A\rightarrow B}$\ approximately
to $\mathcal{K}_{C\rightarrow D}$ with channel infidelity as the error
measure, making use of the semi-definite program for root channel fidelity
from \eqref{eq:SDP-ch-fid-1}--\eqref{eq:SDP-ch-fid-3}\ (see also \cite{KW20}).

\section{Proof of Proposition~\ref{prop:isotropic-sim-perf}%
\label{app:isotropic-proof}}

In this appendix, we establish how the semi-definite program in
\eqref{eq:SDP-no-extra-symmetry} simplifies for isotropic states. Recall from
\eqref{eq:isotropic-def}\ that an isotropic state has the following form:%
\begin{equation}
\rho_{\hat{A}\hat{B}}^{(F,d_{\hat{A}})}\coloneqq F\Phi_{\hat{A}\hat{B}%
}+\left(  1-F\right)  \frac{I_{\hat{A}\hat{B}}-\Phi_{\hat{A}\hat{B}}}%
{d_{\hat{A}}^{2}-1}.
\end{equation}
We prove that%
\begin{equation}
e_{\operatorname{PPT}}(\mathcal{S}_{AB}^{d},\rho_{\hat{A}\hat{B}}%
^{(F,d_{\hat{A}})})=\left\{
\begin{array}
[c]{cc}%
1-\frac{1}{d^{2}} & \text{if }F\leq\frac{1}{d_{\hat{A}}}\\
1-\frac{Fd_{\hat{A}}}{d^{2}} & \text{if }F>\frac{1}{d_{\hat{A}}}\text{ and
}d_{\hat{A}}\leq d^{2}\\
\frac{\left(  1-\frac{1}{d^{2}}\right)  \left(  1-F\right)  }{1-\frac
{1}{d_{\hat{A}}}} & \text{if }F>\frac{1}{d_{\hat{A}}}\text{ and }d_{\hat{A}%
}>d^{2}%
\end{array}
\right.  , \label{eq:error-ppt-isotropic-app}%
\end{equation}
as stated in Proposition~\ref{prop:isotropic-sim-perf}.

To begin with, if $F\leq1/d_{\hat{A}}$, then the resource state is separable
\cite{Horodecki2009a,Watrous2018}. So the resource state can be prepared by
LOCC, and then our previous result from Proposition~\ref{prop:no-res-sim-err}%
\ applies, so that we can conclude that the simulation error is%
\begin{equation}
1-\frac{1}{d^{2}}%
\end{equation}
in this case. So in what follows, we focus exclusively on the case when
$F>1/d_{\hat{A}}$.

Due to the objective function in \eqref{eq:SDP-no-extra-symmetry} being linear
and the constraints being linear inequalities, it follows that a convex
combination of optimal solutions is also optimal. Since the resource state is
isotropic, then this means that if $K_{\hat{A}\hat{B}}$, $L_{\hat{A}\hat{B}}$,
$M_{\hat{A}\hat{B}}$, and $N_{\hat{A}\hat{B}}$ is an optimal solution, then
$(\mathcal{U}_{\hat{A}}\otimes\overline{\mathcal{U}}_{\hat{B}})(K_{\hat{A}%
\hat{B}})$, $(\mathcal{U}_{\hat{A}}\otimes\overline{\mathcal{U}}_{\hat{B}%
})(L_{\hat{A}\hat{B}})$, $(\mathcal{U}_{\hat{A}}\otimes\overline{\mathcal{U}%
}_{\hat{B}})(M_{\hat{A}\hat{B}})$, and $(\mathcal{U}_{\hat{A}}\otimes
\overline{\mathcal{U}}_{\hat{B}})(N_{\hat{A}\hat{B}})$ is too, where
$\mathcal{U}_{\hat{A}}$ and $\mathcal{U}_{\hat{B}}$ unitary channels
corresponding to an arbitrary unitary operator $U$. Following the reasoning
above, it follows that the twirled versions of these operators is also
optimal, so that it suffices for each of $K_{\hat{A}\hat{B}}$, $L_{\hat{A}%
\hat{B}}$, $M_{\hat{A}\hat{B}}$, and $N_{\hat{A}\hat{B}}$ to have an isotropic
form:%
\begin{align}
K_{\hat{A}\hat{B}}  &  =k_{1}\Phi_{\hat{A}\hat{B}}+k_{2}\left(  I_{\hat{A}%
\hat{B}}-\Phi_{\hat{A}\hat{B}}\right)  ,\label{eq:isotropic-reduction-meas}\\
L_{\hat{A}\hat{B}}  &  =l_{1}\Phi_{\hat{A}\hat{B}}+l_{2}\left(  I_{\hat{A}%
\hat{B}}-\Phi_{\hat{A}\hat{B}}\right)  ,\\
M_{\hat{A}\hat{B}}  &  =m_{1}\Phi_{\hat{A}\hat{B}}+m_{2}\left(  I_{\hat{A}%
\hat{B}}-\Phi_{\hat{A}\hat{B}}\right)  ,\\
N_{\hat{A}\hat{B}}  &  =n_{1}\Phi_{\hat{A}\hat{B}}+n_{2}\left(  I_{\hat{A}%
\hat{B}}-\Phi_{\hat{A}\hat{B}}\right)  .
\end{align}
The objective function then evaluates to%
\begin{equation}
\operatorname{Tr}[K_{\hat{A}\hat{B}}\rho_{\hat{A}\hat{B}}^{(F,d_{\hat{A}}%
)}]=k_{1}F+k_{2}\left(  1-F\right)  .
\end{equation}
Consider that the equality constraint%
\begin{equation}
K_{\hat{A}\hat{B}}+L_{\hat{A}\hat{B}}+M_{\hat{A}\hat{B}}+N_{\hat{A}\hat{B}%
}=I_{\hat{A}\hat{B}}%
\end{equation}
implies that%
\begin{align}
k_{1}+l_{1}+m_{1}+n_{1}  &  =1,\\
k_{2}+l_{2}+m_{2}+n_{2}  &  =1.
\end{align}
We also have that all coefficients are non-negative:%
\begin{equation}
k_{1},l_{1},m_{1},n_{1},k_{2},l_{2},m_{2},n_{2}\geq0,
\label{eq:isotropic-all-coeffs-non-neg}%
\end{equation}
which follows from the constraints $K_{\hat{A}\hat{B}},L_{\hat{A}\hat{B}%
},M_{\hat{A}\hat{B}},N_{\hat{A}\hat{B}}\geq0$. All of the inequality
constraints involve the terms $T_{\hat{B}}(K_{\hat{A}\hat{B}})$, $T_{\hat{B}%
}(L_{\hat{A}\hat{B}})$, $T_{\hat{B}}(M_{\hat{A}\hat{B}})$, and $T_{\hat{B}%
}(N_{\hat{A}\hat{B}})$. Let us evaluate the first of these and the rest follow
similarly:%
\begin{align}
T_{\hat{B}}(K_{\hat{A}\hat{B}})  &  =\frac{k_{1}}{d_{\hat{A}}}F_{\hat{A}%
\hat{B}}+k_{2}\left(  I_{\hat{A}\hat{B}}-F_{\hat{A}\hat{B}}/d_{\hat{A}}\right)
\\
&  =\frac{k_{1}}{d_{\hat{A}}}\left(  \Pi_{\hat{A}\hat{B}}^{\mathcal{S}}%
-\Pi_{\hat{A}\hat{B}}^{\mathcal{A}}\right)  +k_{2}\left(  \Pi_{\hat{A}\hat{B}%
}^{\mathcal{S}}+\Pi_{\hat{A}\hat{B}}^{\mathcal{A}}\right)  -\frac{k_{2}%
}{d_{\hat{A}}}\left(  \Pi_{\hat{A}\hat{B}}^{\mathcal{S}}-\Pi_{\hat{A}\hat{B}%
}^{\mathcal{A}}\right) \\
&  =\left(  \frac{k_{1}}{d_{\hat{A}}}+k_{2}-\frac{k_{2}}{d_{\hat{A}}}\right)
\Pi_{\hat{A}\hat{B}}^{\mathcal{S}}+\left(  k_{2}+\frac{k_{2}}{d_{\hat{A}}%
}-\frac{k_{1}}{d_{\hat{A}}}\right)  \Pi_{\hat{A}\hat{B}}^{\mathcal{A}}\\
&  =\frac{1}{d_{\hat{A}}}\left[  \left(  k_{1}+k_{2}\left(  d_{\hat{A}%
}-1\right)  \right)  \Pi_{\hat{A}\hat{B}}^{\mathcal{S}}+\left(  k_{2}\left(
d_{\hat{A}}+1\right)  -k_{1}\right)  \Pi_{\hat{A}\hat{B}}^{\mathcal{A}%
}\right]  .
\end{align}
where we used the facts that%
\begin{align}
T_{\hat{B}}(\Phi_{\hat{A}\hat{B}})  &  =\frac{1}{d_{\hat{A}}}F_{\hat{A}\hat
{B}},\\
F_{\hat{A}\hat{B}}  &  =\Pi_{\hat{A}\hat{B}}^{\mathcal{S}}-\Pi_{\hat{A}\hat
{B}}^{\mathcal{A}},\\
I_{\hat{A}\hat{B}}  &  =\Pi_{\hat{A}\hat{B}}^{\mathcal{S}}+\Pi_{\hat{A}\hat
{B}}^{\mathcal{A}}.
\end{align}
Similarly,%
\begin{align}
T_{\hat{B}}(L_{\hat{A}\hat{B}})  &  =\frac{1}{d_{\hat{A}}}\left[  \left(
l_{1}+l_{2}\left(  d_{\hat{A}}-1\right)  \right)  \Pi_{\hat{A}\hat{B}%
}^{\mathcal{S}}+\left(  l_{2}\left(  d_{\hat{A}}+1\right)  -l_{1}\right)
\Pi_{\hat{A}\hat{B}}^{\mathcal{A}}\right]  ,\\
T_{\hat{B}}(M_{\hat{A}\hat{B}})  &  =\frac{1}{d_{\hat{A}}}\left[  \left(
m_{1}+m_{2}\left(  d_{\hat{A}}-1\right)  \right)  \Pi_{\hat{A}\hat{B}%
}^{\mathcal{S}}+\left(  m_{2}\left(  d_{\hat{A}}+1\right)  -m_{1}\right)
\Pi_{\hat{A}\hat{B}}^{\mathcal{A}}\right]  ,\\
T_{\hat{B}}(N_{\hat{A}\hat{B}})  &  =\frac{1}{d_{\hat{A}}}\left[  \left(
n_{1}+n_{2}\left(  d_{\hat{A}}-1\right)  \right)  \Pi_{\hat{A}\hat{B}%
}^{\mathcal{S}}+\left(  n_{2}\left(  d_{\hat{A}}+1\right)  -n_{1}\right)
\Pi_{\hat{A}\hat{B}}^{\mathcal{A}}\right]  .
\end{align}
Let us consider how each of the constraints simplify. Consider that the first
inequality constraint%
\[
T_{\hat{B}}\!\left(  K_{\hat{A}\hat{B}}+\frac{L_{\hat{A}\hat{B}}}{d+1}%
+\frac{M_{\hat{A}\hat{B}}}{d+1}+\frac{N_{\hat{A}\hat{B}}}{\left(  d+1\right)
^{2}}\right)  \geq0,
\]
simplifies to the following two inequalities:%
\begin{align}
\left(  k_{1}+k_{2}\left(  d_{\hat{A}}-1\right)  \right)  +\frac{l_{1}%
+l_{2}\left(  d_{\hat{A}}-1\right)  +m_{1}+m_{2}\left(  d_{\hat{A}}-1\right)
}{d+1}+\frac{n_{1}+n_{2}\left(  d_{\hat{A}}-1\right)  }{\left(  d+1\right)
^{2}}  &  \geq0,\label{eq:1st-ineq-actually-redund}\\
k_{2}\left(  d_{\hat{A}}+1\right)  -k_{1}+\frac{l_{2}\left(  d_{\hat{A}%
}+1\right)  -l_{1}+m_{2}\left(  d_{\hat{A}}+1\right)  -m_{1}}{(d+1)}%
+\frac{n_{2}\left(  d_{\hat{A}}+1\right)  -n_{1}}{\left(  d+1\right)  ^{2}}
&  \geq0.
\end{align}
The second inequality constraint%
\[
\frac{1}{d-1}T_{\hat{B}}\!\left(  L_{\hat{A}\hat{B}}+\frac{N_{\hat{A}\hat{B}}%
}{d+1}\right)  \geq T_{\hat{B}}\!\left(  K_{\hat{A}\hat{B}}+\frac{M_{\hat
{A}\hat{B}}}{d+1}\right)  ,
\]
simplifies to the following two inequalities:%
\begin{align}
\frac{1}{d-1}\left(  l_{1}+l_{2}\left(  d_{\hat{A}}-1\right)  +\frac
{n_{1}+n_{2}\left(  d_{\hat{A}}-1\right)  }{d+1}\right)   &  \geq k_{1}%
+k_{2}\left(  d_{\hat{A}}-1\right)  +\frac{m_{1}+m_{2}\left(  d_{\hat{A}%
}-1\right)  }{d+1},\\
\frac{1}{d-1}\left(  l_{2}\left(  d_{\hat{A}}+1\right)  -l_{1}+\frac
{n_{2}\left(  d_{\hat{A}}+1\right)  -n_{1}}{d+1}\right)   &  \geq k_{2}\left(
d_{\hat{A}}+1\right)  -k_{1}+\frac{m_{2}\left(  d_{\hat{A}}+1\right)  -m_{1}%
}{d+1}.
\end{align}
The third inequality constraint
\begin{equation}
\frac{1}{d-1}T_{\hat{B}}\!\left(  M_{\hat{A}\hat{B}}+\frac{N_{\hat{A}\hat{B}}%
}{d+1}\right)  \geq T_{\hat{B}}\!\left(  K_{\hat{A}\hat{B}}+\frac{L_{\hat
{A}\hat{B}}}{d+1}\right)  ,
\end{equation}
simplifies to the following two inequalities:%
\begin{align}
\frac{1}{d-1}\left(  m_{1}+m_{2}\left(  d_{\hat{A}}-1\right)  +\frac
{n_{1}+n_{2}\left(  d_{\hat{A}}-1\right)  }{d+1}\right)   &  \geq k_{1}%
+k_{2}\left(  d_{\hat{A}}-1\right)  +\frac{l_{1}+l_{2}\left(  d_{\hat{A}%
}-1\right)  }{d+1},\\
\frac{1}{d-1}\left(  m_{2}\left(  d_{\hat{A}}+1\right)  -m_{1}+\frac
{n_{2}\left(  d_{\hat{A}}+1\right)  -n_{1}}{d+1}\right)   &  \geq k_{2}\left(
d_{\hat{A}}+1\right)  -k_{1}+\frac{l_{2}\left(  d_{\hat{A}}+1\right)  -l_{1}%
}{d+1}.
\end{align}
The final inequality constraint%
\begin{equation}
T_{\hat{B}}\!\left(  K_{\hat{A}\hat{B}}+\frac{N_{\hat{A}\hat{B}}}{\left(
d-1\right)  ^{2}}\right)  \geq\frac{1}{d-1}T_{\hat{B}}\!\left(  L_{\hat{A}%
\hat{B}}+M_{\hat{A}\hat{B}}\right)  ,
\end{equation}
simplifies to the following two inequalities:%
\begin{align}
k_{1}+k_{2}\left(  d_{\hat{A}}-1\right)  +\frac{n_{1}+n_{2}\left(  d_{\hat{A}%
}-1\right)  }{\left(  d-1\right)  ^{2}}  &  \geq\frac{1}{d-1}\left(
l_{1}+l_{2}\left(  d_{\hat{A}}-1\right)  +m_{1}+m_{2}\left(  d_{\hat{A}%
}-1\right)  \right)  ,\\
k_{2}\left(  d_{\hat{A}}+1\right)  -k_{1}+\frac{n_{2}\left(  d_{\hat{A}%
}+1\right)  -n_{1}}{\left(  d-1\right)  ^{2}}  &  \geq\frac{1}{d-1}\left(
l_{2}\left(  d_{\hat{A}}+1\right)  -l_{1}+m_{2}\left(  d_{\hat{A}}+1\right)
-m_{1}\right)  .
\end{align}

Consider that the first inequality constraint in
\eqref{eq:1st-ineq-actually-redund} is redundant, following from
\eqref{eq:isotropic-all-coeffs-non-neg}. So it can be eliminated.

Summarizing, we have reduced the semi-definite program in
\eqref{eq:SDP-no-extra-symmetry} to the following linear program:%
\begin{equation}
1-\sup_{\substack{k_{1},l_{1},m_{1},n_{1},\\k_{2},l_{2},m_{2},n_{2}\geq
0}}\left[  k_{1}F+k_{2}\left(  1-F\right)  \right]  ,
\label{eq:app-LP-iso-simplify-1}%
\end{equation}
subject to%
\begin{align}
k_{1}+l_{1}+m_{1}+n_{1}  &  =1,\\
k_{2}+l_{2}+m_{2}+n_{2}  &  =1,
\end{align}%
\begin{equation}
k_{2}\left(  d_{\hat{A}}+1\right)  -k_{1}+\frac{l_{2}\left(  d_{\hat{A}%
}+1\right)  -l_{1}+m_{2}\left(  d_{\hat{A}}+1\right)  -m_{1}}{(d+1)}%
+\frac{n_{2}\left(  d_{\hat{A}}+1\right)  -n_{1}}{\left(  d+1\right)  ^{2}%
}\geq0,
\end{equation}%
\begin{align}
\frac{1}{d-1}\left(  l_{1}+l_{2}\left(  d_{\hat{A}}-1\right)  +\frac
{n_{1}+n_{2}\left(  d_{\hat{A}}-1\right)  }{d+1}\right)   &  \geq k_{1}%
+k_{2}\left(  d_{\hat{A}}-1\right)  +\frac{m_{1}+m_{2}\left(  d_{\hat{A}%
}-1\right)  }{d+1},\\
\frac{1}{d-1}\left(  l_{2}\left(  d_{\hat{A}}+1\right)  -l_{1}+\frac
{n_{2}\left(  d_{\hat{A}}+1\right)  -n_{1}}{d+1}\right)   &  \geq k_{2}\left(
d_{\hat{A}}+1\right)  -k_{1}+\frac{m_{2}\left(  d_{\hat{A}}+1\right)  -m_{1}%
}{d+1},
\end{align}%
\begin{align}
\frac{1}{d-1}\left(  m_{1}+m_{2}\left(  d_{\hat{A}}-1\right)  +\frac
{n_{1}+n_{2}\left(  d_{\hat{A}}-1\right)  }{d+1}\right)   &  \geq k_{1}%
+k_{2}\left(  d_{\hat{A}}-1\right)  +\frac{l_{1}+l_{2}\left(  d_{\hat{A}%
}-1\right)  }{d+1},\\
\frac{1}{d-1}\left(  m_{2}\left(  d_{\hat{A}}+1\right)  -m_{1}+\frac
{n_{2}\left(  d_{\hat{A}}+1\right)  -n_{1}}{d+1}\right)   &  \geq k_{2}\left(
d_{\hat{A}}+1\right)  -k_{1}+\frac{l_{2}\left(  d_{\hat{A}}+1\right)  -l_{1}%
}{d+1}.
\end{align}%
\begin{align}
k_{1}+k_{2}\left(  d_{\hat{A}}-1\right)  +\frac{n_{1}+n_{2}\left(  d_{\hat{A}%
}-1\right)  }{\left(  d-1\right)  ^{2}}  &  \geq\frac{1}{d-1}\left(
l_{1}+l_{2}\left(  d_{\hat{A}}-1\right)  +m_{1}+m_{2}\left(  d_{\hat{A}%
}-1\right)  \right)  ,\\
k_{2}\left(  d_{\hat{A}}+1\right)  -k_{1}+\frac{n_{2}\left(  d_{\hat{A}%
}+1\right)  -n_{1}}{\left(  d-1\right)  ^{2}}  &  \geq\frac{1}{d-1}\left(
l_{2}\left(  d_{\hat{A}}+1\right)  -l_{1}+m_{2}\left(  d_{\hat{A}}+1\right)
-m_{1}\right)  . \label{eq:app-LP-iso-simplify-final}%
\end{align}
The standard form of a linear program is as follows \cite{BV04}:%
\begin{equation}
1-\sup_{x\geq0}\left\{  c^{T}x:Ax\leq b\right\}  .
\label{eq:LP-isotropic-primal}%
\end{equation}
We can write
\eqref{eq:app-LP-iso-simplify-1}--\eqref{eq:app-LP-iso-simplify-final} in this
way, with
\begin{align}
x^{T}  &  =%
\begin{bmatrix}
k_{1} & l_{1} & m_{1} & n_{1} & k_{2} & l_{2} & m_{2} & n_{2}%
\end{bmatrix}
,\\
c^{T}  &  =%
\begin{bmatrix}
F & 0 & 0 & 0 & 1-F & 0 & 0 & 0
\end{bmatrix}
,\\
b^{T}  &  =%
\begin{bmatrix}
1 & -1 & 1 & -1 & 0 & 0 & 0 & 0 & 0 & 0 & 0
\end{bmatrix}
,
\end{align}%
\begin{equation}
A=%
\begin{bmatrix}
1 & 1 & 1 & 1 & 0 & 0 & 0 & 0\\
-1 & -1 & -1 & -1 & 0 & 0 & 0 & 0\\
0 & 0 & 0 & 0 & 1 & 1 & 1 & 1\\
0 & 0 & 0 & 0 & -1 & -1 & -1 & -1\\
1 & \frac{1}{d+1} & \frac{1}{d+1} & \frac{1}{\left(  d+1\right)  ^{2}} &
-\left(  d_{\hat{A}}+1\right)  & -\frac{d_{\hat{A}}+1}{d+1} & -\frac
{d_{\hat{A}}+1}{d+1} & -\frac{d_{\hat{A}}+1}{\left(  d+1\right)  ^{2}}\\
1 & -\frac{1}{d-1} & \frac{1}{d+1} & -\frac{1}{d^{2}-1} & d_{\hat{A}}-1 &
-\frac{d_{\hat{A}}-1}{d-1} & \frac{d_{\hat{A}}-1}{d+1} & -\frac{d_{\hat{A}}%
-1}{d^{2}-1}\\
-1 & \frac{1}{d-1} & -\frac{1}{d+1} & \frac{1}{d^{2}-1} & d_{\hat{A}}+1 &
-\frac{d_{\hat{A}}+1}{d-1} & \frac{d_{\hat{A}}+1}{d+1} & -\frac{d_{\hat{A}}%
+1}{d^{2}-1}\\
1 & \frac{1}{d+1} & -\frac{1}{d-1} & -\frac{1}{d^{2}-1} & d_{\hat{A}}-1 &
\frac{d_{\hat{A}}-1}{d+1} & -\frac{d_{\hat{A}}-1}{d-1} & -\frac{d_{\hat{A}}%
-1}{d^{2}-1}\\
-1 & -\frac{1}{d+1} & \frac{1}{d-1} & \frac{1}{d^{2}-1} & d_{\hat{A}}+1 &
\frac{d_{\hat{A}}+1}{d+1} & -\frac{d_{\hat{A}}+1}{d-1} & -\frac{d_{\hat{A}}%
+1}{d^{2}-1}\\
-1 & \frac{1}{d-1} & \frac{1}{d-1} & -\frac{1}{\left(  d-1\right)  ^{2}} &
-\left(  d_{\hat{A}}-1\right)  & \frac{d_{\hat{A}}-1}{d-1} & \frac{d_{\hat{A}%
}-1}{d-1} & -\frac{d_{\hat{A}}-1}{\left(  d-1\right)  ^{2}}\\
1 & -\frac{1}{d-1} & -\frac{1}{d-1} & \frac{1}{\left(  d-1\right)  ^{2}} &
-\left(  d_{\hat{A}}+1\right)  & \frac{d_{\hat{A}}+1}{d-1} & \frac{d_{\hat{A}%
}+1}{d-1} & -\frac{d_{\hat{A}}+1}{\left(  d-1\right)  ^{2}}%
\end{bmatrix}
.
\end{equation}

If $F>\frac{1}{d_{\hat{A}}}$ and $d_{\hat{A}}\leq d^{2}$, then the following
choices are feasible for the primal linear program in
\eqref{eq:LP-isotropic-primal}%
\begin{align}
k_{1}  &  =\frac{d_{\hat{A}}}{d^{2}},\qquad l_{1}=m_{1}=0,\qquad n_{1}%
=1-\frac{d_{\hat{A}}}{d^{2}},\qquad k_{2}=0,\\
l_{2}  &  =m_{2}=\frac{\left(  d+1\right)  ^{2}k_{1}+n_{1}-\left(  d_{\hat{A}%
}+1\right)  }{2d\left(  d_{\hat{A}}+1\right)  }=\frac{d_{\hat{A}}}%
{d^{2}\left(  d_{\hat{A}}+1\right)  },\\
n_{2}  &  =1-\frac{2d_{\hat{A}}}{d^{2}\left(  d_{\hat{A}}+1\right)  }.
\end{align}
Thus, it follows that, in the case that $F>\frac{1}{d_{\hat{A}}}$ and
$d_{\hat{A}}\leq d^{2}$,%
\begin{equation}
e_{\operatorname{PPT}}(\mathcal{S}_{AB}^{d},\rho_{\hat{A}\hat{B}}%
^{(F,d_{\hat{A}})})\leq1-\frac{Fd_{\hat{A}}}{d^{2}}.
\end{equation}

Now we turn to the dual linear program. It is given by%
\begin{equation}
1-\inf_{y\geq0}\left\{  b^{T}y:A^{T}y\geq c\right\}  .
\end{equation}
A feasible choice for the dual linear program is as follows:%
\begin{align}
y_{1}  &  \in\left[  \frac{1}{d_{\hat{A}}d^{2}},\frac{F}{d^{2}}\right]  ,\quad
y_{2}=0,\quad y_{3}=\frac{Fd_{\hat{A}}}{d^{2}}-y_{1},\\
y_{4}  &  =0,\quad y_{5}=\frac{1}{4d^{2}}\left(  d+1\right)  ^{2}\left(
F-d^{2}y_{1}\right)  ,\\
y_{6}  &  =\frac{1}{4d^{2}}\left(  d^{2}-1\right)  \left(  F+d^{2}%
y_{1}\right)  ,\quad y_{7}=0,\\
y_{8}  &  =y_{6},\quad y_{9}=0,\quad y_{10}=0,\quad y_{11}=\frac{1}{4d^{2}%
}\left(  d-1\right)  ^{2}\left(  F-d^{2}y_{1}\right)  .
\end{align}
This implies that%
\begin{equation}
e_{\operatorname{PPT}}(\mathcal{S}_{AB}^{d},\rho_{\hat{A}\hat{B}}%
^{(F,d_{\hat{A}})})\geq1-\frac{Fd_{\hat{A}}}{d^{2}}%
\end{equation}

For the final case of $F>\frac{1}{d_{\hat{A}}}$ and $d_{\hat{A}}>d^{2}$, for
the primal, one can check that the following choices are feasible%
\begin{align}
k_{1}  &  =1,\quad l_{1}=0,\quad m_{1}=0,\quad n_{1}=0,\\
k_{2}  &  =\frac{d_{\hat{A}}-d^{2}}{d^{2}(d_{\hat{A}}-1)},\\
l_{2}  &  \in\left\{
\begin{array}
[c]{cc}%
\left[  \frac{(1+d)(d^{2}+d-\left[  d_{\hat{A}}+1\right]  )d_{\hat{A}}}%
{d^{2}(d_{\hat{A}}^{2}-1)},\frac{(d-1)d_{\hat{A}}(d_{\hat{A}}+1-d^{2}%
+d)}{d^{2}(d_{\hat{A}}^{2}-1)}\right]  & \text{if }d_{\hat{A}}+1<d+d^{2}\\
\left[  0,\frac{(d-1)d_{\hat{A}}(d_{\hat{A}}+1-d^{2}+d)}{d^{2}(-1+d_{\hat{A}%
}^{2})}\right]  & \text{if }d_{\hat{A}}+1\geq d+d^{2}%
\end{array}
\right. \\
m_{2}  &  =l_{2},\\
n_{2}  &  =1-k_{2}-l_{2}-n_{2}.
\end{align}
This implies that%
\begin{align}
e_{\operatorname{PPT}}(\mathcal{S}_{AB}^{d},\rho_{\hat{A}\hat{B}}%
^{(F,d_{\hat{A}})})  &  \leq1-k_{1}F-k_{2}\left(  1-F\right) \\
&  =1-F-\frac{d_{\hat{A}}-d^{2}}{d^{2}(d_{\hat{A}}-1)}\left(  1-F\right) \\
&  =\frac{\left(  1-\frac{1}{d^{2}}\right)  \left(  1-F\right)  }{1-\frac
{1}{d_{\hat{A}}}}.
\end{align}

Now we consider the dual program for the case $F>\frac{1}{d_{\hat{A}}}$ and
$d_{\hat{A}}>d^{2}$. A feasible solution is given by%
\begin{align}
y_{1}  &  =\frac{1-F+d^{2}(d_{\hat{A}}F-1)}{d^{2}(d_{\hat{A}}-1)},\quad
y_{2}=0,\quad y_{3}=\frac{(d_{\hat{A}}-1)(F-y_{1})}{d^{2}-1},\\
y_{4}  &  =0,\quad y_{5}=0,\quad y_{6}=\frac{F-y_{1}}{2},\quad y_{7}=0,\quad
y_{8}=y_{6},\\
y_{9}  &  =y_{10}=y_{11}=0.
\end{align}
Note that%
\begin{equation}
F-y_{1}=\frac{\left(  d^{2}-1\right)  \left(  1-F\right)  }{d^{2}\left(
d_{\hat{A}}-1\right)  }\geq0.
\end{equation}
Thus, we find for the case $F>\frac{1}{d_{\hat{A}}}$ and $d_{\hat{A}}>d^{2}$
that%
\begin{align}
e_{\operatorname{PPT}}(\mathcal{S}_{AB}^{d},\rho_{\hat{A}\hat{B}}%
^{(F,d_{\hat{A}})})  &  \geq1-\left(  y_{1}+y_{3}\right) \\
&  =\frac{\left(  1-\frac{1}{d^{2}}\right)  \left(  1-F\right)  }{1-\frac
{1}{d_{\hat{A}}}}.
\end{align}
This completes the proof of the equality in \eqref{eq:error-ppt-isotropic-app}.

We now consider the achievability of the PPT\ simulation error by means of an
LOCC-assisted scheme, for the case $F>\frac{1}{d_{\hat{A}}}$ and $d_{\hat{A}%
}\leq d^{2}$. The basic idea is the same as that which we employed previously:
perform a bilateral twirl of the resource state $\rho_{\hat{A}\hat{B}%
}^{(F,d_{\hat{A}})}$ and then perform teleportation in opposite directions.
Consider that the bilateral twirl in \eqref{eq:bilat-twirl}\ realizes the
following evolution:%
\begin{equation}
X\rightarrow\operatorname{Tr}[\Phi_{AB}^{d^{2}}X]\Phi_{AB}^{d^{2}%
}+\operatorname{Tr}[(I_{AB}-\Phi_{AB}^{d^{2}})X]\frac{I_{AB}-\Phi_{AB}^{d^{2}%
}}{d^{2}-1}.
\end{equation}
Consider that%
\begin{align}
\operatorname{Tr}[\Phi_{AB}^{d^{2}}\rho_{\hat{A}\hat{B}}^{(F,d_{\hat{A}})}]
&  =F\operatorname{Tr}[\Phi_{AB}^{d^{2}}\Phi_{\hat{A}\hat{B}}]+\left(
1-F\right)  \frac{\operatorname{Tr}[\Phi_{AB}^{d^{2}}(I_{\hat{A}\hat{B}}%
-\Phi_{\hat{A}\hat{B}})]}{d_{\hat{A}}^{2}-1}\\
&  =F\frac{d_{\hat{A}}}{d^{2}}+\left(  1-F\right)  \left(  \frac{\frac
{d_{\hat{A}}}{d^{2}}-\frac{d_{\hat{A}}}{d^{2}}}{d^{2}\left(  d_{\hat{A}}%
^{2}-1\right)  }\right) \\
&  =F\frac{d_{\hat{A}}}{d^{2}},
\end{align}
where we are embedding the space $\hat{A}\hat{B}$ in the larger space $AB$ and
we used the facts that%
\begin{align}
\operatorname{Tr}[\Phi_{AB}^{d^{2}}\Phi_{\hat{A}\hat{B}}]  &  =\frac
{d_{\hat{A}}}{d^{2}},\\
\operatorname{Tr}[\Phi_{AB}^{d^{2}}I_{\hat{A}\hat{B}}]  &  =\frac{d_{\hat{A}}%
}{d^{2}}.
\end{align}
Thus,%
\begin{align}
\rho_{\hat{A}\hat{B}}^{(F,d_{\hat{A}})}  &  \rightarrow\operatorname{Tr}%
[\Phi_{AB}^{d^{2}}\rho_{\hat{A}\hat{B}}^{(F,d_{\hat{A}})}]\Phi_{AB}^{d^{2}%
}+\operatorname{Tr}[(I_{AB}-\Phi_{AB}^{d^{2}})\rho_{\hat{A}\hat{B}%
}^{(F,d_{\hat{A}})}]\frac{I_{AB}-\Phi_{AB}^{d^{2}}}{d^{2}-1}\\
&  =F\frac{d_{\hat{A}}}{d^{2}}\Phi_{AB}^{d^{2}}+\left(  1-F\frac{d_{\hat{A}}%
}{d^{2}}\right)  \frac{I_{AB}-\Phi_{AB}^{d^{2}}}{d^{2}-1}.
\end{align}
Following a similar analysis as given in
\eqref{eq:err-eval-no-res-1}--\eqref{eq:err-eval-no-res-last}, we bound the
normalized trace distance between the ideal state $\Phi_{AB}^{d^{2}}$ and this
one as follows:%
\begin{align}
&  \frac{1}{2}\left\Vert \Phi_{AB}^{d^{2}}-\left(  F\frac{d_{\hat{A}}}{d^{2}%
}\Phi_{AB}^{d^{2}}+\left(  1-F\frac{d_{\hat{A}}}{d^{2}}\right)  \frac
{I_{AB}-\Phi_{AB}^{d^{2}}}{d^{2}-1}\right)  \right\Vert _{1}\nonumber\\
&  =\frac{1}{2}\left\Vert \left(  1-F\frac{d_{\hat{A}}}{d^{2}}\right)
\Phi^{d^{2}}-\left(  1-F\frac{d_{\hat{A}}}{d^{2}}\right)  \frac{I_{AB}%
-\Phi^{d^{2}}}{d^{2}-1}\right\Vert _{1}\\
&  =1-F\frac{d_{\hat{A}}}{d^{2}}.
\end{align}
This implies that%
\begin{equation}
e_{\operatorname{LOCC}}(\mathcal{S}_{AB}^{d},\rho_{\hat{A}\hat{B}}%
^{(F,d_{\hat{A}})})\leq1-F\frac{d_{\hat{A}}}{d^{2}},
\end{equation}
and combined with the general inequality in \eqref{eq:ppt-err-locc-err}, we
conclude that%
\begin{equation}
e_{\operatorname{LOCC}}(\mathcal{S}_{AB}^{d},\rho_{\hat{A}\hat{B}}%
^{(F,d_{\hat{A}})})=1-F\frac{d_{\hat{A}}}{d^{2}},
\end{equation}
for the case $F>\frac{1}{d_{\hat{A}}}$ and $d_{\hat{A}}\leq d^{2}$.

For the case $F>\frac{1}{d_{\hat{A}}}$ and $d_{\hat{A}}>d^{2}$, it is unclear
to us at the moment how to achieve the PPT\ simulation error by means of an
LOCC-assisted scheme.

\section{Proof of Proposition~\ref{prop:werner-sim-perf}%
\label{app:werner-state-proof}}

In this appendix, we establish how the semi-definite program in
\eqref{eq:SDP-no-extra-symmetry} simplifies for Werner states. The analysis
bears some structural similarities to that in
Appendix~\ref{app:isotropic-proof}. Recall from
\eqref{eq:Werner-state-def}\ that a Werner state has the following form:%
\begin{equation}
W_{\hat{A}\hat{B}}^{(p,d_{\hat{A}})}\coloneqq \left(  1-p\right)  \frac
{2}{d_{\hat{A}}\left(  d_{\hat{A}}+1\right)  }\Pi_{\hat{A}\hat{B}%
}^{\mathcal{S}}+p\frac{2}{d_{\hat{A}}\left(  d_{\hat{A}}-1\right)  }\Pi
_{\hat{A}\hat{B}}^{\mathcal{A}},
\end{equation}
where $p\in\left[  0,1\right]  $ and $\Pi_{\hat{A}\hat{B}}^{\mathcal{S}%
}\coloneqq (I_{\hat{A}\hat{B}}+F_{\hat{A}\hat{B}})/2$ and $\Pi_{\hat{A}\hat
{B}}^{\mathcal{A}}\coloneqq (I_{\hat{A}\hat{B}}-F_{\hat{A}\hat{B}})/2$.

We prove that%
\begin{equation}
e_{\operatorname{PPT}}(\mathcal{S}_{AB},W_{\hat{A}\hat{B}}^{(p,d_{\hat{A}}%
)})=\left\{
\begin{array}
[c]{cc}%
1-\frac{1}{d^{2}} & \text{if }p\leq\frac{1}{2}\\
1-\frac{4p-2+d_{\hat{A}}}{d^{2}d_{\hat{A}}} & \text{if }p>\frac{1}{2}%
\end{array}
\right.  .
\end{equation}
If $p\leq\frac{1}{2}$, the state is separable
\cite{Horodecki2009a,Watrous2018}. So the resource state can be prepared by
LOCC, and then our previous result from Proposition~\ref{prop:no-res-sim-err}%
\ applies, so that we can conclude that the simulation error is%
\begin{equation}
1-\frac{1}{d^{2}}%
\end{equation}
in this case. So in what follows, we focus exclusively on the case when
$p>\frac{1}{2}$.

By applying a similar argument as given above
\eqref{eq:isotropic-reduction-meas} but using
\eqref{eq:werner-twirl-1}--\eqref{eq:werner-twirl-2}\ instead, it suffices for
each of $K_{\hat{A}\hat{B}}$, $L_{\hat{A}\hat{B}}$, $M_{\hat{A}\hat{B}}$, and
$N_{\hat{A}\hat{B}}$ to have the Werner form:%
\begin{align}
K_{\hat{A}\hat{B}}  &  =k_{1}\Pi_{\hat{A}\hat{B}}^{\mathcal{S}}+k_{2}\Pi
_{\hat{A}\hat{B}}^{\mathcal{A}},\\
L_{\hat{A}\hat{B}}  &  =l_{1}\Pi_{\hat{A}\hat{B}}^{\mathcal{S}}+l_{2}\Pi
_{\hat{A}\hat{B}}^{\mathcal{A}},\\
M_{\hat{A}\hat{B}}  &  =m_{1}\Pi_{\hat{A}\hat{B}}^{\mathcal{S}}+m_{2}\Pi
_{\hat{A}\hat{B}}^{\mathcal{A}},\\
N_{\hat{A}\hat{B}}  &  =n_{1}\Pi_{\hat{A}\hat{B}}^{\mathcal{S}}+n_{2}\Pi
_{\hat{A}\hat{B}}^{\mathcal{A}}.
\end{align}
The objective function evaluates to%
\begin{equation}
\operatorname{Tr}[K_{\hat{A}\hat{B}}W_{\hat{A}\hat{B}}^{(p,d_{\hat{A}}%
)}]=k_{1}\left(  1-p\right)  +k_{2}p.
\end{equation}
Consider that the equality constraint%
\begin{equation}
K_{\hat{A}\hat{B}}+L_{\hat{A}\hat{B}}+M_{\hat{A}\hat{B}}+N_{\hat{A}\hat{B}%
}=I_{\hat{A}\hat{B}}%
\end{equation}
implies that%
\begin{align}
k_{1}+l_{1}+m_{1}+n_{1}  &  =1,\\
k_{2}+l_{2}+m_{2}+n_{2}  &  =1.
\end{align}
We also have that all coefficients are non-negative:%
\begin{equation}
k_{1},l_{1},m_{1},n_{1},k_{2},l_{2},m_{2},n_{2}\geq0.
\label{eq:werner-vars-non-neg}%
\end{equation}
All of the inequality constraints involve the terms $T_{\hat{B}}(K_{\hat
{A}\hat{B}})$, $T_{\hat{B}}(L_{\hat{A}\hat{B}})$, $T_{\hat{B}}(M_{\hat{A}%
\hat{B}})$, and $T_{\hat{B}}(N_{\hat{A}\hat{B}})$. Let us evaluate the first
of these and the rest follow similarly:%
\begin{align}
T_{\hat{B}}(K_{\hat{A}\hat{B}})  &  =T_{\hat{B}}(k_{1}\Pi_{\hat{A}\hat{B}%
}^{\mathcal{S}}+k_{2}\Pi_{\hat{A}\hat{B}}^{\mathcal{A}})\\
&  =k_{1}T_{\hat{B}}(\Pi_{\hat{A}\hat{B}}^{\mathcal{S}})+k_{2}T_{\hat{B}}%
(\Pi_{\hat{A}\hat{B}}^{\mathcal{A}})\\
&  =k_{1}T_{\hat{B}}\left(  \frac{I_{\hat{A}\hat{B}}+F_{\hat{A}\hat{B}}}%
{2}\right)  +k_{2}T_{\hat{B}}\left(  \frac{I_{\hat{A}\hat{B}}-F_{\hat{A}%
\hat{B}}}{2}\right) \\
&  =\frac{k_{1}}{2}I_{\hat{A}\hat{B}}+\frac{k_{1}d_{\hat{A}}}{2}\Phi_{\hat
{A}\hat{B}}+\frac{k_{2}}{2}I_{\hat{A}\hat{B}}-\frac{k_{2}d_{\hat{A}}}{2}%
\Phi_{\hat{A}\hat{B}}\\
&  =\frac{d_{\hat{A}}}{2}\left(  k_{1}-k_{2}\right)  \Phi_{\hat{A}\hat{B}%
}+\frac{1}{2}\left(  k_{1}+k_{2}\right)  I_{\hat{A}\hat{B}}\\
&  =\frac{d_{\hat{A}}}{2}\left(  k_{1}-k_{2}\right)  \Phi_{\hat{A}\hat{B}%
}+\frac{1}{2}\left(  k_{1}+k_{2}\right)  \left(  I_{\hat{A}\hat{B}}-\Phi
_{\hat{A}\hat{B}}+\Phi_{\hat{A}\hat{B}}\right) \\
&  =\left[  \frac{d_{\hat{A}}}{2}\left(  k_{1}-k_{2}\right)  +\frac{1}%
{2}\left(  k_{1}+k_{2}\right)  \right]  \Phi_{\hat{A}\hat{B}}+\frac{1}%
{2}\left(  k_{1}+k_{2}\right)  \left(  I_{\hat{A}\hat{B}}-\Phi_{\hat{A}\hat
{B}}\right) \\
&  =\left[  \frac{d_{\hat{A}}+1}{2}k_{1}-\frac{d_{\hat{A}}-1}{2}k_{2}\right]
\Phi_{\hat{A}\hat{B}}+\frac{1}{2}\left(  k_{1}+k_{2}\right)  \left(
I_{\hat{A}\hat{B}}-\Phi_{\hat{A}\hat{B}}\right)  .
\end{align}
For the other operators $L_{\hat{A}\hat{B}}$, $M_{\hat{A}\hat{B}}$, and
$N_{\hat{A}\hat{B}}$, we have that%
\begin{align}
T_{\hat{B}}(L_{\hat{A}\hat{B}})  &  =\left[  \frac{d_{\hat{A}}+1}{2}%
l_{1}-\frac{d_{\hat{A}}-1}{2}l_{2}\right]  \Phi_{\hat{A}\hat{B}}+\frac{1}%
{2}\left(  l_{1}+l_{2}\right)  \left(  I_{\hat{A}\hat{B}}-\Phi_{\hat{A}\hat
{B}}\right)  ,\\
T_{\hat{B}}(M_{\hat{A}\hat{B}})  &  =\left[  \frac{d_{\hat{A}}+1}{2}%
m_{1}-\frac{d_{\hat{A}}-1}{2}m_{2}\right]  \Phi_{\hat{A}\hat{B}}+\frac{1}%
{2}\left(  m_{1}+m_{2}\right)  \left(  I_{\hat{A}\hat{B}}-\Phi_{\hat{A}\hat
{B}}\right)  ,\\
T_{\hat{B}}(N_{\hat{A}\hat{B}})  &  =\left[  \frac{d_{\hat{A}}+1}{2}%
n_{1}-\frac{d_{\hat{A}}-1}{2}n_{2}\right]  \Phi_{\hat{A}\hat{B}}+\frac{1}%
{2}\left(  n_{1}+n_{2}\right)  \left(  I_{\hat{A}\hat{B}}-\Phi_{\hat{A}\hat
{B}}\right)  .
\end{align}
Now we evaluate the inequalities from \eqref{eq:SDP-no-extra-symmetry}. We
begin with the following one%
\[
T_{\hat{B}}\!\left(  K_{\hat{A}\hat{B}}+\frac{L_{\hat{A}\hat{B}}}{d+1}%
+\frac{M_{\hat{A}\hat{B}}}{d+1}+\frac{N_{\hat{A}\hat{B}}}{\left(  d+1\right)
^{2}}\right)  \geq0,
\]
and observe that it reduces to the following two inequalities:%
\begin{multline}
\frac{d_{\hat{A}}+1}{2}k_{1}-\frac{d_{\hat{A}}-1}{2}k_{2}+\frac{1}{\left(
d+1\right)  ^{2}}\left(  \frac{d_{\hat{A}}+1}{2}n_{1}-\frac{d_{\hat{A}}-1}%
{2}n_{2}\right) \\
+\frac{1}{d+1}\left(  \frac{d_{\hat{A}}+1}{2}l_{1}-\frac{d_{\hat{A}}-1}%
{2}l_{2}+\frac{d_{\hat{A}}+1}{2}m_{1}-\frac{d_{\hat{A}}-1}{2}m_{2}\right)
\geq0
\end{multline}%
\begin{equation}
\frac{1}{2}\left(  k_{1}+k_{2}+\frac{l_{1}+l_{2}+m_{1}+m_{2}}{d+1}+\frac
{n_{1}+n_{2}}{\left(  d+1\right)  ^{2}}\right)  \geq0.
\end{equation}
The first one simplifies as follows:%
\begin{equation}
\left(  d_{\hat{A}}+1\right)  \left(  k_{1}+\frac{l_{1}+m_{1}}{d+1}%
+\frac{n_{1}}{\left(  d+1\right)  ^{2}}\right)  \geq\left(  d_{\hat{A}%
}-1\right)  \left(  k_{2}+\frac{l_{2}+m_{2}}{d+1}+\frac{n_{2}}{\left(
d+1\right)  ^{2}}\right)  .
\end{equation}
The second inequality is redundant, as a consequence of
\eqref{eq:werner-vars-non-neg}. The following inequality%
\[
\frac{1}{d-1}T_{\hat{B}}\!\left(  L_{\hat{A}\hat{B}}+\frac{N_{\hat{A}\hat{B}}%
}{d+1}\right)  \geq T_{\hat{B}}\!\left(  K_{\hat{A}\hat{B}}+\frac{M_{\hat
{A}\hat{B}}}{d+1}\right)  ,
\]
reduces to the following two inequalities:%
\begin{equation}
\frac{1}{d-1}\left(  \frac{d_{\hat{A}}+1}{2}l_{1}-\frac{d_{\hat{A}}-1}{2}%
l_{2}+\frac{\frac{d_{\hat{A}}+1}{2}n_{1}-\frac{d_{\hat{A}}-1}{2}n_{2}}%
{d+1}\right)  \geq\frac{d_{\hat{A}}+1}{2}k_{1}-\frac{d_{\hat{A}}-1}{2}%
k_{2}+\frac{\frac{d_{\hat{A}}+1}{2}m_{1}-\frac{d_{\hat{A}}-1}{2}m_{2}}{d+1}%
\end{equation}%
\begin{equation}
\frac{1}{2\left(  d-1\right)  }\left(  l_{1}+l_{2}+\frac{n_{1}+n_{2}}%
{d+1}\right)  \geq\frac{1}{2}\left(  k_{1}+k_{2}+\frac{m_{1}+m_{2}}%
{d+1}\right)  .
\end{equation}
The first inequality simplifies as follows:%
\begin{equation}
\left(  \frac{d_{\hat{A}}+1}{d-1}\right)  \left(  l_{1}+\frac{n_{1}}%
{d+1}\right)  +\left(  d_{\hat{A}}-1\right)  \left(  k_{2}+\frac{m_{2}}%
{d+1}\right)  \geq\left(  d_{\hat{A}}+1\right)  \left(  k_{1}+\frac{m_{1}%
}{d+1}\right)  +\left(  \frac{d_{\hat{A}}-1}{d-1}\right)  \left(  l_{2}%
+\frac{n_{2}}{d+1}\right)
\end{equation}
and the second one as%
\begin{equation}
\left(  \frac{1}{d-1}\right)  \left(  l_{1}+l_{2}+\frac{n_{1}+n_{2}}%
{d+1}\right)  \geq k_{1}+k_{2}+\frac{m_{1}+m_{2}}{d+1}%
\end{equation}
The following inequality%
\begin{equation}
\frac{1}{d-1}T_{\hat{B}}\!\left(  M_{\hat{A}\hat{B}}+\frac{N_{\hat{A}\hat{B}}%
}{d+1}\right)  \geq T_{\hat{B}}\!\left(  K_{\hat{A}\hat{B}}+\frac{L_{\hat
{A}\hat{B}}}{d+1}\right)
\end{equation}
reduces to the following two inequalities:%
\begin{equation}
\frac{1}{d-1}\left(  \frac{d_{\hat{A}}+1}{2}m_{1}-\frac{d_{\hat{A}}-1}{2}%
m_{2}+\frac{\frac{d_{\hat{A}}+1}{2}n_{1}-\frac{d_{\hat{A}}-1}{2}n_{2}}%
{d+1}\right)  \geq\frac{d_{\hat{A}}+1}{2}k_{1}-\frac{d_{\hat{A}}-1}{2}%
k_{2}+\frac{\frac{d_{\hat{A}}+1}{2}l_{1}-\frac{d_{\hat{A}}-1}{2}l_{2}}{d+1}%
\end{equation}%
\begin{equation}
\frac{1}{2\left(  d-1\right)  }\left(  m_{1}+m_{2}+\frac{n_{1}+n_{2}}%
{d+1}\right)  \geq\frac{1}{2}\left(  k_{1}+k_{2}+\frac{l_{1}+l_{2}}%
{d+1}\right)
\end{equation}
The first simplifies as follows:%
\begin{equation}
\frac{d_{\hat{A}}+1}{d-1}\left(  m_{1}+\frac{n_{1}}{d+1}\right)  +\left(
d_{\hat{A}}-1\right)  \left(  k_{2}+\frac{l_{2}}{d+1}\right)  \geq
\frac{d_{\hat{A}}-1}{d-1}\left(  m_{2}+\frac{n_{2}}{d+1}\right)  +\left(
d_{\hat{A}}+1\right)  \left(  k_{1}+\frac{l_{1}}{d+1}\right)
\end{equation}
and the second as%
\begin{equation}
\frac{1}{d-1}\left(  m_{1}+m_{2}+\frac{n_{1}+n_{2}}{d+1}\right)  \geq
k_{1}+k_{2}+\frac{l_{1}+l_{2}}{d+1}.
\end{equation}
The final inequality%
\begin{equation}
T_{\hat{B}}\!\left(  K_{\hat{A}\hat{B}}+\frac{N_{\hat{A}\hat{B}}}{\left(
d-1\right)  ^{2}}\right)  \geq\frac{1}{d-1}T_{\hat{B}}\!\left(  L_{\hat{A}%
\hat{B}}+M_{\hat{A}\hat{B}}\right)
\end{equation}
reduces to the following two inequalities:%
\begin{equation}
\frac{d_{\hat{A}}+1}{2}k_{1}-\frac{d_{\hat{A}}-1}{2}k_{2}+\frac{\frac
{d_{\hat{A}}+1}{2}n_{1}-\frac{d_{\hat{A}}-1}{2}n_{2}}{\left(  d-1\right)
^{2}}\geq\frac{1}{d-1}\left(  \frac{d_{\hat{A}}+1}{2}l_{1}-\frac{d_{\hat{A}%
}-1}{2}l_{2}+\frac{d_{\hat{A}}+1}{2}m_{1}-\frac{d_{\hat{A}}-1}{2}m_{2}\right)
\end{equation}%
\begin{equation}
\frac{1}{2}\left(  k_{1}+k_{2}+\frac{n_{1}+n_{2}}{\left(  d-1\right)  ^{2}%
}\right)  \geq\frac{1}{2\left(  d-1\right)  }\left(  l_{1}+l_{2}+m_{1}%
+m_{2}\right)  .
\end{equation}
The first inequality simplifies as follows:%
\begin{equation}
\left(  d_{\hat{A}}+1\right)  \left(  k_{1}+\frac{n_{1}}{\left(  d-1\right)
^{2}}\right)  +\frac{d_{\hat{A}}-1}{d-1}\left(  l_{2}+m_{2}\right)
\geq\left(  d_{\hat{A}}-1\right)  \left(  k_{2}+\frac{n_{2}}{\left(
d-1\right)  ^{2}}\right)  +\frac{d_{\hat{A}}+1}{d-1}\left(  l_{1}%
+m_{1}\right)
\end{equation}
and the second as%
\begin{equation}
k_{1}+k_{2}+\frac{n_{1}+n_{2}}{\left(  d-1\right)  ^{2}}\geq\frac{1}%
{d-1}\left(  l_{1}+l_{2}+m_{1}+m_{2}\right)  .
\end{equation}
Summarizing, we have reduced the optimization problem to the following linear
program:%
\begin{equation}
1-\sup_{\substack{k_{1},l_{1},m_{1},n_{1},\\k_{2},l_{2},m_{2},n_{2}\geq
0}}\left[  k_{1}\left(  1-p\right)  +k_{2}p\right]
\end{equation}
subject to%
\begin{align}
k_{1}+l_{1}+m_{1}+n_{1}  &  =1,\\
k_{2}+l_{2}+m_{2}+n_{2}  &  =1,
\end{align}%
\begin{align}
\left(  d_{\hat{A}}+1\right)  \left(  k_{1}+\frac{l_{1}+m_{1}}{d+1}%
+\frac{n_{1}}{\left(  d+1\right)  ^{2}}\right)   &  \geq\left(  d_{\hat{A}%
}-1\right)  \left(  k_{2}+\frac{l_{2}+m_{2}}{d+1}+\frac{n_{2}}{\left(
d+1\right)  ^{2}}\right)  ,\\
\left(  \frac{d_{\hat{A}}+1}{d-1}\right)  \left(  l_{1}+\frac{n_{1}}%
{d+1}\right)  +\left(  d_{\hat{A}}-1\right)  \left(  k_{2}+\frac{m_{2}}%
{d+1}\right)   &  \geq\left(  d_{\hat{A}}+1\right)  \left(  k_{1}+\frac{m_{1}%
}{d+1}\right)  +\left(  \frac{d_{\hat{A}}-1}{d-1}\right)  \left(  l_{2}%
+\frac{n_{2}}{d+1}\right)  ,\\
\left(  \frac{1}{d-1}\right)  \left(  l_{1}+l_{2}+\frac{n_{1}+n_{2}}%
{d+1}\right)   &  \geq k_{1}+k_{2}+\frac{m_{1}+m_{2}}{d+1},\\
\frac{d_{\hat{A}}+1}{d-1}\left(  m_{1}+\frac{n_{1}}{d+1}\right)  +\left(
d_{\hat{A}}-1\right)  \left(  k_{2}+\frac{l_{2}}{d+1}\right)   &  \geq
\frac{d_{\hat{A}}-1}{d-1}\left(  m_{2}+\frac{n_{2}}{d+1}\right)  +\left(
d_{\hat{A}}+1\right)  \left(  k_{1}+\frac{l_{1}}{d+1}\right)  ,\\
\frac{1}{d-1}\left(  m_{1}+m_{2}+\frac{n_{1}+n_{2}}{d+1}\right)   &  \geq
k_{1}+k_{2}+\frac{l_{1}+l_{2}}{d+1},\\
\left(  d_{\hat{A}}+1\right)  \left(  k_{1}+\frac{n_{1}}{\left(  d-1\right)
^{2}}\right)  +\frac{d_{\hat{A}}-1}{d-1}\left(  l_{2}+m_{2}\right)   &
\geq\left(  d_{\hat{A}}-1\right)  \left(  k_{2}+\frac{n_{2}}{\left(
d-1\right)  ^{2}}\right)  +\frac{d_{\hat{A}}+1}{d-1}\left(  l_{1}%
+m_{1}\right)  ,\\
k_{1}+k_{2}+\frac{n_{1}+n_{2}}{\left(  d-1\right)  ^{2}}  &  \geq\frac{1}%
{d-1}\left(  l_{1}+l_{2}+m_{1}+m_{2}\right)  .
\end{align}
We can write this in the standard form of a linear program as follows:%
\begin{equation}
1-\sup_{x\geq0}\left\{  c^{T}x:Ax\leq b\right\}  ,
\end{equation}
where%
\begin{align}
x^{T}  &  =%
\begin{bmatrix}
k_{1} & l_{1} & m_{1} & n_{1} & k_{2} & l_{2} & m_{2} & n_{2}%
\end{bmatrix}
,\\
c^{T}  &  =%
\begin{bmatrix}
1-p & 0 & 0 & 0 & p & 0 & 0 & 0
\end{bmatrix}
,\\
b^{T}  &  =%
\begin{bmatrix}
1 & -1 & 1 & -1 & 0 & 0 & 0 & 0 & 0 & 0 & 0
\end{bmatrix}
,
\end{align}%
\begin{equation}
A=%
\begin{bmatrix}
1 & 1 & 1 & 1 & 0 & 0 & 0 & 0\\
-1 & -1 & -1 & -1 & 0 & 0 & 0 & 0\\
0 & 0 & 0 & 0 & 1 & 1 & 1 & 1\\
0 & 0 & 0 & 0 & -1 & -1 & -1 & -1\\
-\left(  d_{\hat{A}}+1\right)  & -\frac{d_{\hat{A}}+1}{d+1} & -\frac
{d_{\hat{A}}+1}{d+1} & -\frac{d_{\hat{A}}+1}{\left(  d+1\right)  ^{2}} &
\left(  d_{\hat{A}}-1\right)  & \frac{d_{\hat{A}}-1}{d+1} & \frac{d_{\hat{A}%
}-1}{d+1} & \frac{d_{\hat{A}}-1}{\left(  d+1\right)  ^{2}}\\
\left(  d_{\hat{A}}+1\right)  & -\frac{d_{\hat{A}}+1}{d-1} & \frac{d_{\hat{A}%
}+1}{d+1} & -\frac{d_{\hat{A}}+1}{d^{2}-1} & -\left(  d_{\hat{A}}-1\right)  &
\frac{d_{\hat{A}}-1}{d-1} & -\frac{d_{\hat{A}}-1}{d+1} & \frac{d_{\hat{A}}%
+1}{d^{2}-1}\\
1 & -\frac{1}{d-1} & \frac{1}{d+1} & -\frac{1}{d^{2}-1} & 1 & -\frac{1}{d-1} &
\frac{1}{d+1} & -\frac{1}{d^{2}-1}\\
\left(  d_{\hat{A}}+1\right)  & \frac{d_{\hat{A}}+1}{d+1} & -\frac{d_{\hat{A}%
}+1}{d-1} & -\frac{d_{\hat{A}}+1}{d^{2}-1} & -\left(  d_{\hat{A}}-1\right)  &
-\frac{d_{\hat{A}}-1}{d+1} & \frac{d_{\hat{A}}-1}{d-1} & \frac{d_{\hat{A}}%
-1}{d^{2}-1}\\
1 & \frac{1}{d+1} & -\frac{1}{d-1} & -\frac{1}{d^{2}-1} & 1 & \frac{1}{d+1} &
-\frac{1}{d-1} & -\frac{1}{d^{2}-1}\\
-\left(  d_{\hat{A}}+1\right)  & \frac{d_{\hat{A}}+1}{d-1} & \frac{d_{\hat{A}%
}+1}{d-1} & -\frac{d_{\hat{A}}+1}{\left(  d-1\right)  ^{2}} & \left(
d_{\hat{A}}-1\right)  & -\frac{d_{\hat{A}}-1}{d-1} & -\frac{d_{\hat{A}}%
-1}{d-1} & \frac{d_{\hat{A}}-1}{\left(  d-1\right)  ^{2}}\\
-1 & \frac{1}{d-1} & \frac{1}{d-1} & -\frac{1}{\left(  d-1\right)  ^{2}} &
-1 & \frac{1}{d-1} & \frac{1}{d-1} & -\frac{1}{\left(  d-1\right)  ^{2}}%
\end{bmatrix}
.
\end{equation}
Let us refer to this as the primal linear program. A feasible solution for the
primal linear program is as follows:%
\begin{align}
k_{1}  &  =\frac{d_{\hat{A}}-2}{d^{2}d_{\hat{A}}},\\
l_{1}  &  \in\left[  \frac{2}{d^{2}(d_{\hat{A}}+1)},\frac{2+d(d_{\hat{A}%
}-1)-d_{\hat{A}}}{d^{2}d_{\hat{A}}}\right]  ,\\
m_{1}  &  =l_{1},\quad n_{1}=1-k_{1}-l_{1}-m_{1},\\
k_{2}  &  =\frac{d_{\hat{A}}+2}{d^{2}d_{\hat{A}}},\\
l_{2}  &  =\frac{d^{2}(d_{\hat{A}}+1)l_{1}-2}{d^{2}(d_{\hat{A}}-1)},\quad
m_{2}=l_{2},\quad n_{2}=1-k_{2}-m_{2}-n_{2}.
\end{align}
Then we conclude that%
\begin{align}
e_{\operatorname{PPT}}(\mathcal{S}_{AB},W_{\hat{A}\hat{B}}^{(p,d_{\hat{A}})})
&  \leq1-\left[  \left(  1-p\right)  k_{1}+pk_{2}\right] \\
&  =1-\frac{d_{\hat{A}}-2+4p}{d^{2}d_{\hat{A}}}.
\end{align}

The dual linear program is given by%
\begin{equation}
1-\inf_{y\geq0}\left\{  b^{T}y:A^{T}y\geq c\right\}  .
\end{equation}
A feasible choice for the dual variables is%
\begin{align}
y_{1}  &  =\frac{(d_{\hat{A}}+2)p-1}{d^{2}d_{\hat{A}}},\qquad y_{2}=0,\\
y_{3}  &  =\frac{d_{\hat{A}}\left(  1-p\right)  +2p-1}{d^{2}d_{\hat{A}}%
},\qquad y_{4}=0,\\
y_{5}  &  =\frac{(d+1)^{2}(2p-1)}{4d^{2}d_{\hat{A}}},\qquad y_{6}=0,\\
y_{7}  &  =\frac{(d^{2}-1)(2p-1+d_{\hat{A}})}{4d^{2}d_{\hat{A}}},\\
y_{8}  &  =0,\qquad y_{9}=y_{7},\\
y_{10}  &  =\frac{(d-1)^{2}(2p-1)}{4d^{2}d_{\hat{A}}},\qquad y_{11}=0.
\end{align}
Then we find that%
\begin{align}
e_{\operatorname{PPT}}(\mathcal{S}_{AB},W_{\hat{A}\hat{B}}^{(p,d_{\hat{A}})})
&  \geq1-\left(  y_{1}+y_{3}\right) \\
&  =1-\frac{4p-2+d_{\hat{A}}}{d^{2}d_{\hat{A}}},
\end{align}
and finally conclude that%
\begin{equation}
e_{\operatorname{PPT}}(\mathcal{S}_{AB},W_{\hat{A}\hat{B}}^{(p,d_{\hat{A}}%
)})=1-\frac{4p-2+d_{\hat{A}}}{d^{2}d_{\hat{A}}}.
\end{equation}

\section{Proof of Proposition~\ref{prop:swap-sdp-simplify-BCQT}}

\label{app:simplified-SDP-BCQT}

A proof for Proposition~\ref{prop:swap-sdp-simplify-BCQT}\ proceeds similarly
to the proof for Proposition~\ref{prop:swap-sdp-simplify}, as given in
Appendix~\ref{app:SDP-simplify-swap}. The main goal is to simplify the
semi-definite program from Proposition~\ref{prop:gen-MP-SDP}.

Let $\mathcal{P}_{AB\hat{A}\hat{B}\hat{C}\rightarrow A^{\prime}B^{\prime}}$ be
an arbitrary multipartite C-PPT-P\ channel, to be considered for bidirectional
controlled teleportation. Its Choi operator $P_{AB\hat{A}\hat{B}\hat
{C}A^{\prime}B^{\prime}}$\ obeys the following conditions:%
\begin{align}
P_{AB\hat{A}\hat{B}\hat{C}A^{\prime}B^{\prime}}  &  \geq
0,\label{eq:MP-swap-simp-CP}\\
\operatorname{Tr}_{A^{\prime}B^{\prime}}[P_{AB\hat{A}\hat{B}\hat{C}A^{\prime
}B^{\prime}}]  &  =I_{AB\hat{A}\hat{B}\hat{C}},\label{eq:MP-swap-simp-TP}\\
T_{A\hat{A}A^{\prime}}(P_{AB\hat{A}\hat{B}\hat{C}A^{\prime}B^{\prime}})  &
\geq0,\label{eq:MP-swap-simp-CPPTP-1}\\
T_{B\hat{B}B^{\prime}}(P_{AB\hat{A}\hat{B}\hat{C}A^{\prime}B^{\prime}})  &
\geq0,\label{eq:MP-swap-simp-CPPTP-2}\\
T_{\hat{C}}(P_{AB\hat{A}\hat{B}\hat{C}A^{\prime}B^{\prime}})  &  \geq0.
\label{eq:MP-swap-simp-CPPTP-3}%
\end{align}
The simulation error for a particular channel $\mathcal{P}_{AB\hat{A}\hat
{B}\hat{C}\rightarrow A^{\prime}B^{\prime}}$\ is as follows:%
\begin{equation}
\frac{1}{2}\left\Vert \mathcal{S}_{AB\rightarrow A^{\prime}B^{\prime}}%
^{d}-\mathcal{P}_{AB\hat{A}\hat{B}\hat{C}\rightarrow A^{\prime}B^{\prime}%
}\circ\mathcal{A}_{\hat{A}\hat{B}\hat{C}}^{\rho}\right\Vert _{\diamond}.
\end{equation}
By following the same reasoning in
\eqref{eq:swap-symmetry-simplify-1}--\eqref{eq:P-choi-simplified}, it suffices
to optimize over symmetrized channels $\widetilde{\mathcal{P}}_{AB\hat{A}%
\hat{B}\hat{C}\rightarrow A^{\prime}B^{\prime}}$ with Choi operator given by%
\begin{multline}
\widetilde{P}_{AB\hat{A}\hat{B}\hat{C}A^{\prime}B^{\prime}}=\Gamma
_{AB^{\prime}}\otimes\Gamma_{BA^{\prime}}\otimes K_{\hat{A}\hat{B}\hat{C}%
}+\Gamma_{AB^{\prime}}\otimes\frac{dI_{BA^{\prime}}-\Gamma_{BA^{\prime}}%
}{d^{2}-1}\otimes L_{\hat{A}\hat{B}\hat{C}}\\
+\frac{dI_{AB^{\prime}}-\Gamma_{AB^{\prime}}}{d^{2}-1}\otimes\Gamma
_{BA^{\prime}}\otimes M_{\hat{A}\hat{B}\hat{C}}+\frac{dI_{AB^{\prime}}%
-\Gamma_{AB^{\prime}}}{d^{2}-1}\otimes\frac{dI_{BA^{\prime}}-\Gamma
_{BA^{\prime}}}{d^{2}-1}\otimes N_{\hat{A}\hat{B}\hat{C}},
\end{multline}
where%
\begin{align}
K_{\hat{A}\hat{B}\hat{C}},L_{\hat{A}\hat{B}\hat{C}},M_{\hat{A}\hat{B}\hat{C}%
},N_{\hat{A}\hat{B}\hat{C}}  &  \geq0,\\
K_{\hat{A}\hat{B}\hat{C}}+L_{\hat{A}\hat{B}\hat{C}}+M_{\hat{A}\hat{B}\hat{C}%
}+N_{\hat{A}\hat{B}\hat{C}}  &  =I_{\hat{A}\hat{B}\hat{C}}.
\end{align}
This Choi operator satisfies the conditions in \eqref{eq:MP-swap-simp-CP} and
\eqref{eq:MP-swap-simp-TP}. Then we apply the conditions in
\eqref{eq:MP-swap-simp-CPPTP-1}--\eqref{eq:MP-swap-simp-CPPTP-3} to determine
further conditions on $K_{\hat{A}\hat{B}\hat{C}}$, $L_{\hat{A}\hat{B}\hat{C}}%
$, $M_{\hat{A}\hat{B}\hat{C}}$, and $N_{\hat{A}\hat{B}\hat{C}}$. Following the
same reasoning given in
\eqref{eq:partial-transpose-cond-1-swap-simp}--\eqref{eq:partial-transpose-cond-swap-simp-last},
we conclude that the following conditions hold%
\begin{align*}
T_{\hat{S}}\!\left(  K_{\hat{A}\hat{B}\hat{C}}+\frac{L_{\hat{A}\hat{B}\hat{C}%
}}{d+1}+\frac{M_{\hat{A}\hat{B}\hat{C}}}{d+1}+\frac{N_{\hat{A}\hat{B}\hat{C}}%
}{\left(  d+1\right)  ^{2}}\right)   &  \geq0,\\
\frac{1}{d-1}T_{\hat{S}}\!\left(  L_{\hat{A}\hat{B}\hat{C}}+\frac{N_{\hat
{A}\hat{B}\hat{C}}}{d+1}\right)   &  \geq T_{\hat{S}}\!\left(  K_{\hat{A}%
\hat{B}\hat{C}}+\frac{M_{\hat{A}\hat{B}\hat{C}}}{d+1}\right)  ,\\
\frac{1}{d-1}T_{\hat{S}}\!\left(  M_{\hat{A}\hat{B}\hat{C}}+\frac{N_{\hat
{A}\hat{B}\hat{C}}}{d+1}\right)   &  \geq T_{\hat{S}}\!\left(  K_{\hat{A}%
\hat{B}\hat{C}}+\frac{L_{\hat{A}\hat{B}\hat{C}}}{d+1}\right)  ,\\
T_{\hat{S}}\!\left(  K_{\hat{A}\hat{B}\hat{C}}+\frac{N_{\hat{A}\hat{B}\hat{C}%
}}{\left(  d-1\right)  ^{2}}\right)   &  \geq\frac{1}{d-1}T_{\hat{S}}\!\left(
L_{\hat{A}\hat{B}\hat{C}}+M_{\hat{A}\hat{B}\hat{C}}\right)  .
\end{align*}
for $\hat{S}\in\{\hat{A},\hat{B}\}$. Finally imposing the condition in
\eqref{eq:MP-swap-simp-CPPTP-3} and using the orthogonality of $\Gamma
_{AB^{\prime}}\otimes\Gamma_{BA^{\prime}}$, $\Gamma_{AB^{\prime}}\otimes
\frac{dI_{BA^{\prime}}-\Gamma_{BA^{\prime}}}{d^{2}-1}$, $\frac{dI_{AB^{\prime
}}-\Gamma_{AB^{\prime}}}{d^{2}-1}\otimes\Gamma_{BA^{\prime}}$, and
$\frac{dI_{AB^{\prime}}-\Gamma_{AB^{\prime}}}{d^{2}-1}\otimes\frac
{dI_{BA^{\prime}}-\Gamma_{BA^{\prime}}}{d^{2}-1}$, we conclude that%
\begin{equation}
T_{\hat{C}}(K_{\hat{A}\hat{B}\hat{C}}),T_{\hat{C}}(L_{\hat{A}\hat{B}\hat{C}%
}),T_{\hat{C}}(M_{\hat{A}\hat{B}\hat{C}}),T_{\hat{C}}(N_{\hat{A}\hat{B}\hat
{C}})\geq0.
\end{equation}
The rest of the proof then proceeds as in
\eqref{eq:final-steps-swap-simp-1}--\eqref{eq:final-steps-swap-simp-last}. We can furthermore exploit the extra symmetry in \eqref{eq:swap-ch-add-symmetry-for-LOCC} to simplify the SDP, as stated in Proposition~\ref{prop:swap-sdp-simplify-BCQT}. The
proof that $e_{\operatorname{PPT}}^{F}(\mathcal{S}_{AB}^{d},\rho_{\hat{A}%
\hat{B}\hat{C}})$ can be calculated by the same semi-definite program follows
from reasoning similar to that given in
Appendix~\ref{app:ch-infid-simplify-swap-LOCC}.

\end{document}